\PassOptionsToPackage{unicode}{hyperref}
\PassOptionsToPackage{naturalnames}{hyperref}
\documentclass[letter,12pt]{article}

\usepackage{genericPreamble}
\usepackage{magneticTightBindingMacros}
\title{Tight-Binding Reduction and Topological Equivalence in Strong Magnetic Fields}
\author{Jacob Shapiro \\
	\footnotesize{Department of Physics}\\ \footnotesize{ Princeton University}\\ \\
	Michael I. Weinstein\\
	\footnotesize{Department of Applied Physics and Applied Mathematics and}\\ 
	\footnotesize{ Department of Mathematics}\\ \footnotesize{ Columbia University}}

\begin{document}
	
	\maketitle
	
	\begin{abstract}
	Topological insulators (TIs) are a class of materials which are insulating in their bulk form yet, 
	 upon introduction of an a boundary or edge, e.g. by  abruptly terminating the material, may exhibit spontaneous current along their boundary. This property is quantified by topological indices associated with either the bulk or the edge system. 
	 In the field of condensed matter physics,  tight binding (discrete) approximate models, parametrized by hopping coefficients, have been used successfully to capture the topological behavior of TIs in many settings. However, whether such tight binding models capture the same topological features as the underlying continuum models of quantum physics has been an open question. 
 
We resolve this question in the context of the 
	 archetypal example of topological behavior in materials, the integer quantum Hall effect. 
	We study a class of continuum Hamiltonians, $H^\lambda$, which govern electron motion in a two-dimensional crystal under the influence of a perpendicular magnetic field. No assumption is made on translation invariance of the crystal. 
	 	We prove, in the regime where both the magnetic field strength and depth of the crystal potential are sufficiently large, $\lambda\gg1$, that the low-lying energy spectrum and eigenstates (and corresponding large time dynamics) of $H^\lambda$  are well-described by a scale-free discrete Hamiltonian, $H^{\rm TB}$; we show norm resolvent convergence. The relevant topological index is the Hall conductivity, which is expressible as a Fredholm index. We prove that for large $\lambda$ the topological indices of $H^\lambda$ and $H^{\rm TB}$ agree.
		 This is proved separately for bulk and edge geometries. Our results justify the principle of using discrete models in the study of topological matter.
	
	

			\end{abstract}

\tableofcontents

	\section{Introduction}
	
The field of topological insulators (TI) was born in 1979 with the experimental discovery of the integer quantum Hall effect \cite{vonKlitzing_1980}. The topological nature of this phenomenon was derived from the underlying quantum mechanical description in \cite{Laughlin1981,TKNN_1982,Buettiker_1988_PhysRevB.38.9375,Halperin_1982_PhysRevB.25.2185,NTW_1985_PhysRevB.31.3372}. The breakthrough discovery of Kane and Mele in 2005 \cite{Kane_Mele_2005} (in a sense, a re-discovery of \cite{Frohlich_Studer_1993_RevModPhys.65.733}), that systems without magnetic field could also exhibit topological features, catalyzed the discovery of topological insulators \cite{Schnyder_Ryu_Furusaki_Ludwig_PhysRevB.78.195125,Schnyder_Ryu_Furusaki_Ludwig_1367-2630-12-6-065010}; see also \cite{Haldane_1988_PhysRevLett.61.2015} which was influential in this direction. A periodic table of topological insulators was then constructed by Kitaev \cite{Kitaev2009}, organized by space dimension and symmetry (time reversal, particle hole, or their composition--chiral symmetry), and patterned after the Bott periodicity of K-theory, has been very influential.

Briefly stated, what characterizes an insulating material as {\it topological} is: (1) a macroscopic quantity (\emph{e.g.} transverse conductance in two dimensions) whose values are quantized at pre-determined intervals, {\it i.e.} taking on only integer multiples of fundamental constants, and such that this quantity is stable under various deformations of the system, (2) this physical quantity may be computed in either bulk geometries (infinite space structure) or edge geometries (half-infinite geometries) and the two quantities agree, {\it i.e.}, the so-called \emph{bulk-edge correspondence} principle is satisfied, and (3) there is a suitable notion of the topological space of all insulating Hamiltonians, whose  associated set of path-connected components is isomorphic to a discrete set where the macroscopic quantities take values. This paper is directly related to the former two features. We focus on the archetypal example of the integer quantum Hall effect (IQHE). We believe that, for any given cell of the Kitaev table, our approach provides a general program for identifying  topological indices of the corresponding discrete and continuum models.

The IQHE describes the motion of electrons in two dimensions under the influence of a strong perpendicular magnetic field at low temperatures. 
Associated with such systems is a topological index called the Hall conductance. 
It is experimentally measured by applying an electric field along one axis of the two-dimensional sample and probing the resulting conductance induced along the perpendicular axis of the sample due to the presence of the magnetic field.

 In the field of condensed matter physics,  tight binding (discrete) approximate models, with numerically determined model (hopping) parameters, have been used successfully to capture the topological behavior of many materials \cite{Ashcroft_Mermin_1976}. The reason for this is the basic fact that electrons in a solid acquire relatively small momentum, i.e., an effective momentum UV cut-off is appropriate, which is equivalent to a discretization of space. However, whether such tight binding models capture the same topological features as the underlying continuum models of quantum physics has been an open question. 
 

\subsection{The continuum IQHE Hamiltonian}\label{IQHE-ham}
We work in the non-interacting electron approximation.
Introduce two parameters, $b>0$, the magnetic field strength and $\lambda$, the depth of the crystal potential. We model the crystal as a collection of disjoint atomic wells, whose locations are specified by a countable set $\GG\subseteq\RR^2$. Our Hamiltonian is $$ H^{\lambda,b} := \left(P- b Ax\right)^2+ \lambda^2 V(x)-e^\lambda_0\Id,\qquad P\equiv-\ii\nabla, $$ 
acting in the space $\contHilSp$, 
where the vector potential $\calA :=b\ A x$ (symmetric gauge, with $Ax \equiv \frac{1}{2}e_3\wedge x$) gives rise to a constant magnetic field which is  perpendicular to the plane;
	$B(x)=\nabla\wedge\calA  = b \hat{e}_3$.  The scalar potential  $V$ models a  {\it crystal} and is given by an infinite sum of 
	atomic potential wells $v_0$, each centered on a point in a discrete set of nuclei or atomic centers  $\GG\subseteq\RR^2$:
	\begin{equation} V(x)\ =\ \sum_{m\in\GG} v_0(x-m),\quad x=(x_1,x_2).\label{V0toV}\end{equation}
We note that $H^{\lambda,b}$ includes a shift by $e_0^\lambda$, the ground state energy of the atomic Hamiltonian, $h^{\lambda,b}:=\left(P- b Ax\right)^2+ \lambda^2 v_0(x)$.
	The simplest examples to keep in mind are $\discSp=(a\ZZ)^2$ for some lattice spacing $a>0$, or a truncation of $(a\ZZ)^2$ along an edge, $\discSp=a\ZZ\times a\NN$. The results of this article apply to any discrete set, $\GG$, with no accumulation point. There is   no requirement that   $\GG$ be translation invariant.
We study the regime where $\lambda$ and $b$ are comparable and large.

\begin{enumerate}
\item We first prove that the lowest part of the energy spectrum and the corresponding eigenstates of $H^{\lambda,b}$
can be approximated by a discrete {\it tight-binding} (effective) Hamiltonian, $\tbH$ acting on $\discHilSp$. 
A detailed  outline  of the tight-binding approximation is presented in \cref{TB-idea}. 
\item We then apply these approximation results to establish equality between topological indices associated with $H^\lambda$ in the strong binding ($\lambda$ large) regime and the indices associated with $\tbH$. 
%
\end{enumerate}
	
A review concerning the rigorous derivation of the tight binding approximation, from a semi-classical analysis perspective, for non-magnetic Hamiltonians ($b\equiv0$) via localized states "near atoms" in a general setting appears in  \cite{dimassi_sjostrand_1999}; see also \cite{Carlsson1990} and references cited therein. A strategy based on atomic orbitals was  used in \cite{FLW17_doi:10.1002/cpa.21735} to show convergence of the low-lying band structure of honeycomb (graphene-like) Hamiltonians to that of the 2-band tight binding model.
The magnetic semi-classical limit has been studied previously. Perturbation methods based on non-magnetic orbitals 
	\cite{Helffer_Sjostrand_1987_magnetic_ASNSP_1987_4_14_4_625_0}, or non-magnetic Wannier functions \cite{RevModPhys.63.91}  
	have been used to treat the case of weak magnetic fields, $b$ small; see also \cite{Panati2003,Ablowitz_Cole_2020_PhysRevA.101.023811}.
	
	We study the magnetic problem in the strong binding regime, corresponding to $\lambda$ sufficiently large. 
	Note however that for $b$ fixed and
	 $\wellDepth\to\infty$  (what one might attempt first as a tight-binding limit), the emergent tight-binding operator has a band structure which is topologically trivial; its associated topological indices vanish \cite{nakamura1990}.
	 We obtain asymptotic models with topological properties by taking $b\sim\lambda\to\infty$.
	Alternatively,  topological properties have also been obtained by by fixing $\wellDepth$ (\emph{not} the tight-binding limit) and either: taking $b\ll1$ (Peierls substitution) or $b\gg1$ (Landau Hamiltonian limit); see \cite{Bellissard:1987dy,Helffer_Sjostrand_1987_magnetic_ASNSP_1987_4_14_4_625_0,CORNEAN2017206}. These latter two limits generate a so-called "rotation algebra" which yields non-trivial topological features.
In all the above cases, the tight-binding models that emerge, such as the Harper model (\cref{Z2-tb}), give rise to exotic spectrum and non-trivial topology; a plot of its (fractal) spectrum vs. magnetic flux  is the well-known Hofstadter  butterfly. Its colored version encodes the different Chern numbers in the various energy regions \cite{doi:10.1063/1.1412464}.
	
The derivation of discrete approximations to a continuum Hamiltonian relies on projecting onto a subspace spanned by highly localized functions, which approximate the spectral subspace of interest.
The two most common bases are obtained from atomic orbitals or Wannier functions \cite{Ashcroft_Mermin_1976}.
Now, there is a well known topological obstruction to constructing an exponentially localized basis of functions that spans an isolated spectral region when its Chern number is non-zero, see e.g. \cite{Brouder_Panati_et_al_PhysRevLett.98.046402,Panati2007,Alex_2020_2003.06676,Ludewig_Thiang_2020_doi:10.1063/1.5143493}. The presence of this obstruction is not a problem, since it is  \emph{sub-bands} of the spectral region we capture by our tight binding model, that capture non-trivial topology.
	\subsection{Main results}\label{main-results}
	
		{\it We study the emergence of tight-binding models which capture non-trivial topology of the Hamiltonian
		$H^\lambda$, for  magnetic parameter $b=\lambda$ and $\lambda\gg1$:}
	\begin{equation} H^\lambda \equiv H^{\lambda,\lambda} = \left(P- \lambda Ax\right)^2+ \lambda^2 V(x)-e^\lambda_0\Id.\label{Hmag}\end{equation}
	A  scaling of this type was used in  \cite{MATSUMOTO1995168}, which considers the low-lying spectrum of systems with finitely many potential wells.
	
	Our main results concern spectral and topological properties associated with $H^\lambda$ and its discrete approximation $\tbH $, for large $\lambda$. 
	Since the Hamiltonian $H^\lambda$ is specified by a choice of atomic potential, $v_0$, as well as a discrete set of nuclei or atomic centers, $\GG$, we make assumptions on the microscopic model: the \emph{single atom} magnetic Hamiltonian 
	\[ h^\lambda=(P-\lambda Ax)^2+\wellDepth^2v_0(x)-e^\lambda_0\Id.\]
	Here, $e^\lambda_0$ is chosen so that the ground state energy of $h^\lambda$ is zero, $v_0$ is taken to be radial, compactly supported with radial ground state. 
	Furthermore, we assume that  $h^\lambda$ has a spectral gap between zero energy and the remainder of the spectrum, which is uniformly bounded away from zero, for all $\lambda$ large.
	The set  $\GG$ is not required to be translation invariant, but we assume it has no points of accumulation; see below.
	
	\subsubsection{Emergence of a tight binding model - \cref{res-conv}}
	The idea of the tight-binding approximation is that  $H^\lambda$, when restricted to the subspace associated with the lowest part of its spectrum, is well represented by $H^\lambda$ acting on an ``orbital subspace'', given by the closed linear span of appropriate translates of the atomic ground state, centered at points in the discrete set $\GG$. As $\lambda$ increases these orbitals become increasingly concentrated about the points of $\GG$, and an appropriately scaled limit of this reduced operator converges to a tight binding operator, $\tbH $ acting on $\discHilSp$. A more  detailed sketch of the construction of $\tbH $ appears in \cref{TB-idea}. 
	
	In \cref{res-conv} below we prove norm resolvent convergence of a scaling of $H^\lambda$ by \emph{the inter-well hopping coefficient} $\rho^{\lambda}$ to
	the tight binding Hamiltonian, $\tbH $:
	\[ \left( (\rho^{\lambda})^{-1}H^\lambda-z\Id) \right)
	\to J^*\left( H^{TB}-z\Id \right)^{-1}J,\qquad(z\in\CC\setminus\sigma(\tbH))
	\]
	where $\rho^\lambda$ gives the microscopic tunneling amplitude from one atom to its nearest neighbor atom, $J:L^2(\RR^2)\to l^2(\GG)$ is a partial isometry and the convergence is in the  $L^2(\RR^2)\to L^2(\RR^2)$ operator norm.
\cref{res-conv} implies that the low-lying spectrum  of $\crysHamil[e]$ is clustered around zero energy, and the corresponding long time dynamics are described  by the tight binding, scale free, Hamiltonian $\tbH $.

\subsubsection{Topological equivalence of continuum and discrete models: \cref{thm:topological equivalence for bulk geometries} and \cref{thm:edge continuum discrete correspondence} }\label{top-consq}
	
When considered in bulk (infinite space) geometries, in a specified range of energies, TIs are insulating, while in edge (half-infinite) geometries at those same energy ranges they are conducting. In two-dimensional systems which break time-reversal symmetry, e.g. magnetic systems, the material acts as a TI in energy regimes, where the topological index  associated with the bulk material, an integer-valued invariant, is non-zero; this number is equal to the  transverse Hall conductivity of the integer quantum Hall effect. For translation invariant (periodic) systems, this integer-valued invariant is the  {\it Chern number} of  a vector bundle associated with the Fermi spectral projection; see \cref{Hall-bulk}. It is given by an integral over the quasi-momentum torus (Brillouin zone) of the Berry curvature, which is expressible in terms of the bulk Floquet-Bloch modes. The corresponding edge topological invariant may be identified as a spectral flow within the same energy ranges \cite{Hasan_Kane_2010}.

We use \cref{res-conv} to prove, for $\lambda\gg1$,  the equality of  topological invariants associated with low-lying states of $H^\lambda$ and those of the tight binding Hamiltonian, $H^{TB}$:\\
	(a) \cref{thm:topological equivalence for bulk geometries}:\ Equality of the topological indices of isolated subsets of the spectra of $H^\lambda$ (within its lowest band) and of $\tbH $ when $\discSp$ corresponds to bulk geometries, e.g. $\GG=(a\ZZ)^2$, and\\
	(b) \cref{thm:edge continuum discrete correspondence}:\ Equality of the edge topological indices associated with $H^\lambda$ and $\tbH $ when $\discSp$ corresponds to edge geometries, e.g. $\GG=a\NN\times a\ZZ$.

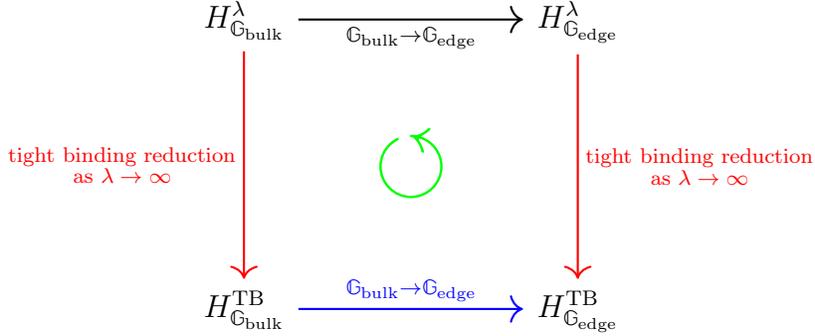
\begin{figure}\[\begin{tikzcd}[row sep=3cm,column sep=3cm,inner sep=1ex]
		H^\lambda_{{\discSp_\mathrm{bulk}}}  \arrow[thick,swap,red] {d}{\substack{\text{tight binding reduction}\\ \text{as $\lambda\to\infty$}}} \arrow[thick,swap]{r}[name=U]{\discSp_\mathrm{bulk}\to\discSp_\mathrm{edge}}
		&
		H^\lambda_{\discSp_\mathrm{edge}}
		\arrow[thick,red]{d}{\substack{\text{tight binding reduction}\\ \text{as $\lambda\to\infty$}}}
		\\
		H^{\mathrm{TB}}_{\discSp_\mathrm{bulk}}   \arrow[thick,blue]{r}[name=D]{\discSp_\mathrm{bulk}\to\discSp_\mathrm{edge}}   
		& H_{\discSp_\mathrm{edge}}^{\mathrm{TB}}
		\arrow[to path={(U) node[midway,scale=3,green] {$\circlearrowleft$}  (D)}]{}
	\end{tikzcd}\]
	\caption{Our crystal Hamiltonians depend on the set of atomic centers $\discSp$. A choice of bulk or edge Hamiltonian is made by choosing an appropriate $\discSp$; for example, $\discSp_\mathrm{bulk}=(a\ZZ)^2,\discSp_\mathrm{edge}=a\ZZ\times a\NN$. We establish the two vertical (red) arrows in this diagram. The bulk-edge correspondence for discrete systems, lower horizontal (blue) arrow,   is well-known \cite{SBKR_2000,Elbau_Graf_2002}. Our results imply the commutativity of this diagram, and hence an alternative proof of the bulk-edge correspondence for continuum systems. }
	
	\label{fig:commutative diagram}
\end{figure}

\cref{thm:topological equivalence for bulk geometries,thm:edge continuum discrete correspondence} justify, in the context of our continuum magnetic Hamiltonian, $H^\lambda$, the use of tight binding models to compute the topological properties of continuum systems which are a-priori more realistic for the description of electrons.

Through these results and \cite{SBKR_2000,Elbau_Graf_2002} we also recover the bulk-edge correspondence at the level of the continuum operators, \cref{cor:bec}, see \cref{fig:commutative diagram}.

We emphasize that our approach handles both translation invariant and non-translation invariant systems. 
Thus both the bulk and edge are treated by a unified approach. We do not make use of  Floquet-Bloch modes or Wannier functions 
	to obtain the limiting effective tight-binding Hamiltonian. Hence our approach could also be used to understand disorder effects in continuum models via their tight-binding counterparts, a program already pioneered in \cite{AIF_1995__45_1_265_0,Klopp_1993} 
	but which we anticipate may be  taken further by incorporating our methods.

\begin{rem}
The Hamiltonians we study, which model the IQHE, are in class A in two dimensions of the Kitaev table. It is natural to ask whether our analysis carries over to other cells of the Kitaev table. 

The first part of our analysis, which is the tight-binding reduction, is dimension agnostic (though the way in which the magnetic field enters into our Hamiltonian is particular to two dimensions; the magnetic case in other dimensions should be similarly handled). Moreover, in the non-magnetic case $\nnHopping^\lambda$ is not given by an oscillatory integral and can be bounded from below more easily. 

The second part of our analysis, the proof of topological equivalence, is a general argument in the following sense: in $d=2$, the very same Fredholm index formulas (see \cref{eq:Hall conductivity as Fredholm index,eq:edge Hall conductivity as a Fredholm index} below) hold for all other symmetry classes of the Kitaev table \cite{Grossmann2016} in both bulk and edge, possibly replacing the $\ZZ-$ valued index with its $\ZZ_2-$ valued analog \cite{Fonseca2020} for some of the real symmetry classes (associated with time-reversal and particle hole anti-unitary symmetries). In other even spatial dimensions and symmetry classes analogous formulas to \cref{eq:Hall conductivity as Fredholm index} hold. Hence our index-stability arguments (based on compactness and norm continuity) should in principle carry over to all even dimensions and all symmetry classes of the Kitaev table. The odd dimensions have slightly different formulas which we have not fully explored yet.

The problem, however, is that in order to obtain non-trivial models in other symmetry classes (for example, $d=2$ with odd time-reversal invariance) one needs either further than nearest-neighbor models or matrix models (e.g., internal degrees of freedom, which could arise for instance by analyzing
more than just the ground state $\vf_0^\lambda$ of $h^\lambda$). This is beyond the scope of the present tight-binding reduction scheme. It would, for example, be interesting to start with a realistic continuum Hamiltonian that takes the spin-orbit interaction into account and exhibit a tight-binding reduction that yields the Kane-Mele model \cite{Kane_Mele_2005}.  
\end{rem}

In the remainder of this introduction we discuss \cref{res-conv} on resolvent convergence, which is  used to prove 
\cref{thm:topological equivalence for bulk geometries,thm:edge continuum discrete correspondence}
on topological equivalence.

\subsection{Discussion of \cref{res-conv} on resolvent convergence to a tight binding model }\label{TB-idea} 
It is useful to first  discuss the idea of the tight binding approximation, beginning with the non-magnetic Hamiltonian
($\calA \equiv0$): 
$H^\lambda=-\Delta + \lambda^2 V-e^\lambda_0\Id=P^2+\lambda^2 V-e^\lambda_0\Id$.  
Our analytical approach is motivated by that taken in  \cite{FLW17_doi:10.1002/cpa.21735,FW:20}, with differences highlighted at the conclusion of this subsection.
Let $\varphi^\lambda_0$, $e_0^\lambda$ be the ground state eigenpair of $h$, i.e.,
$h^\lambda \vf^\lambda_0=0$,\ $\|\vf^\lambda_0\|_{L^2(\RR^2)}=1$.
The span of the set of translates $\{\varphi^\lambda_m\}_{m\in\GG}$, where $\varphi^\lambda_m(x)\equiv \varphi^\lambda_0(x-m)$,  is called the subspace of ground 
state  {\it atomic orbitals}, denoted $\orbitalSubSp$, and $\orbitalProj\equiv\Pi^\lambda$ denotes the orthogonal projection of $\contHilSp$ onto 
$\orbitalSubSp$. Since $(P^2+\lambda^2 v_0(x-m)-e_0^\lambda\Id)\varphi_m=0$ and $v_0(x-n)\vf_m\approx0$ for $n\neq m$, it follows that $\orbitalSubSp$ is an 
approximate eigenspace of the crystal Hamiltonian $H^\lambda$ for $\lambda\gg1$. Moreover, for $\lambda\gg1$ the  atomic orbitals are spatially concentrated about the discrete sites of $\GG$, and the set of atomic orbitals is a nearly orthognormal set.
The weak coupling across atomic wells implies that $H^\lambda$, projected onto this subspace has energy-spectrum which concentrates 
near energy zero as $\lambda$ increases.

The width of the spread of energies in the spectrum about zero is related to the tunneling probability between nearest neighbor wells. This can be approximated  by the {\it  hopping coefficient}, $\rho^\lambda$,  an overlap integral  involving the atomic well $\lambda^2 v_0$, its atomic ground state, $\varphi^\lambda_0$, and the  atomic ground state spatially translated to a nearest neighbor site at $d$, $\varphi^\lambda_0(x-d)$. 
In the non-magnetic case, the  quantity $\rho^\lambda$ is known to satisfy exponentially small upper and lower bounds \cite{FLW17_doi:10.1002/cpa.21735}: 
\begin{equation}
	e^{-c_2\lambda}\lesssim\rho^\lambda\lesssim e^{-c_1\lambda},\quad 0<c_1<c_2.\label{rho-bds}
\end{equation}

Since the low energy spectrum of  $H^\lambda$ concentrates about zero with an energy width of about $\rho^\lambda$, we divide by the hopping coefficient:
\[ \tilde{H}^\lambda \equiv \left( -\Delta + \lambda^2 V - e_0^\lambda\Id \right)/\rho^\lambda.\]
For $\lambda\gg1$, the low energy spectrum of $\tilde{H}^\lambda$ is centered about zero energy and is spread across an energy range of order one. 

To obtain a low energy discrete approximation to $H^\lambda$, consider its restriction  to $\orbitalSubSp:={\rm Range}(\Pi)$:\  $\left.\Pi\tilde{H}^\lambda\Pi\right|_{\orbitalSubSp}$. Since $\{\varphi_m\}_{m\in\GG}$ is a nearly orthonormal basis for $\orbitalSubSp$, we have that $\Pi\tilde{H}^\lambda\Pi$ is equivalent (via a partial isometry) to an  operator  $\Big[\Pi\tilde{H}^\lambda\Pi\Big]$ acting on $\discHilSp$, whose matrix elements:  
\[  \Big[\Pi\tilde{H}^\lambda\Pi\Big]_{m,n\in\GG}\quad\textrm{are well-approximated by}\quad \Big\{\ \left\langle \varphi_m^\lambda, \tilde{H}^\lambda\varphi_n^\lambda\right\rangle\ \Big\}_{m,n\in\GG}\ .\]
The core of the proof of resolvent convergence to a tight binding Hamiltonian is to show that: 

\[ 
\textrm{for all $z\notin\sigma(\tbH )$, }\quad  \Big\|\ \left(\ \Big[\Pi\tilde{H}^\lambda\Pi\Big]\ -\ z\Id\ \right)^{-1} -\ (\tbH -z\Id)^{-1}\ \Big\|_{\mathcal{B}(\discHilSp)}\to0,\]
as $\lambda\to\infty$. 
The spectrum of $H^{TB}$ is situated in an order one neighborhood of zero (indeed $\tbH$ is defined to be independent of the asymptotic parameter $\wellDepth$). 
The low-lying spectrum of $H^\lambda$ is approximately  given by the set:
$\rho^\lambda \sigma(\tbH ),\ \lambda\gg1$.

An important ingredient of a rigorous treatment is a sufficiently good lower bound on the hopping coefficient $\rho^\lambda$. 
An expansion of $\rho^\lambda$ in the non-magnetic case, based on precision  semi-classical asymptotic analysis, 
for general smooth $v_0$ with the assumption that $v_0$ has a non-degenerate minimum, was obtained in \cite{Helffer_Sjostrand_84_doi:10.1080/03605308408820335}. 
In the work of \cite{FLW17_doi:10.1002/cpa.21735}  upper and lower bounds for $\rho^\lambda$ were obtained without the assumption of a non-degenerate minimum. All these results depend crucially 
on $\varphi^\lambda_0$, the ground state eigenfunction of the atomic Hamiltonian $h^\lambda=-\Delta+\lambda^2 v_0-e_0\Id$, being real-valued and  of definite sign.

{\it Our goal in this paper is to treat the magnetic Hamiltonian \cref{Hmag}, in the simultaneous limit of deep wells and strong magnetic field by a parallel strategy.}
An important issue that arises is that $\rho^\lambda$ is now an oscillatory integral, since the integrand now involves the oscillatory phase of the {\it magnetically translated} ground state $\varphi_0^\lambda$ (see \cref{mag-transl}). Thus, for $b=\lambda$ large
(large magnetic fields), $\rho^\lambda$  cannot be understood by perturbative arguments about a non-magnetic basis. 

This is where the assumptions, mentioned above, that the atomic potential, $v_0$, is radial and the ground state of $h^\wellDepth$ is also radial, enter. To control the convergence of $\tilde{H}^\lambda$ in the magnetic case, we make use of a recent result \cite{FSW_2020_magnetic_double_well} on a lower bound on the magnetic hopping coefficient, $\rho^\lambda$, which uses these assumptions on $v_0$.

Note that  while  the general strategy used  in \cite{FLW17_doi:10.1002/cpa.21735} for bulk periodic structures and in \cite{FW:20}
for truncated periodic structures with an edge is followed, the setting of this article is much more general;  
\cref{res-conv} concerns norm resolvent convergence to tight binding models for a class of continuum operators admitting crystal potentials, $V(x)$ (see \eqref{V0toV}), \emph{without} any translation invariance. Furthermore, our assumptions on the symmetry of $v_0$ are only necessary to control $\nnHopping^\lambda$ via \cite{FSW_2020_magnetic_double_well} and are not used in the actual tight-binding scheme.
\begin{rem}
We note that the present tight-binding reduction procedure handles just as well the non-magnetic case, and is indeed not special to two-dimensions. All that is needed are appropriate decay estimates, e.g. Gaussian or exponential Agmon type \cite{agmon1985schrödinger},
on the ground state $\vf_0^\lambda$ (\cref{thm:Gaussian decay of ground state}) and upper and lower bounds on the hopping coefficient, $\nnHopping^\lambda$ \cref{rho-bds}. This has already been established for the non-magnetic case in two-dimensions in \cite{FLW17_doi:10.1002/cpa.21735}. For dimensions three or higher the estimates on $\nnHopping^\lambda$ can also be proved.
\end{rem}
\begin{rem}
We remark on the contrast between the strategy 
of this paper with that of \cite{FLW17_doi:10.1002/cpa.21735} for bulk honeycomb structures and  for structures obtained from the bulk honeycomb by sharp termination
 \cite{FW:20}. The focus of the present article is a proof of resolvent convergence for general structures (not necessarily translation invariant) and applications of this convergence result to topological equivalence. In \cite{FLW17_doi:10.1002/cpa.21735} uniform convergence of the low-lying bulk dispersion surfaces to those of
 the tight binding model is proved, and in \cite{FW:20} 
 edge states and their expansions are constructed along 
 zigzag terminations of the bulk.
Norm resolvent convergence is proved as a consequence of the detailed estimates these papers 
({\it e.g.} pointwise estimates on the resolvent kernel in
 \cite{FW:20}), which establish uniform convergence of the matrix elements of the Schur / Lyapunov-Schmidt reduction of $\tilde{H}^\lambda$ on the full Hilbert space $L^2(\RR^2)$ to the orbital subspace $\orbitalSubSp$. 
In contrast, in this paper we avoid  estimation of resolvent integral kernel and only require energy estimates. Furthermore, we use the Gramian operator (see \cref{sec:gramian1} below) to systematize estimation of the overlap between distinct atomic orbitals.

This results in a proof of resolvent convergence which is simpler and shorter. Although not yielding as detailed information (necessary, for example, to construct continuum edge states from those of $H^{\rm TB}$), this approach suffices to obtain the full tight-binding reduction and our results on topological equivalence. Indeed, resolvent convergence
is sufficient due to the stability properties of the Fredholm index.
 
 An alternative route to the equality of topological indices for continuum models and their discrete tight-binding limits would be through the Kubo trace formula \cref{eq:kubo formula} or the trace formula for the edge \cref{eq:edge hall conductance trace formula}. This strategy would require  more local information to  control the resolvent's \emph{integral kernel} as in \cite{FW:20}. The strategy of the present paper gets by with energy estimates alone, since (local) trace formulas \cref{eq:kubo formula,eq:edge hall conductance trace formula} have been connected with (global) index formulas \cref{eq:Hall conductivity as Fredholm index,eq:edge Hall conductivity as a Fredholm index}; see \cite{Bellissard_1994,SBKR_2000}.
 \end{rem}

\subsection{Outline and discussion}\label{outline}

Our main results are obtained via the following three steps:
	\begin{enumerate}
	    \item Using PDE techniques, we localize the operator $H^\lambda$ to an {\it orbital subspace} of $L^2(\RR^2)$, which is an exponentially accurate (in $\lambda$) approximation of the low-lying spectral subspace of $H^\lambda$.
	    \item Using techniques of classical and functional analysis we prove resolvent convergence of a scaling of $H^\lambda$, by an emergent hopping parameter, $\rho^\lambda$, to a discrete operator, $H^{\rm TB}$, as $\lambda$ and $b\sim\lambda\to\infty$.
	    \item Through a succession of deformations of Fredholm operators,
	  which preserve the Fredholm index (via its norm continuity and stability under compact perturbations) we prove that topological indices of continuum systems equal their discrete counterparts.
	    \end{enumerate}
	    
	    The paper is structured as follows:
\cref{setup} contains a precise mathematical framework and \cref{res-conv-thm} contains the statements of our results. In \cref{sec:proof of main theorem} we embark on the proof of \cref{res-conv} on resolvent convergence. In \cref{energy-est} we prove the key energy estimate for $H^\lambda$, when projected on the  on orbital space $\orbitalSubSp$. This enables reduction of the low energy spectral problem for $H^\lambda$ acting in $L^2(\RR^2)$ to discrete operator
acting in $l^2(\GG)$.  In \cref{sec:topological equivalence} we provide the proofs of \cref{thm:topological equivalence for bulk geometries,thm:edge continuum discrete correspondence}
on equality of continuum and discrete topological indices. \cref{sec:gramian}, \cref{subsec:proofs of propositions of the main theorem} and \cref{sec:decay-props} contain technical estimates stated earlier in the paper. \cref{topology} provides background on the field of  topological insulators in the context of this article.
\subsubsection{Further directions}
We discuss various directions in which this work may be extended:
\begin{enumerate}
	\item Instead of working with the lowest, ground state, of the single atom Hamiltonian $h^\wellDepth$, one may consider the lowest $N$  eigenstates and obtain  a matrix-model acting on $\lt(\discSp)\otimes \CC^N$.
	\item Derive discrete models with couplings that involve beyond nearest-neighbor hopping terms. An example is the Haldane model \cite{Haldane_1988_PhysRevLett.61.2015}.
	\item Explore the situation where due to disorder $\crysHamil^\lambda$ has no spectral gap separating the first band from the second band, but rather merely a {\it mobility gap} \cite{EGS_2005}. The same goes for the assumption on the isolated subset made in \cref{thm:topological equivalence for bulk geometries} to obtain the spectral projections onto sub-bands of the first band for the calculation of the topological index. We note that these two gaps are of a different nature mathematically though they probably both stem from strong disorder, and that it is mostly likely easier to control the latter one.
	\item Extend \cref{res-conv} to deal with \cref{example:random hopping model,rdm}.
	\item Extract more information about Anderson localization for continuum random operators from the more extensively studied  discrete models, in the direction of \cite{Klopp_1993}.
	\item Given a discrete Hamiltonian on $\GG$ with specified  edge geometry and boundary conditions, can it be realized as the tight binding limit of a continuum Hamiltonian?
	\item What is the relationship with quantum graph models (e.g. \cite{Becker_et_al_2008.06329,Kuchment2007,berkolaiko2013introduction,Fisher_Li_Shipman_2005.13764}) and do they arise as tight-binding limits, for example, for potentials consisting of a tube-like network?
\end{enumerate}

\noindent\textbf{Acknowledgements:}
The authors thank Charles L. Fefferman, and Gian Michele Graf for stimulating discussions which have had an impact on this paper.
M.I.W. was supported in part by National Science Foundation grants DMS-1412560, DMS-1620418 and DMS-1908657 as well as by the Simons Foundation Math + X Investigator Award \#376319. 
J.S. acknowledges support by the Swiss National Science Foundation (grant number P2EZP2\_184228), as well as support from the Columbia University Mathematics Department and Simons Foundation Award \#376319, while a postdoctoral fellow
during 2018-2019. 

\subsection{Notation, conventions and a lemma}
Throughout, we use the following more or less standard mathematical abbreviations:
\begin{itemize}
	\item For $r>0$ and $x\in\contSp$, $B_r(x)\equiv\Set{y\in\contSp|\normEuc{x-y}<r}$.
	\item $|A|^2\equiv A^\ast A$.
	\item $H^2$ is the Sobolev space of functions which have two weak derivatives in $L^2$.
	\item For any set $S$, $\chi_S$ is the characteristic function equal to $1$ for arguments in $S$.
	\item For the resolvent of an operator $B$ at spectral parameter $z\in\CC$ we write\\ $R_B(z)\equiv(B-z\Id)^{-1}$.
	\item We use the norms $\norm{A}_{\infty,1} := \sup_n \sum_m |A_{nm}|$ and $\norm{A}_{1,\infty} := \sup_m \sum_n |A_{nm}|$.
\end{itemize}
as well as specialized notation for our model
\begin{itemize}
	\item $\wellDepth>0$ is the tight-binding strength parameter (equal to the magnetic field strength), which is fixed, finite, but arbitrarily large.
	\item The continuum Hilbert space is $\contHilSp$.
	\item On $\contHilSp$, we consider $\posOp\equiv(\posOp_1,\posOp_2)$ the position (multiplication) operator and $\momOp\equiv-\ii\nabla$, the momentum operator. Sometimes, when no ambiguity arises, $x$ shall also denote the position operator.
	\item $\oneAtomPotOp_0$ denotes the atomic potential well.
	\item $\crysPotOp$ is the full-crystal potential, a sum of infinitely many translates of $\oneAtomPotOp_0$;  see \cref{V-def}.
	\item $\supp(\oneAtomPotOp_0)\subset B_\oneAtomPotSuppRad(0)$, where $\oneAtomPotSuppRad<2 \minLatticeSpacing$; see \cref{agt2r_0}.
	\item $\oneAtomHamil^\lambda=(P-b Ax)^2+\oneAtomPotOp_0(x)-e_0\Id$ is the single-atom magnetic Hamiltonian
	\item  $\crysHamil^\lambda$ is the crystal Hamiltonian
	\item $\discSp$ is a countable subset of $\contSp$ such that 
	$$ \minLatticeSpacing := \inf_{\substack{u,u^\prime \in\discSp\\ u\ne u\prime } }\ \normEuc{u-u^\prime} > 0,$$
	where $\normEuc{\cdot}$ is the Euclidean norm on $\contSp$; see \cref{GG-set}.
	\item The relation $u\sim_\discSp u^\prime $, or simply $u\sim u^\prime$, means that $u$ and $u^\prime$ are nearest-neighbors in $\discSp$, {\it i.e.} $u\sim_\discSp u^\prime$ iff 
	$$\normEuc{u-u^\prime} = \inf_{ v\in\setwo{\discSp}{u} }\ \normEuc{u-v}$$
	\item The target discrete Hilbert space is $\discHilSp$.
	\item $\contToDiscHilSp:\contHilSp\to\discHilSp$ is the partial isometry that takes us to the discrete space.
	\item For any $n\in\discSp$, $\magTransl{n}$ is the magnetic translation by $n$; see \cref{mag-transl}.
	\item $\nnHopping^\lambda$ is the nearest-neighbor hopping amplitude.
	\item For any operator $B$ on $\contHilSp$, $\orbitalProjBLO B := \orbitalProj B \orbitalProj$.
\end{itemize}

For $x,y,z\in\contSp$ we shall repeatedly use
\begin{align} \normEuc{x-y}^2+\normEuc{x-z}^2 	&= 2\Big|x-\frac{y+z}{2}\Big|^2+\frac{1}{2}\normEuc{y-z}^2 \,.
	\label{eq:sumsq}\end{align}

\subsubsection{Holmgren-Schur-Young type bound}
We shall often make use of the following general bound:
\begin{lem} Let $\calH$ be a separable Hilbert space and $A$ a bounded linear operator on $\calH$. Then
	$$ \norm{A} \leq \max\left(\left\{\sup_n \sum_m |A_{nm}|, \sup_m \sum_n |A_{nm}|\right\}\right)
	= \max\left( \|A\|_{\infty,1}, \|A\|_{1,\infty} \right), $$
	where $A_{nm} \equiv \ip{u_n}{A u_m}$ and $\Set{u_n}_n$ is any ONB for $\calH$. We call the first element in the above set the $(\infty,1)-$ norm of $A$ and the second one the $(1,\infty)-$ norm of $A$.
	\label{lem:hsy}\end{lem}

\section{Mathematical framework and formulation of main results}\label{setup}

We start by introducing the Hamiltonian for a single atom in a constant magnetic field. We then build 
up a macroscopic model of a crystal (periodic or aperiodic) comprised of infinitely many such atoms,
and then formulate our main results.

\subsection{Hamiltonian of a single atom in a magnetic field}\label{1atom}

The Hamiltonian for an electron bound to a single atom in  the plane $\RR^2$ immersed in a constant magnetic field is given by
\begin{equation}
	\oneAtomHamil[\wellDepth,\magFieldStrength]:= (P-bAx)^2 + \lambda^2v_0(x) - e_0^{\lambda,b}\Id,\quad  
	P\equiv\frac{1}{i}\nabla_x\ ,
	\label{Hatom}  \end{equation} 
where $\oneAtomHamil[\wellDepth,\magFieldStrength]$ is a recentering about the ground state energy, $e_0^{\lambda,b}$, of 
$ (P-bAx)^2 + \lambda^2v_0(x)$, and 
\begin{equation}
	Ax\ =\  \frac{1}{2} \ e_3 \wedge x=\frac12\ (x_2,-x_1) ,\quad B(x)\ =\ {\rm curl}\ b Ax\ =\ b\ e_3.\label{A-def}\end{equation} 
is the vector potential for a  \emph{constant} magnetic field in the direction of the $3$-axis. The particular choice  of vector potential \cref{A-def} is known as the symmetric gauge. 

In this paper we make an assumption on the relation of magnetic field strength and the depth of the atomic potential well:
\begin{assumption}
	
	\label{ass:b=lam}
	The magnetic parameter $b$ scales like the coupling $\lambda$,\ $b=\lambda$.
\end{assumption} 
Henceforth, we set $e_0^{\lambda}=e_0^{\lambda,\lambda}$ and 
\[ h^\lambda= h^{\lambda,\lambda}=(P-\lambda Ax)^2 + \lambda^2v_0(x) - e_0^\lambda\Id\,.\]
\begin{rem}
	In this paper, $\lambda$ is a dimensionless asymptotic parameter, and we use the natural units \cite{guidry1991gauge}
	$\hbar = c = 2m_e = q_e = 1$ ($m_e,q_e$ being the mass and respectively the charge of the electron, $c$ the speed of light). Hence, the equation $b= \lambda$, valid only in natural units, is equivalent to the universal $b = 4\frac{m_e^2 c^2}{q_e \hbar} \lambda$ and similarly throughout.
\end{rem}

Let  $\vf^{\lambda}_0$ denote the $L^2-$ normalized ground state
of the \emph{single} well magnetic  Hamiltonian $ \oneAtomHamil[\wellDepth,\magFieldStrength]$:
\[ h^\lambda \vf^{\lambda}_0=0
,\quad \|\vf^{\lambda}_0\|_{L^2(\RR^2)}=1\ .\] Sometimes, when no ambiguity arises, we may also use the symbols $(\vf^\lambda,e^\lambda)$ or, simply $(\vf,e)$.

We make the following assumptions on the atomic potential well, $v_0 \leq 0$ and the low-lying eigenvalues of $h^\lambda$:
\begin{assumption}
	
	\label{ass:v0}
	\begin{enumerate}
		\item[(v1)] $v_0$ is a non-positive, bounded and radial function: $v_0(x)=v_0(r)$, $r=|x|$, which attains a strictly negative minimum $v_\mathrm{min} < 0$. We also assume that $v_0(x)$ is compactly supported with:
		\begin{equation}\supp (v_0) \subset B_{r_0}(0)\ .\label{supv0}\end{equation}
		The radius $r_0$ will be chosen below to be sufficiently small relative to the minimal distance between atomic centers of the crystal; see 
		\cref{agt2r_0}.
		\item[(v2)] $\vf^{\lambda}_0$, the ground state of $h^\lambda$, is a \emph{radial} function.
		\item[(v3)] For $b=\lambda$, the ground state energy satisfies 
		\begin{equation}
			e_0^{\lambda}\le -c_{\rm gs}\lambda^2
			\label{C-gs}\end{equation}
		for some constant $c_{\rm gs}$ which is independent of $\lambda$.
		\item[(v3)] For $b=\lambda$, the distance between $e_0^\lambda$ and the remainder of the spectrum of $(P-\lambda Ax)^2 + \lambda^2v_0(x)$ is bounded away from zero, uniformly in $\lambda$ large, {\it i.e.} there exist positive constants $\lambda_\star$ and $c_{\rm gap}$ for all $\lambda>\lambda_\star$ sufficiently large, if $e_1^{\lambda}$ is the energy for the first excited state then \begin{align} e_1^{\lambda}-e_0^{\lambda}\ge c_{\rm gap}>0\,. \label{eq:spectral gap for atomic Hamiltonian}\end{align}
	\end{enumerate}
\end{assumption}
\begin{rem}\label{rem:non-decr} 
We believe that (v2) and (v3) hold for a large class of radial potentials, $v_0(r)$ and $\lambda$ sufficiently large, {\it e.g.} $v_0(r)$ is sufficiently smooth with a non-degenerate minimum at $r=0$. One could also come up with examples of radial potentials for which this condition is violated, e.g. a potential shaped like a Mexican hat. The problem of obtaining analogs of our results for nonradial $v_0$ remains open.
\end{rem}

The following Gaussian decay bound on $\oneAtomGrndStt[e]$ was proved in \cite[Theorem 2.3]{FSW_2020_magnetic_double_well}.
In a more general setting, a less precise upper bound has been proved \cite{Nakamura_doi:10.1080/03605309608821214,Erdos1996} .
\begin{thm}\label{gs-bound} Under \cref{ass:b=lam} and \cref{ass:v0}, the ground state of $h^\lambda$, $\oneAtomGrndStt[e]_0$, satisfies the following bound. There exist constants $\lambda_\star>0$ and $C>0$ such that for all  $\wellDepth>\wellDepth_\star$, 
	\begin{align}
		|\oneAtomGrndStt[e]_0(x)| \leq C\  \sqrt{\wellDepth}\ \ee^{-\frac14\wellDepth\left(\normEuc{x}^2-\oneAtomPotSuppRad^2\right)},\quad \textrm{for all}\quad  |x|\ge\oneAtomPotSuppRad,
	\end{align}
	\label{thm:Gaussian decay of ground state}
	where $r_0$ denotes the radius of $\supp(v_0)$.
\end{thm}

\subsection{Magnetic translations}\label{mag-transl}

Introduce the family of translation operators on $\RR^2$: 
\[ R^z=\exp(-\ii \momOp \cdot z)\qquad(z\in\contSp)\,.\]

It turns out that the magnetic momentum operator, $P-b Ax$, does not commute with $\contSp$-translations, due to the magnetic term, $-b Ax$. 
To deal with this Zak \cite{Zak_Magnetic_Transl_1964} introduced the notion of {\it magnetic translations}:
\begin{defn}[Magnetic translations] For $z\in\contSp$, {\it magnetic translation by $z$} is given by the operator:
	\begin{equation}
		\magTransl{z} := \exp(\ii X \cdot b \magVecPot z)\exp(-\ii \momOp \cdot z)\ =\ \exp(\ii X \cdot b \magVecPot z)\ R^z.
		\label{mag-trans}
	\end{equation}
\end{defn}
Note that       for $z\in\RR^2$, $ \magTransl{z}$ is a unitary operator and 
\[ \big[\magTransl{z},P-b A x\big]\ =\ 0.\]
\begin{rem} Unfortunately, the mapping $z\mapsto \magTransl{z}$ is not a group morphism:
	\[ \magTransl{z}\magTransl{y} = \exp(\ii \magFieldStrength(y\cdot A z-z\cdot A y)) \magTransl{y}\magTransl{z}\,.\] 
	We note that if $b(=\wellDepth)\in2\pi\QQ$ then restricting to a subgroup of translations of $\ZZ^2$ one may recover an honest group representation and eventually apply Bloch's theorem (i.e. perform a Fourier series to $L^2(\TT^2)$). Indeed, in that case, $b=2\pi p/q$ for some $p,q\in\NN$ and then
	$$ \magTransl{n}\magTransl{m} = \exp(2\pi\ii \frac pq\normEuc{n}\normEuc{m})\ \magTransl{m}\magTransl{n} $$ 
	and so if $\normEuc{n}\normEuc{m}\in q\NN$ the operators commute. This may be arranged if we always make horizontal translations which are multiples of $q$, i.e., we enlarge the unit cell to the so-called \emph{magnetic unit cell}. When all the magnetic translations commute with each other and the Hamiltonian, Bloch's theorem may be applied as usual; see e.g. \cite{KOHMOTO1985343}.
\end{rem}

\subsection{The crystal Hamiltonian}\label{GandH}

\subsubsection{$\GG$, the discrete set of atomic centers}\label{GG-set}
Our model of crystal is an infinite collection of atoms, described by identical potential wells, each centered on a point of a discrete set $\GG\subset\RR^2$, which we take, without any loss of generality, to contain $0$.
\begin{assumption}\label{ass:min-sp}
	Minimal lattice spacing of points in $\GG$: 
		\begin{equation}
			\minLatticeSpacing := \inf_{\substack{u,u^\prime\in\discSp\\ u\ne u^\prime}}\normEuc{u-u'} >0.
			\label{min-sp}\end{equation}
\end{assumption}

In order to ensure that contributions from the sums over terms, which account for all beyond nearest-neighbor interactions, when divided by $\rho^\lambda$, tend to zero as $\lambda\to\infty$, we make the additional 
\begin{assumption}\label{ass:min-NNN}
		The next-to-nearest-neighbor distances do not converge to $a$ throughout the infinite lattice:
		\begin{align}
			b := \inf_{u,u'\in\discSp:\normEuc{u-u'}>a}\normEuc{u-u'} > a
			\label{min-NNN}
		\end{align}
\end{assumption}

To study the large $\lambda$ behavior of $H^\lambda$, we require the following lemma on the summability of Gaussians centered on $\GG$, which is a consequence of \cref{min-sp}:
\begin{lem}\label{gauss-sum}
    There a constants $C>0$ such that for all $\lambda\ge1$:
		\begin{align}
    	\sup_{x\in\contSp}\sum_{m\in\discSp}\exp(-\lambda \normEuc{x-m}^2) \le C\,.
			\label{eq:Gaussians are summable in discrete space}
		\end{align}
\end{lem}

\begin{proof}

   Express the sum in \cref{eq:Gaussians are summable in discrete space} as limit of partial sums: $$ \sum_{m\in\discSp}\exp(-\lambda \normEuc{x-m}^2) = \lim_{R\to\infty}\sum_{m\in\discSp\cap B_R(x)}\exp(-\lambda \normEuc{x-m}^2). $$ Due to the minimal pairwise distance condition  \cref{min-sp}, each partial sum is a finite sum. Now consider the highest density lattice arrangement of circles of radius $a/2$ in $\RR^2$; this is the hexagonal packing of circles, for which each circle is surrounded by six other circles \cite{Fejes1942}.
  The set of all centers of this circle-packing is an equilateral triangular lattice, which we denote by $\Lambda_{a/2}$.

  We have the upper bound:
  $$ \sum_{m\in\discSp\cap B_R(x)}\exp(-\lambda \normEuc{x-m}^2) \leq \sum_{m\in\Lambda_{a/2}\ \cap B_R(x)}\exp(-\lambda \normEuc{x-m}^2). $$ 
 and passing to the limit  gives:
 $$ \sum_{m\in\discSp}\exp(-\lambda \normEuc{x-m}^2) \leq \sum_{m\in\Lambda_{a/2}}\exp(-\lambda \normEuc{x-m}^2). $$
We next evaluate the sum over  the lattice $\Lambda_{a/2}$ using the Poisson summation formula:
 \[
 \sum_{m\in\Lambda_{a/2}}f(m)= |\Omega_a|^{-1}\sum_{m\in \Lambda_{a/2}^*}\hat{f}(m)
 ,\] where $\hat{f}$ denotes the Fourier transform of $f$ and $|\Omega_{a/2}|$ the area of the fundamental cell. Applying this with $f(y)=\exp\left(-\lambda |x-y|^2\right)$ yields, for some constant $C_a$:
		\begin{align*}
			\sum_{m\in\Lambda_{a/2}} \ee^{-\lambda \normEuc{x-m}^2} = 
			\frac{C_a}{\lambda}
			\sum_{m\in(\Lambda_{a/2})^*} \ee^{-2\pi i m\cdot x}\ee^{-\pi^2 |m|^2/\lambda}.
		\end{align*}
		Therefore, for any $x\in\RR^2$
		\begin{align*}
			\sum_{m\in\Lambda_{a/2}} \ee^{-\lambda \normEuc{x-m}^2} \le 
			\frac{C_a}{\lambda}
			\sum_{m\in(\Lambda_{a/2})^*} \ee^{-\pi^2 |m|^2/\lambda}
			\lesssim \frac{1}{\lambda} \int_{\RR^2} \ee^{-\pi^2 |z|^2/\lambda} dz\le C.
		\end{align*}
		This completes the proof of \cref{gauss-sum}.
\end{proof}

We present some examples to which our main results apply. 
\begin{example}[$\GG$ translation invariant]  The basic examples to which the condition of \cref{ass:min-sp} 
and \cref{ass:min-NNN} apply which we have in mind are Bravais lattices such as: $\discSp := a\ZZ\times a\ZZ$ and the equilateral triangular lattice in $\RR^2$, and multi-lattices such as the honeycomb lattice, the union of interpenetrating triangular lattices and the Lieb-lattice, the union of three interpenetrating $\ZZ^2$ lattices. 
	These are translation invariant  \emph{bulk} lattices.
	\cref{ass:min-sp} also applies to discrete sets, $\GG$, which are invariant under a \emph{subgroup} of translations, which arise in bulk systems which are sharply terminated, {\it e.g.} $\discSp := a \ZZ\times a\NN$; these truncations model edge systems.
\end{example}

We believe it is possible to relax our assumptions on the discrete set $\discSp$ and allow for $\discSp$ to depend on $\wellDepth$, as in \cref{rem-Glam} right below, though we have not formulated \cref{res-conv} in this generality. 
In this case, $\minLatticeSpacing$, the minimal nearest-neighbor distance would depend on $\wellDepth$, and we would require that bounds on it hold uniformly in all  
$\wellDepth$ sufficiently large.

\begin{example}[$\discSp$ depending on $\lambda$]\label{rem-Glam}
	A subset $\mathbb{D}\subseteq\contSp$ is an \emph{(r,R)-Delone set} \cite{Germinet_2015} iff (i) it is uniformly discrete, i.e., there exists $r>0$ such that $|\mathbb{D}\cap B_r(x)|\leq 1$ for all $x\in\contSp$, and (ii) it is relatively dense, i.e., there exists a real $R\geq r$ such that $|\mathbb{D}\cap B_R(x)|\geq 1$ for all $x\in\contSp$. 
	Hence the minimal distance between any two points in $\mathbb{D}$ is $r$ and given any point $x\in\mathbb{D}$, there must be another point at most $\sqrt{2} R$ apart. $\ZZ^2$ is an example of a Delone set.      
If $\discSp$ is an $\left(a,a\frac{1}{\sqrt{2}}\sqrt{1+\frac{1}{\wellDepth}\alpha}\right)-$ Delone set, we 
	can, for $\lambda\to\infty$, realize a disordered tight-binding model, whose matrix elements are themselves disordered, varying between $1$ and $\exp(-c\alpha)$; see \cref{example:random hopping model} below.

\end{example}

\subsubsection{The crystal (magnetic) Hamiltonian, $H^\lambda$}\label{crystl-def}

To define the crystal Hamiltonian we introduce the crystal potential, $V$,  as a sum over $\GG$ of translates of the atomic potential, $v_0$:
\begin{equation} V(x)\ :=\ \sum_{n\in\discSp} v_n(x),\quad \textrm{where}\ \ v_n(x)=v_0(x-n)=(R^nv_0)(x)\,.\label{V-def}\end{equation}
The (centered) crystal Hamiltonian  is given by
\begin{equation}
	\crysHamil[e] :=  (P-b Ax)^2 + \lambda^2 V(x) - e_0^\lambda\Id\ ,\ b=\lambda;
	\label{Hlam}\end{equation}
	see also  \cref{Hmag}.
Recall that we have included a centering about the atomic ground state energy, $e_0^\lambda$, of $(P-b Ax)^2 + \lambda^2 v_0(x)$. Hence, for $\lambda\gg1$, 
the lowest part of the spectrum of $\crysHamil[e]$ will be near zero energy. 
We make the simplifying assumption of $V$ (see also \cite{FLW17_doi:10.1002/cpa.21735}) that the translates of the atomic potentials in \cref{Hlam} have disjoint support:
\begin{assumption}\label{r0a} The minimal spacing within $\GG$
	\cref{min-sp}, $a$,  and the radius of $\supp(v_0)$ \cref{supv0}, $r_0$, are related by: 
	\begin{equation}
		2\oneAtomPotSuppRad<\minLatticeSpacing \ ,\label{agt2r_0}
	\end{equation}
	{\it i.e.} the atomic wells comprising $V(x)$ have non-overlapping supports.
\end{assumption}
\noindent We believe that \cref{r0a} is not essential to our tight-binding limit procedure.

\subsubsection{The  orbital subspace, $\orbitalSubSp^\lambda$}\label{V-lambda}
The connection between our assumptions on the microscopic (atomic) Hamiltonian, $\oneAtomHamil[e]$, and the macroscopic Hamiltonian $\crysHamil[e]$ is accomplished  by  introducing the subspace of ground state {\it magnetic atomic orbitals}. 
For each $n\in\GG$, let $\oneAtomGrndStt[e]_n $  denote the ground state 
of $h_\lambda$ magnetically translated so that it is centered at $x=n$:
\begin{equation}
	\oneAtomGrndStt[e]_n := \magTransl{n} \oneAtomGrndStt[e]\ .\label{magphi-n}
\end{equation}
Since  $(h_\lambda-e_0^\lambda\Id)\varphi_0^\lambda=0$ and $[\hat{R}^n, h^\lambda]=0$, we have that 
\begin{equation}
	\oneAtomHamil[e]_n \oneAtomGrndStt[e]_n = 0, \ \label{h_nphi_n}
\end{equation}
where $\oneAtomHamil[e]_n$ is defined as the  recentered and magnetically  translated Hamiltonian:
\begin{equation}
	\oneAtomHamil[e]_n := \magTranslConj{n}{(\oneAtomHamil[e]-e_0^\lambda\Id)} =  (P-b Ax)^2+\lambda^2 v_0(x-n)-e_0^\lambda\Id\,.
	\label{ham_n}\end{equation}
For $\lambda$ sufficiently large, the set  $\orbitalSet[e]$ is linearly independent; see \cref{Gram-prop} below.

We now introduce $\orbitalSubSp^\lambda$, the subspace of $L^2(\RR^2)$ given by:
\begin{equation}
	\orbitalSubSp^\lambda:= \szpan\orbitalSet[e].
	\label{Vdef}
\end{equation}
Thus,  $L^2(\RR^2)=\orbitalSubSp\oplus\orbitalSubSp^\perp$ and we introduce the associated  orthogonal projections:
\begin{equation}
	\Pi:L^2(\RR^2)\to \orbitalSubSp,\quad \Pi^\perp:L^2(\RR^2)\to \orbitalSubSp^\perp\ .
	\label{PiPiperp}\end{equation}
For convenience we shall often suppress the $\lambda-$ dependence of $\orbitalSubSp^\lambda$ and these projections.

\subsection{The Gramian and an orthonormal basis for  $\orbitalSubSp^\lambda={\rm span}\{\oneAtomGrndStt[e]_n\}_{n\in\GG}$}\label{sec:gramian1}

At the heart of the proof of  \cref{res-conv} is the approximation of $\Pi H^\lambda \Pi$, where $\Pi=\Pi^\lambda$ is the orthogonal projection onto the
subspace $\orbitalSubSp={\rm span}\orbitalSet[e]$ of $L^2(\RR)$. 
Now the basis  $\orbitalSet[e]$ is not orthonormal, so in order to construct $\Pi^\lambda$ we introduce an orthonormal basis for 
$\orbitalSubSp$. 
For this  we introduce a linear operator, $\Gram\equiv\Gram[e]$, the Gramian, which acts on $\discHilSp$, and is defined via its matrix elements:
\begin{equation}
	\Gram[e]_{nm} := \ip{\oneAtomGrndStt[e]_n}{\oneAtomGrndStt[e]_m}_{\contHilSp}\qquad(n,m\in\discSp)\ .
	\label{Gram-def}
\end{equation} 
By normalization of $\oneAtomGrndStt[e]_0$, $\Gram[e]$ has ones on its diagonal. Below, we shall use the Gramian to define an appropriate orthonormal basis of  $\orbitalSubSp^\lambda$.

\begin{prop}\label{Gram-prop}  There exist constants $C, c>0$ and $\lambda_\star>0$ such that for all $\lambda>\lambda_\star$:
	\begin{enumerate}
		\item Locality property: The matrix elements $ \Gram[e]_{nm}$ satisfy the off-diagonal Gaussian decay bound:
		\begin{align} |\Gram[e]_{nm}|\leq C\ \lambda^{1\over 2}\  \ee^{-\frac{\lambda}{16}(|m-n|^2-(2r_0)^2)},\quad \textrm{for $ m\ne n$},
			\label{eq:locality of Gramian}\end{align}
		where we note that $|m-n|\ge a>2r_0$.
		\item The mapping $f\to G^\lambda f$, with $(Gf)_n=\sum_{m\in\GG}G^\lambda_{nm}f_m$  defines a self-adjoint and bounded linear operator on $\discHilSp$
		satisfying, for any $\ve>0$,
		\begin{equation} \|\Gram^\lambda-\Id\|_{_{\mathcal{B}(\discHilSp})}\le C\ee^{-\frac{1}{16}\lambda ((1-\ve)a^2-4r_0^2)}\ .\label{eq:G_offdiag}\end{equation}
		\item $G^\lambda=G:\discHilSp\to \discHilSp$ is strictly positive,  invertible and its inverse, $G^{-1}$ satisfies
		\begin{equation}
			\|G^{-1}\|_{_{\mathcal{B}(\discHilSp}}\ \le\ \frac{1}{1-\|G^\lambda-\Id\|_{_{\mathcal{B}(\discHilSp}}}\le \frac{1}{1-C\ee^{-c\lambda}},
			\label{eq:strictly positive gap for Gramian}
		\end{equation}
		and hence we have the lower bound on its gap,
		\begin{equation}
			{\rm dist}(0,\sigma(G))\ \ge\ 1-C\ee^{-c\lambda}\,.
			\label{gap-lb}\end{equation}
		\item The bounded invertible operator on $\discHilSp$:
		\[\ONBer^\lambda:=\sqrt{\Gram^{-1}}\] satisfies
		$$ \Gram^{-1} = M^2 \quad {\rm and}\quad G=M^{-2}.$$
	\end{enumerate}
\end{prop}

The proof of \cref{Gram-prop} is given in \cref{sec:gramian}. 
It will also be useful to have the following consequence of $G$ having a strictly positive spectral gap:
\begin{prop}
	$\ONBer[e]$ has off-diagonal exponentially decaying matrix elements:  Fix $a>2r_0$. There are constants $C,\tilde\mu$ and $\lambda_\star$, which depend on $a$ and $r_0$, such that for all  $\lambda>\lambda_\star$, 
\begin{equation}\label{Mnm-bd}
 |\ONBer_{nm}| 
 \leq C \exp(-\tilde\mu \lambda \normEuc{n-m})\qquad(n,m\in\discSp)\, . 
 \end{equation}
	\label{prop:M has off-diagonal decay too}
\end{prop}
\begin{proof}[Proof of Proposition \ref{prop:M has off-diagonal decay too}]
	 We recall a formula from \cite[Lemma 4]{Graf_Shapiro_2018_1D_Chiral_BEC}: For $B>0$, $ B^{-1/2} = C\int_0^\infty t^{-1/2}(B+t\Id)^{-1}\dif{t} $. Setting $B=G$, we  obtain
	  \begin{align*}
	 	|\ONBer_{nm}| 
	 	\leq C \int_{-\infty}^0 \frac{1}{\sqrt{-t}}\abs*{\left((\Gram-t\Id)^{-1}\right)_{nm}} \dif{t}\,.
		\end{align*}
		Hence, the bound \eqref{Mnm-bd} has been reduced to a bound on the integral kernel of the resolvent of $\Gram$, which 
		we obtain from a Combes-Thomas estimate \cite{AizenmanWarzel2016} on $G$.
		Specifically, we fix any $\mu>0$ and estimate $|G_{n,m}|e^{\mu |n-m|}$ for $n\ne m$ using \cref{eq:locality of Gramian}:
	\begin{align*}
	|G_{n,m}|e^{\mu |n-m|} &\lesssim e^{-\frac{\lambda}{16}\left(|n-m|^2-(2r_0)^2\right)}e^{\mu|n-m|}\\
	&\le e^{-\frac{\lambda}{16} |n-m|^2 + \frac{\lambda}4 r_0^2 \frac{|n-m|^2}{a^2} + a \mu \frac{|n-m|^2}{a^2}}\qquad (|n-m|\ge a)\\
	&= e^{|n-m|^2\left( -\frac{\lambda}{16}\left(1-\left(\frac{2r_0}{a}\right)^2\right) + \frac{\mu}{a} \right)}
	\end{align*}
	Choose $\mu(\lambda)=\tilde\mu \lambda$. Recalling our hypothesis that $a>2r_0$, we have using \eqref{eq:Gaussians are summable in discrete space} that 
if \[\tilde\mu < \frac{a}{16}\left(1-\left(\frac{2r_0}{a}\right)^2\right)\ ,\]	
then
	 $$ S_{\mu(\lambda)} := \sup_{n\in\discSp}\sum_{m\in\discSp} |G_{nm}|\left(e^{-\tilde\mu \lambda \normEuc{n-m}}-1\right)\to0
	 \quad \textrm{as $\lambda\to\infty$.} $$ 
	Since  $t\leq 0$, using \cref{gap-lb} we have 
	$$\dist{t,\sigma(G)}\geq |t| + 1 - C \ee^{-c\lambda}\,.$$ 
Applying \cite[Theorem 10.5, eqn (10.34)]{AizenmanWarzel2016} we have that if $t \le 0$, then 
	\begin{equation}\label{Gnm-dec} |(G-t\Id)^{-1}_{nm}| \leq \frac{1}{\dist{t,\sigma(G)}-S_\mu}e^{-\tilde\mu \lambda\normEuc{n-m}}, \qquad n,m\in\discSp \ .
	\end{equation}
	Hence 
	for $\lambda$ sufficiently large and some $c>0$,
	 $$\Big|(G-t\Id)^{-1}_{nm}\Big| \leq \frac{C}{|t|+c}\exp(-\tilde\mu\lambda\normEuc{n-m}),\qquad t<0;\quad n,m\in\discSp\,.$$
	We substitute this estimate into the integral and obtain our result, using that for all $c>0$,
	$$ \int_{t=0}^{\infty}\frac{1}{(c+t)\sqrt{t}}\dif{t} = \frac{\pi}{\sqrt{c}}. $$
\end{proof}
 
We note that it may be possible to improve \cref{Gnm-dec}
to Gaussian decay, in which case $M_{nm}$ would also have Gaussian, rather than exponential decay.
\bigskip

We now introduce an orthonormal basis, $\orbitalSetONB[e]$, for $\orbitalSubSp$. For all $n\in\GG$, let 
\begin{align}\oneAtomGrndSttONB[e]_n := \sum_{n'\in\discSp} \overline{\ONBer[e]_{nn'}} \oneAtomGrndStt[e]_{n'}\qquad(n\in\discSp)\label{eq:rotation from orbital basis to ONB}\end{align}
One can verify $\orbitalSetONB[e]$ is an orthonormal basis for $\orbitalSubSp$ and hence $\Pi=\sum_{n\in\GG}\left\langle\oneAtomGrndSttONB[e]_n,\cdot\right\rangle\oneAtomGrndSttONB[e]_n$. Indeed,
\begin{align*}
	\left\langle\oneAtomGrndSttONB[e]_n,\oneAtomGrndSttONB[e]_m\right\rangle&=
	\sum_{n^\prime,m^\prime}
	\left\langle \overline{M_{nn^\prime}}\varphi_{n^\prime},\overline{M_{mm^\prime}}\varphi_{m^\prime}\right\rangle\\
	& = \sum_{n^\prime,m^\prime}
	M_{nn^\prime}\left\langle \varphi_{n^\prime},\varphi_{m^\prime}\right\rangle \overline{M_{mm^\prime}}
	= \sum_{n^\prime,m^\prime}
	M_{nn^\prime} G_{n^\prime m^\prime}\overline{M_{mm^\prime}}\\
	&= \left(MGM^*\right)_{nm} =\ \left(M\ M^{-1} (M^*)^{-1}\ M\right)_{nm}\ =\ \delta_{nm}.
\end{align*}

Next, let $\delta_n= \{\delta_{nq}\}_{q\in\GG}$ denote a canonical basis  vector in $\discHilSp$ and introduce the mapping

\begin{align}
	J: L^2(\RR^2)\to \discHilSp\quad \textrm{ via}\quad   \oneAtomGrndSttONB[e]_n\mapsto\delta_n \qquad(n\in\discSp).
	\label{eq:cont to disc map}\end{align} 

We note that $J$ is a partial isometry:
\begin{equation}  J^*J = |J|^2 \ =\ \orbitalProj \,;\qquad JJ^\ast  =  |J^\ast|^2= \Id_{\discHilSp}\ .
	\label{p-isom}\end{equation}

Any operator $B$ on $\discHilSp$ can be lifted to an operator on $\orbitalSubSp$ via conjugation with $J$:
\begin{equation} \contToDiscBLO B := J^\ast B J\,. 
\label{J-lift}\end{equation}
Furthermore, using \cref{eq:rotation from orbital basis to ONB} one can verify that 
\begin{align}
	\contToDiscBLO M^{-1} \oneAtomGrndSttONB_n=\oneAtomGrndStt_n \,;\quad \contToDiscBLO M \oneAtomGrndStt_n=\oneAtomGrndSttONB_n\qquad(n\in\discSp)\,. \label{rotation}\end{align}

We conclude this section with the following
\begin{lem}
	If $\psi\in\orbitalSubSp$ is spanned as $\psi = \sum_{n\in\discSp} \beta_n\oneAtomGrndStt_n$ for some $\beta:\discSp\to\CC$ then $\beta\in\discHilSp$ and 
	\[ \norm{\beta}_{\discHilSp} = \norm{\contToDiscBLO M\psi}_{\contHilSp}\,;\quad \norm{M^{-1} \beta}_{\discHilSp} = \norm{\psi}_{\contHilSp}\ . \]
	\label{lem:Relating the little l^2 norm with the L^2 norm}
\end{lem}
\begin{proof}
	Since  $\orbitalSetONB$ is an orthonormal basis for $\orbitalSubSp$, it follows that $\norm{\beta}_{\discHilSp}^2 = \sum_{n\in\discSp} |\beta_n|^2 = \norm{\sum_{n\in\discSp} \beta_n \oneAtomGrndSttONB_n}^2_{\contHilSp} $. By \cref{rotation} and linearity of 
	$\contToDiscBLO(M^\ast)$, we have 
	\[ \norm{\beta}_{\discHilSp}^2=\norm{\contToDiscBLO M\sum_{n\in\discSp} \beta_n \oneAtomGrndStt_n}^2_{\contHilSp}=\norm{\contToDiscBLO M \psi}^2_{\contHilSp}.\]

\end{proof}

\subsection{Heuristics of the tight-binding approximation}\label{TB-ham}

In this section we  explain schematically how the tight-binding
approximation emerges.
The approximation is made rigorous in the proof of resolvent convergence in \cref{sec:proof of main theorem}.

As discussed above, for $\lambda$ large,   the low energy part of the spectrum of $H^\lambda$ will be reduced to the approximation of 
$\Pi \crysHamil[e] \Pi$.  By \cref{p-isom}, $ \Pi H^\lambda \Pi\ =\ J^*\left(J H^\lambda J^*\right)J$. Thus, we study the operator on the discrete space:
$ JHJ^*:\discHilSp\to \discHilSp$. 
Its matrix elements with respect to the orthonormal basis $\{\tilde\varphi_n^\lambda\}_{n\in\GG}$ are: 
\begin{equation}  \left[ J\crysHamil[e] J^*\right]_{nm}\ =\ \left\langle \tilde\varphi^\lambda_n, H^\lambda \tilde\varphi^\lambda_m\right\rangle,
	\quad m,n\in\GG.
	\label{mels-mn}\end{equation}
By \cref{Gram-prop}, we have that $\|G^\lambda-\Id\|_{\mathcal{B}(\discHilSp)}\lesssim \ee^{-c\lambda}$ and since 
$\tilde\varphi_n=J^*G^{-\frac12}J\varphi_n\approx J^*J\varphi_n = \varphi_n$ for $\lambda\gg1$, we shall approximate 
$\tilde\varphi_n$ by $\varphi_n$ in \cref{mels-mn} yielding
\begin{equation}  \left[ J\crysHamil[e] J^*\right]_{nm}\ \approx\ \left\langle \varphi^\lambda_n, H^\lambda \varphi^\lambda_m\right\rangle,
	\quad m,n\in\GG.
	\label{mels-mn-app}\end{equation}
Note that 
\begin{equation}
	\crysHamil[e] \oneAtomGrndStt[e]_m =  \sum_{k\in\setwo{\discSp}{m}} \lambda^2 v_k \oneAtomGrndStt[e]_m\ \approx 0,
	\label{Hphi-m}
\end{equation}
since $\varphi_m^\lambda$ has Gaussian decay outside the support of $v_m$ (\cref{thm:Gaussian decay of ground state}).
Thus,  
\begin{equation}  \left[ J\crysHamil[e] J^*\right]_{nm}\ \approx\ \left\langle \varphi^\lambda_n, H^\lambda \varphi^\lambda_m\right\rangle\ =\ \lambda^2\sum_{k\in\setwo{\discSp}{m}} \left\langle \varphi_n,v_k  \varphi_m\right\rangle,\quad m,n\in\GG.
	\label{mels-mn1}
\end{equation}
Below we show that the dominant contributions to the sum in \cref{mels-mn1} come from the $k=n$ term, 
where  $|m-n|= a$ is the minimal pairwise distance of distinct points in $\GG$. Hence, 
\begin{equation} 
	\left[ J\crysHamil[e] J^*\right]_{nm}\ \approx \ \lambda^2 \left\langle \varphi^\lambda_n,v_n  \varphi^\lambda_m\right\rangle\ \delta_{|m-n|,a},\quad m,n\in\GG.
	\label{Hmn_dom}
\end{equation}
Let us evaluate the right hand side of \cref{Hmn_dom}. Since $v_n(x)\equiv v_0(x-n)$ and $\varphi_n^\lambda(x)=\ee^{ix\cdot\lambda An}\varphi_0^\lambda(x-n)$, we have for $m,n\in\GG$:
\begin{align*}
	\left\langle \varphi^\lambda_n,v_n \varphi^\lambda_m\right\rangle\ &=\
	\int_{\RR^2}\ \ee^{ix\cdot\lambda A(m-n)}\ 
	\varphi^\lambda_0(x-n)\ v_0(x-n)\ \varphi^\lambda_0(x-m)\ dx\\
	&\ = \ee^{i\lambda n\cdot A(m-n)}\ 
	\int_{\RR^2}\ \ee^{iy\cdot\lambda A(m-n)}\ 
	\varphi^\lambda_0(y)\ v_0(y)\ \varphi^\lambda_0(y-(m-n))\ dy
\end{align*}
For simplicity let us take $m,n\in\discSp$ such that $\normEuc{m-n}=a$ and write $m-n= a\ \widehat{m-n}=a\ (\cos\theta,\sin\theta)$ and introduce the rotation matrix $R(\theta)$ such that 
$R(\theta)(m-n)=a\ e_1$.  Setting $y=R^\top(\theta) z$ in the previous integral, and using radial symmetry of $\varphi^\lambda_0$, 
we obtain 
\begin{align}
	\left\langle \varphi^\lambda_n,v_n  \varphi^\lambda_m\right\rangle\
	&= \ee^{i\lambda n\cdot A(m-n)}\ \int_{\RR^2}\ \ee^{iz\cdot\lambda  A R(\theta) (m-n)}\ 
	\varphi^\lambda_0(z)\ v_0(z)\ \varphi^\lambda_0\left(z-R(\theta)(m-n)\right)\ dz\nonumber\\
	&= \ee^{i\lambda n\cdot A(m-n)}\ \int_{\RR^2}\ \ee^{iz\cdot\lambda  a Ae_1}\ 
	\varphi^\lambda_0(z)\ v_0(z)\ \varphi^\lambda_0\left(z-ae_1\right)\ dz,
	\label{mels-mn2}\end{align}
where we have used that $A=\frac12 R(\frac{\pi}2)=\frac12\begin{pmatrix} 0&1\\-1&0\end{pmatrix}$ and $R(\theta)$ commute. Further, since $Ae_1=-\frac12 e_2$,
we obtain from  \cref{Hmn_dom} and \cref{mels-mn2}:
\begin{align}
	\left[ J\crysHamil[e] J^*\right]_{nm}\ &\approx\ \  \rho^\lambda\  \ee^{i\frac{\lambda}{2} n\wedge m}\ \delta_{|m-n|,a}\ , \quad m,n\in\GG.
	\label{JHJ*} \end{align}
where $n\wedge m =n_1m_2-n_2 m_1$ and $\rho^\lambda$ denotes the {\it magnetic hopping coefficient} defined by the oscillatory integral:
\begin{align}
	\rho^\lambda &:= \int_{\RR^2} \ee^{-i\frac{\lambda a}{2}z_2}\  
	\varphi^\lambda_0(z)\ \lambda^2v_0(z)\ \varphi^\lambda_0(z-a e_1)\ dz.
	\label{rho-def}
\end{align}
In \cite{FSW_2020_magnetic_double_well} rigorous lower bounds for $\rho^\lambda$ (a-posteriori a positive number) were obtained. 
We discuss this result in \cref{sec:lb-hop}.

We proceed by dividing \cref{JHJ*} by $\rho^\lambda$, and defining $\tilde{H}^\lambda := (\rho^\lambda)^{-1} H^\lambda$, we have 
\begin{align}
	\left[ J \tilde{H}^\lambda J^*\right]_{nm}\ &\approx\ (\rho^\lambda)^{-1}\ \lambda^2 \left\langle \varphi^\lambda_n,v_n  \varphi^\lambda_m\right\rangle\  \delta_{|m-n|,a}\ =\    
	\ee^{i\frac{\lambda }{2}  n\wedge m}\ \delta_{|m-n|,a},\ \quad m,n\in\GG.
	\label{JHJ*1} \end{align}

Note that the phase in \cref{JHJ*1} is oscillatory  and so there may not be limiting operator as $\lambda\to\infty$. 
However, $\left[ J \tilde{H}^\lambda J^*\right]_{nm}$ may have a limit along certain sequences, $\{\lambda_j\}$, which tend to infinity.

\begin{defn}\label{def:admiss}
Let $\beta\in[0,2\pi)$. We say that $\{\lambda_j=\lambda_j(\beta)\}$ is an {\it admissible sequence with
(magnetic flux per unit cell) parameter $\beta$} if 
 \begin{align} \lim_{j\to\infty}\exp(\ii \frac{\lambda_j}{2} n\wedge m) = \exp(\ii\beta n\wedge m)\qquad(n,m\in\discSp)\,. \label{eq:appropriate beta-sequences}\end{align}
\end{defn}
\begin{defn} \label{def:admiss-ham} Given an admissible sequence with
(magnetic flux per unit cell) parameter $\beta$, we define the associated tight-binding Hamiltonian as 	\begin{equation} \tbH_{\beta;m,n}\ = 
		\exp\left(\ii \beta\  n\wedge m\right)\ \delta_{|m-n|,a}
		\qquad(n,m\in\discSp)\,.
		\label{HTBnm}\end{equation}
\end{defn}

We present some examples of magnetic tight binding models.
\begin{example}[Tight binding model for $\GG\subset(a\ZZ)^2$]\label{Z2-tb}
	Let $\GG=(a\ZZ)^2$  or any infinite connected (via edges of length $a$) subset of $(a\ZZ)^2$. We write $n=(n_1,n_2)$ and $m=(m_1,m_2)$. For any $\beta\in\RR$ and $\{r_j\}\subset\ZZ$ tending to infinity, a sequence $\left\{\lambda_j\right\}_{j\in\NN}$  with
	$$  \Big|\lambda_j - 2(\beta+\frac{2\pi}{a^2} r_j)\Big|\to0\qquad(j\to\infty)$$ is admissible and the tight-binding Hamiltonian is given by \cref{HTBnm}. In particular, \[\tbH_{n,n\pm ae_1} = \exp(\pm \ii \beta n_2),\quad \tbH_{n,n\pm ae_2} =  \exp(\mp \ii \beta n_1),\quad n=(n_1,n_2)\in(a\ZZ)^2.\]

This is the celebrated Harper model Hamiltonian (in the symmetric gauge) (see \cite{Avron1990}), 
which has so-called Hofstadter butterfly spectrum \cite{PhysRevB.14.2239}. It is known \cite{doi:10.1063/1.1412464} to have non-trivial Chern numbers associated with some of its sub-bands depending on the value $\beta$.
\end{example}

\begin{example}[Tight binding model for the honeycomb lattice]\label{honeycomb}

In this example we suppress the lattice constant, $a$,  in order to avoid cluttered expressions. Let  
 $\Lambda$ denote the equilateral triangular lattice:
  $\Lambda=\ZZ v_1\oplus \ZZ v_2$, with 
   $v_1=(\sqrt3/2,1/2)$ and $v_2=(\sqrt3/2,-1/2)$.
The honeycomb lattice, $\mathbb{H}$, is the union of two interpenetrating equilateral triangular sublattices: $\mathbb{H}=\Lambda^A\cup\Lambda^B$. Here, $\Lambda_A=v_A+\Lambda$, $\Lambda_B=v_B+\Lambda$, where  $v_A$ and $v_B$ are base points located so that the points of $\mathbb{H}$ lie at the vertices of a hexagonal tiling of $\RR^2$:
   $v_A=(0,0)$ and $v_B(0,1/\sqrt3)$.
      The plane can also be tiled by parallelograms, indexed by $n=(n_1,n_2)\in\ZZ^2$, each containing one $A-$ site, $v_A^{n_1,n_2}\in \Lambda_A$, and one $B-$ site, $v_B^{n_1,n_2}\in \Lambda_B$, where  $v^{n_1,n_2}_I=v_I + n_1v_1+n_2v_2$, for $I=A,B$ and $(n_1,n_2)\in\ZZ^2$.

 Consider the sites, $v_A^{n_1,n_2}$ and  $v_B^{n_1,n_2}$ in the $(n_1,n_2)-$ parallelogram.
Referring to the specific dimerization given in \cite[Figure 1.2]{FLW17_doi:10.1002/cpa.21735}, 
the site $v_A^{n_1,n_2}$ has nearest neighbors: $v_B^{n_1,n_2}$, $v_B^{n_1-1,n_2}$
 and $v_B^{n_1,n_2-1}$ and the site
  $v_B^{n_1,n_2}$ has nearest neighbors: $v_A^{n_1,n_2}$, $v_A^{n_1+1,n_2}$
 and $v_A^{n_1,n_2+1}$. It follows from  \cref{HTBnm}
  that
  \begin{align*}
\left(H^{TB}\psi\right)_{n_1,n_2}
  = \begin{pmatrix}
  \ee^{\ii \Phi_{n,n}}\psi^B_{n_1,n_2} +
   \ee^{\ii \Phi_{n,(n_1-1,n_2)}}\psi^B_{n_1-1,n_2}+\ee^{\ii \Phi_{n,(n_1,n_2-1)}}\psi^B_{n_1,n_2-1}\\
  \ee^{\ii \Phi_{n,n}}\psi^A_{n_1,n_2} +
   \ee^{\ii \Phi_{n,(n_1+1,n_2)}}\psi^A_{n_1+1,n_2}+\ee^{\ii \Phi_{n,(n_1,n_2+1)}}\psi^A_{n_1,n_2+1}
  \end{pmatrix}
  \end{align*}
  where $$ \Phi_{n,m} := \beta v_A^{n}\wedge v_B^{m}  \qquad n,m\in\ZZ^2 .$$
 This is the model found in \cite[eqn (2.1)]{Agazzi2014} (though in a different gauge).
\end{example}

Though we have not formulated \cref{res-conv} to deal with $\lambda$ dependent lattices $\GG$, we discuss now two related examples:
\begin{example}[Disordered lattice]\label{example:random hopping model} Consider the disordered $\lambda-$ dependent lattice presented in \cref{rem-Glam}. By \cref{JHJ*1}, the tight-binding Hamiltonian has matrix elements 
\begin{align}
		H_{nm}^{TB} \ =\  \ee^{i\beta(n\wedge m)}q_{nm}\ \delta_{|m-n|,a} ,\quad m,n\in\GG\subset (a\ZZ)^2.\label{eq:random tb Hamiltonian matrix elements}
		\end{align}
	Here, $\Set{q_{nm}}_{n,m\in(a\ZZ)^2}$ are the hopping coefficients, given by the emergent tight-binding hopping amplitude, proportional to $ \rho(d_{nm})/\rho(a)$  where $d_{nm}$ is the continuum distance between the two lattice sites $n,m$ and $a$ is the minimal lattice spacing. 
	By \cref{eq:bound on rho}, just below,  $\rho(d)$ scales like $\exp(-c\lambda \normEuc{d}^2)$ so the hoppings are indeed disordered;  if $\norm{d_{nm}}$ scales like $\sqrt{1+\frac{1}{\wellDepth}\alpha_{nm}}$, as in \cref{rem-Glam}, then the hoppings are asymptotically $\lambda$ independent and proportional to $$ q_{nm}\sim \exp(-c a^2 \alpha_{nm})\,. $$

\end{example}

\begin{example}[Random displacements from $(a\ZZ)^2$]\label{rdm}
Here we present an example, related to \cref{rem-Glam} and \cref{example:random hopping model}, in which  we introduce the randomness concretely. 
We describe how $\discSp$ may be chosen so as to obtain a so-called {\it random displacement model}, similar to the models discussed in, for example, \cite{Klopp_et_al_2012_10.1007/978-3-0348-0414-1_10}. Pick two sequences of independent, identically distributed random variables: $$ \Set{t_n}_{n\in\ZZ^2}\subseteq\left[0,1\right],\Set{\theta_n}_{n\in\ZZ^2}\subseteq[0,2\pi)\,. $$
Then, define $\discSp=\discSp_\lambda$ as the following deformation of $(a\ZZ)^2$: For each $n=a(n_1,n_2)\in(a\ZZ)^2$, introduce a random displacement:
$$\mathfrak{s}:a(n_1,n_2) \mapsto a(n_1,n_2) + \frac{a}{\lambda}t_n(\cos(\theta_n),\sin(\theta_n)),\quad (n_1,n_2)\in\ZZ^2$$ to obtain $\discSp_\lambda=\mathfrak{s}\left((a\ZZ)^2\right)$.
Consider the neighboring points: $\mathfrak{s}(n)\in\GG_\lambda$ and 
 $\mathfrak{s}(m)\in\GG_\lambda$, where $n=a(n_1,n_2)$ and $m=a(n_1+1,n_2)$.
The edge-length between  $\mathfrak{s}(n)$ and $\mathfrak{s}(m)$ is 
\begin{align}
    \normEuc{\mathfrak{s}(n)-\mathfrak{s}(m)} &= a\sqrt{\left(1+\frac{t_n}{\lambda}  \cos(\theta_n) - \frac{t_m}{\lambda}  \cos(\theta_m)\right)^2 + \left(\frac{t_n}{\lambda}  \sin(\theta_n) - \frac{t_m}{\lambda}  \sin(\theta_m)\right)^2} \nonumber\\
    &= a \sqrt{1+\frac{2}{\lambda}\left(t_n\cos(\theta_n)-t_m\cos(\theta_m)\right)+\mathcal{O}(\frac{1}{\lambda^2})} \nonumber\\
    &= a \sqrt{1+\frac{1}{\lambda} \ve_{nm}+\mathcal{O}(\frac{1}{\lambda^2})},
\quad \textrm{where} 
\quad \ve_{nm}:=2(t_n\cos(\theta_n)-
t_m\cos(\theta_m))\label{eq:random lengths in rdm}\,.\end{align}

The  minimal lattice spacing for $\GG_\lambda$  is not precisely equal to $a$ but rather $\approx a(1-\frac{2}{\lambda})$
which is $\lambda$ dependent. However, this does not substantively effect the analysis; one simply replaces $a$ with $a/(1-\frac{2}{\lambda})$ in the initial definition of the lattice to get exactly $a$. We refrain from doing this in order to avoid cluttered notation in the resulting expressions.
For any bond $b = \Set{n,m}$ the edge lengths of $\discSp_\lambda$ are of the form
$$ \sqrt{1+\frac{1}{\wellDepth}\ve_b}\times a $$ 
where $a>0$ is the fixed minimal lattice spacing and $\Set{\ve_b}_{b\in E(\ZZ^2)}$ are random variables (no longer independent), and since we know that $\RR^2\ni d\mapsto \nnHopping(d)$ scales like $\exp(-c\wellDepth\normEuc{d}^2)$ (see \cref{eq:bound on rho} below), the hopping coefficients associated with these edges are of the general form
$$ \nnHopping\left(\sqrt{1+\frac{1}{\wellDepth}\ve_b}d\right)\sim\exp(-c\wellDepth(1+\frac{1}{\wellDepth}\ve_b)\normEuc{d}^2) \sim \nnHopping(d) \exp\left(-c\ve_b\normEuc{d}^2\right)\,. $$

The matrix elements of the Hamiltonian are still of the form \cref{eq:random tb Hamiltonian matrix elements} with $$ q_{nm} \sim  \exp(-c\ve_{nm}a^2)\,. $$

\end{example}

\subsection{Lower bounds on magnetic hopping}\label{sec:lb-hop}

The following two results, proved in \cite[Theorem 1.6]{FSW_2020_magnetic_double_well} and \cite[Theorem 1.9]{FSW_2020_magnetic_double_well} , play an important role in our study of $H^\lambda$ for $\lambda\gg1$. 
 Recall that $\rho^\lambda$ denotes the magnetic hopping coefficient displayed in \cref{rho-def}.
\begin{thm}[Lower bound on double-well hopping probability]\label{thm:lb-hop}
Assume that $a$, the minimal pairwise distance in $\GG$,
 and $r_0$, the radius of the support of $v_0$, satisfy the relation:
 \begin{align}
		a > 2r_0.
		\label{eq:constraint on lattice and potentail support radius}
	\end{align}
	Then, $\rho^\lambda$ satisfies the bounds
	\begin{align}
		\exp(-\frac14\wellDepth(a^2+4\sqrt{|v_{\mathrm{min}}|}a+\gamma_0))\leq \nnHopping^\lambda \leq C \wellDepth^{5/2}\exp(-\frac14\wellDepth((a-r_0)^2-r_0^2))\ .
		\label{eq:bound on rho}
	\end{align}
	The constants $C$ and $\gamma_0>0$ are independent of $\wellDepth$, and $\gamma_0$ is independent of $a$. 
\end{thm}

\begin{thm}[Nearest neighbor vs. beyond nearest neighbor hopping]\label{nn_v_nnn}
 Assume $a$ and $r_0$ are as in \cref{thm:lb-hop}.
	Let $\xi\in\RR^2$ such that $\normEuc{\xi}>2r_0$, and let $x>1$. Then,  there are constants $C_\star$ and $\lambda_\star$ such that for all $\lambda\geq\lambda_\star$, 
	$$ \abs*{\frac{\nnHopping^\lambda(x \xi)}{\nnHopping^\lambda(\xi)}} \leq C_\star \exp\left(-\frac{1}{8}\lambda(x^2-1)\normEuc{\xi}\right) $$
	\label{thm:monotonicity of rho in displacement} where $\rho^\lambda(\xi)$ is defined to be the nearest-neighbor hopping evaluated at displacement $\xi$: $$ \rho^\lambda(\xi) \equiv \int_{\RR^2} \ee^{-i\frac{\lambda a}{2}z \wedge \xi}\  
	\varphi^\lambda_0(z)\ \lambda^2v_0(z)\ \varphi^\lambda_0(z-\xi)\ dz \,.$$
	We note that by definition, $\rho^\lambda(e_1)=\rho^\lambda$, given in \cref{rho-def}.
\end{thm}
\section{Main theorems}\label{res-conv-thm}

Since $ \Pi \tilde{H}^\lambda \Pi\ =\ J^*\left(J \tilde{H}^\lambda J^*\right)J$, the discussion of  \cref{TB-ham} suggests that for sequences $\Set{\lambda_j}_j$ which satisfy \cref{eq:appropriate beta-sequences} we have:
\begin{equation}
	\Pi \tilde{H}^\lambda \Pi\ \approx\ J^* \tbH \ J,\quad \lambda\gg1.
	\label{J*HTBJ}\end{equation}
We prove convergence in the norm resolvent sense.
\begin{thm}[Resolvent convergence to tight-binding]\label{res-conv}
Let $a$ denote the minimal pairwise lattice spacing in $\GG$; see \cref{min-sp}. 
Let $\{\lambda_j(\beta)\}$ denote an admissible sequence with flux parameter $\beta$ and $H^{\rm TB}$ the associated tight-binding Hamiltonian; see Definitions \ref{def:admiss} and \ref{def:admiss-ham}. 

There exists $a_0>0$, such that for all $a>a_0$ the following holds:
	 Fix $K$, a compact subset of the resolvent set of $H^{\rm TB}$.
	 There exist positive constants  $\lambda_\star$ and $C,c$,  such that  for all $j$ with $\lambda_j>\lambda_\star$ and all $z\in K$:
	\[
	\Big\|\ \left((\rho^{\lambda_j})^{-1} H^\lambda-z\Id\right)^{-1}- J^*\left( \tbH -z\Id \right)^{-1}J\Big\|_{_{\mathcal{B}(L^2(\RR^2))}}\le C\ee^{-c\lambda_j}\,.\]
\end{thm}

\subsection{Application to topological equivalence}\label{appl-top}

Choose $\mu>\inf\sigma(H^\lambda)$ and $\mu\not\in\sigma(H^\lambda)$. Let  $\Gamma_\mu$ denote a contour in $\CC$ encircling $\sigma(H^\lambda)\cap (-\infty,\mu)$ but not intersecting $\sigma(H^\lambda)$. The Riesz projection formula gives for the spectral projection onto $(-\infty,\mu)$:
\begin{equation}
\chi_{\left(-\infty,\mu\right)}(H^\lambda) = -\frac{1}{2\pi\ii}\oint_{\Gamma_\mu} (H^\lambda-z\Id)^{-1} \dif{z} .
\label{riesz}\end{equation}
 \cref{res-conv} together with \cref{riesz} implies that, suitably scaled, the spectral projections onto isolated regions of the spectrum of $\crysHamil^\lambda$, converge in operator norm to those of $\tbH$ as $\lambda$ tends to infinity.
 
\subsubsection{Equality of continuum and tight binding bulk indices}
\label{eq-bulk}

We first prove the equality of the bulk topological indices, computed from $H^\lambda$ for $\lambda$ sufficiently large and from $H^{TB}$; see \cref{topology} for explanation of the formulas we use for the indices and their range of validity. 

Let $X^{\mathrm{TB}}$, with components $X^{\mathrm{TB}}_1,X^{\mathrm{TB}}_2$,  be the position operator on $\discHilSp$, i.e., $(X^{\mathrm{TB}}\psi)(n) = n\psi(n)$ for all $n\in\discSp$ and $\psi\in\discHilSp$. 
Let $S\subseteq\RR$ be an isolated subset of $\sigma(\tbH)$, the spectrum of $\tbH$. It is known \cite{Aizenman_Graf_1998} that $P^{\mathrm{TB}}:=\chi_S(\tbH)$ has a well-defined Hall conductivity associated with it,
 given by the index of the Fredholm operator:
 $$ P^{\mathrm{TB}} U^{\mathrm{TB}} P^{\mathrm{TB}} + \Id-P^{\mathrm{TB}}\ \textrm{ on $\discHilSp$}. $$  
 Here, $ U^{\mathrm{TB}} := \exp(\ii \arg(X_1^{\mathrm{TB}}+\ii X_2^{\mathrm{TB}}))$  denotes the unitary operator associated with {\it flux insertion} at the origin. In the translation invariant setting, this index agrees with the Chern number of the vector bundle associated with $P^{\mathrm{TB}}$ \cite{Bellissard_1994}.
 
Let us now introduce the spectral projection of the scaled continuum Hamiltonian $ P^\lambda := \chi_{ S}\left((\nnHopping^\lambda)^{-1}H^\lambda\right)$.

\begin{thm}[Topological equivalence for bulk geometries]\label{thm:topological equivalence for bulk geometries}
 For $\lambda$ sufficiently large we have:
 \begin{itemize}
     \item $P^\lambda U P^\lambda + \Id-P^\lambda$ is a Fredholm operator on $\contHilSp$.
     \item Equality of the continuum and discrete indices:\[  {\rm index}_{l^2(\GG)}\left( P^{\mathrm{TB}} U^{\mathrm{TB}} P^{\mathrm{TB}} + \Id-P^{\mathrm{TB}}\right) = {\rm index}_{L^2(\RR^2)}\left(P^\lambda U P^\lambda + \Id-P^\lambda \right)\,.\]
 \end{itemize}
 
\end{thm}

\cref{thm:topological equivalence for bulk geometries} is proved in \cref{sec:topological equivalence}.

\begin{rem}[Mobility gap; see also \cref{rem:the mobility gap}]\label{mob-gap} In \cref{thm:topological equivalence for bulk geometries}, we assume that $S$ is an isolated spectral subset of the tight-binding Hamiltonian, i.e., $H^\lambda$, has  a spectral gap. More generally when $\discSp$ is disordered, it might happen that such a spectral gap would close and fill with localized states, a situation referred to as \emph{the mobility gap regime} \cite{EGS_2005}. In this regime, the topological invariants continue to be well-defined (assuming some form of off-diagonal decay on the position basis matrix elements of $P^{\mathrm{TB}}$) and we expect the equality to still hold. 
\end{rem}

\subsubsection{Equality of continuum and tight binding edge indices}
\label{sec:edge}
A system with an edge is commonly obtained  by truncation to a half-space of the bulk or by insertion of a domain wall line-defect \cite{FLW-2d_edge:16,DW:19}.
 In either case, sufficiently far from the ``edge'',  the system should "look like" the bulk system. Let us suspend for a moment our usual notations in this paper and use generic $H$ to denote a bulk Hamiltonian (discrete or continuum setting) and $\hat{H}$ to denote  an edge Hamiltonian derived from it by truncation. Suppose that the bulk system's Hamiltonian, $H$, has a spectral gap which 
 includes an open interval, $\Delta$. Since our Hamiltonians act locally, the position basis matrix elements of $H-\hat{H}$ tend to zero with increasing distance from the edge. Moreover, $\hat{H}$ is a non-compact perturbation of $H$, and so the operator  $\hat{H}$ may have essential spectrum within $\Delta$. Indeed, in topologically nontrivial cases we expect the bulk gap, $\Delta$, to be filled. Nevertheless, due to the convergence operators, even in this case we expect that far from the edge, the resolvent of $\hat{H}$ restricted to the spectral subspace associated with $\Delta$ will be exponentially decaying.
 
 Let now $g:\RR\to[0,1]$ be a smooth function, which transitions from the value $1$ at $-\infty$ to the value $0$ at $+\infty$ and such that $g^\prime$ is non-zero within $\Delta$. 
Thus also $\supp(g^2-g)\subset\Delta$.

 We note that while $g(H)$ is trivially a projection, $g(\hat{H})$ itself is not one,
but the above intuition suggests that $g(\hat{H})$ is approximately a projection, {\it i.e.} $g(\hat{H})^2-g(\hat{H})$ (which is spectrally supported within $\Delta$), and has position basis matrix elements (i.e., integral kernel) which decay rapidly with distance from the edge. This reflects the decay of the "edge" states of $\hat{H}$ that fill the bulk gap, $\Delta$.

It is therefore natural to take the decay of the matrix elements of $ g(\hat{H})^2-g(\hat{H})$ in a direction $\hat{d}$ as a generalized definition of a edge Hamiltonian
 with edge-direction  $\hat{d}^\perp$, without making any reference to a bulk Hamiltonian:

\begin{defn}[Generalized Edge Hamiltonian with edge direction perpendicular to $\hat{d}$]\label{def:edge Hamiltonian} 
Let $g:\RR\to[0,1]$ be a smooth function and $\hat{d}\in\mathbb{S}^1\subseteq\RR^2$ be a unit vector. A Hamiltonian $K\in\mathcal{B}(\discHilSp)$ is a {\it generalized edge Hamiltonian
 with edge perpendicular to the direction $\hat{d}$ with "bulk-gap" containing $\supp(g^2-g)$} iff for any $N\in\NN$ sufficiently large there is a constant $C_N$ and another constant $\kappa$ such that
    \begin{align} \left|\left(g(K)^2-g(K)\right)_{n,m}\right| \leq C_N(1+\normEuc{n-m})^{-N}\exp(-\kappa |\hat{d}\cdot n|)\qquad(n,m\in\discSp)\,.
\label{eq:g(H) almost a projection}\end{align}
\end{defn}

Note that even in a situation where $\GG$ has no obviously discernible physical edge, e.g. if $\GG$ is very disordered, the property \cref{eq:g(H) almost a projection} gives a sense in which $\hat{H}$ has an edge along the direction $\hat{d}^\perp$. Furthermore, there are many systems for which \cref{eq:g(H) almost a projection} can be verified. For example,

\begin{lem}\label{lem:edge definition makes sense for simple systems}
    Let $K$ denote a (bulk) Hamiltonian acting on $l^2(\ZZ^2)$ with spectral gap containing $\Delta$, and introduce a truncation to $\ZZ\times\NN$, $\hat{K}$, which acts on $l^2(\ZZ\times\NN)$ (with any boundary conditions which decay rapidly into the bulk). Then, $\hat{K}$ is an edge Hamiltonian with edge direction $\hat{d}^\perp=\hat{e}_1=(1,0)$, with respect to the $g$, in the sense of \cref{def:edge Hamiltonian}. Here, $g$ is chosen  such that $g(\hat{K})(\Id-g(\hat{K}))$ is a smoothed out projection onto the bulk gap of $K$.
\end{lem}

\begin{proof}
Apply \cite[Lemma A.3 iii]{Elbau_Graf_2002}, with  $G:=g(g-1)$, which is supported in $\Delta$. 
\end{proof}
Finally, one could envision implementing edge truncations that are not simply along a straight line. We do not pursue this direction here and assume, going forward in this section, that the edge truncation is in the direction $\hat{d}^\perp=e_1=(1,0)$ so that $\hat{d}=e_2=(0,1)$.

Next, introduce a {\it switch function}, $\Lambda:\RR\ni s \mapsto\Lambda(s) \in [0,1]$ \cite{Elbau_Graf_2002} which is equal to $0$ for all $s\ll-1$ and equal to $1$ for all $s\gg1$. Let $X^{\rm TB}_1$ and $X^{\rm TB}_2$ denote discrete position operators on $\GG$ and introduce the abbreviated notation:
$$\Lambda^{\mathrm{TB}}_j := \Lambda(X^{\mathrm{TB}}_j) .$$
For example, take $\Lambda$ as the Heaviside step function.
 Analogously, in the continuum we have  $\Lambda_j(X) \equiv \Lambda(X_j)$ where $X_j$ is the position operator in the $j$th direction on $\contHilSp$.

\begin{lem}[The discrete index is well-defined]\label{lem:discrete edge index is well-defined} Suppose that $g:\RR\to[0,1]$ is a smooth function such that $\tbH$ is an edge Hamiltonian with respect to $g$ (in the sense of \cref{def:edge Hamiltonian}) with edge-truncation in the direction  $\hat{d}^\perp=(1,0)=e_1$. Then,
\begin{equation}
 \Lambda^{\mathrm{TB}}_1 \exp(-2\pi\ii g(\tbH))\Lambda^{\mathrm{TB}}_1+\Id-\Lambda^{\mathrm{TB}}_1 \label{pr-ind}\end{equation}
is a Fredholm operator on $\discHilSp$. 
\end{lem}
The index of this operator, as explained in \cref{topology}, is the Hall conductance along the edge; see \cref{eq:edge Hall conductivity as a Fredholm index} below.
\begin{proof}
    By \cite[Prop A.3]{Fonseca2020}, the operator \cref{pr-ind} is Fredholm if the commutator $$ \left[\Lambda^{\mathrm{TB}}_1,\,\exp(-2\pi\ii g(\tbH))\right]  $$ is compact. By \cite[Prop A.4]{Fonseca2020} this commutator is trace-class (and hence compact), due to our assumed bound \cref{eq:g(H) almost a projection}.
\end{proof}

\begin{thm}[Topological equivalence for edge geometries]\label{thm:edge continuum discrete correspondence} Suppose there is  smooth function $g:\RR\to[0,1]$  such that, in the sense of \cref{def:edge Hamiltonian}, $\tbH$ is an edge Hamiltonian with edge-truncation in the direction  $\hat{d}^\perp=(1,0)=e_1$ (with respect to $g$). Then, for $\lambda$ sufficiently large,  the operator  $$ \Lambda_1 \exp(-2\pi\ii g(H^\lambda/\nnHopping^\lambda))\Lambda_1+\Id-\Lambda_1 $$  is a Fredholm operator on $\contHilSp$ and 
 \begin{align*}& {\rm index}_{L^2(\RR^2)}\Big( \Lambda_1 \exp(-2\pi\ii g(H^\lambda/\nnHopping^\lambda))\Lambda_1+\Id-\Lambda_1 \Big)=\\
 &\qquad\qquad = {\rm index}_{l^2(\GG)}\Big( \Lambda^{\mathrm{TB}}_1 \exp(-2\pi\ii g(\tbH))\Lambda^{\mathrm{TB}}_1+\Id-\Lambda^{\mathrm{TB}}_1\Big)\,.\end{align*}
 
 The conclusions hold for any $g$ satisfying the constraints of \cref{def:edge Hamiltonian} and the index stays the same as long as $\supp(g^2-g)$ is unchanged.
\end{thm}

The proof of \cref{thm:edge continuum discrete correspondence} is given in  \cref{sec:topological equivalence}.
\subsubsection{The bulk-edge correspondence}
Using these two preceding theorems, as well as the known bulk-edge correspondence for discrete systems \cite{Elbau_Graf_2002,SBKR_2000}, we obtain an alternative proof of the bulk-edge correspondence for continuum systems.
\begin{cor}[The bulk-edge correspondence for continuum systems]\label{cor:bec} Let\ $\GG_{\rm bulk}, \GG_{\rm edge}\subseteq\contSp$ be two discrete sets such that $\GG_{\rm edge}\subset\GG_{\rm bulk}$ and such that $H^\lambda_{\GG_{\rm bulk}}$ has a spectral gap about $\mu\in\RR$, $g:\RR\to[0,1]$ is a smooth version of $\chi_{(-\infty,\mu)}$ and $H^\lambda_{\GG_{\rm edge}}$ obeys \cref{def:edge Hamiltonian}
 corresponding to $g$ and $\GG_{\rm edge}$ along some direction $\hat{d}^\perp\in\mathbb{S}^1$.
  Then if $\lambda$ is sufficiently large, the topological indices associated with $H^\lambda_{\GG_{\rm bulk}}$ and $H^\lambda_{\GG_{\rm edge}}$ agree; see \cref{fig:commutative diagram} of the Introduction.
\end{cor}

The first proof of the bulk-edge correspondence in the continuum setting for Hamiltonians with a  spectral gap appeared in \cite{Kellendonk_Schulz-Baldes_Boundary_Maps_2004}, and for Hamiltonians with a  mobility gap in \cite{Taarabt_2014arXiv1403.7767T}; see also \cite{Bourne2018,Alexis_BEC_19,Hannabuss2018,Braverman2019,10.1007/978-3-030-24748-5_2,Frohlich_2018_doi:10.1142/S0129055X1840007X}.

\section{Detailed outline of the proof of resolvent convergence \cref{res-conv}}\label{sec:proof of main theorem}
Recall the crystal Hamiltonian  ${H}^\lambda=(P-bA)^2+\lambda^2V-e_0^\lambda$ and introduce
\begin{align} \tilde{H}^\lambda :=
	(\rho^\lambda)^{-1}\crysHamil^\lambda \quad\textrm{acting in 
		$L^2(\RR^2)=\orbitalSubSp\oplus\orbitalSubSp^\perp$}. \label{eq:the scaled crystal Hamiltonian}\end{align} 
			We let $\{\lambda_j=
		\lambda_j(\beta)\}$ be an admissible sequence with flux parameter $\beta$ and $H^{\rm TB}_\beta=H^{\rm TB}$
		 is the associated tight-binding Hamiltonian; see Definitions \ref{def:admiss} and \ref{def:admiss-ham}.
	To simplify notation, we often omit the subscript from $\wellDepth$. 
		 
We shall prove that if $z\in\CC\setminus\sigma(\tbH)$ then
\begin{equation}
	\Big\|\ \left( \tilde{H}^\lambda-z\Id\right)^{-1}- J^*\left( \tbH -z\Id \right)^{-1}J\Big\|_{_{\mathcal{B}(L^2(\RR^2))}}\le C\ee^{-c\lambda} \frac{\norm{R_{\tbH}(z)}^2}{1-\norm{R_{\tbH}(z)}\ee^{-c\lambda}}
	\label{res-conv1}\end{equation}
via  a reduction strategy based on the Schur complement, which we now review.

Let $W$ be a $\ZZ_2$ graded vector space, i.e., $W = W_0 \oplus W_1$ and let $L$ be a linear operator acting on $W$. Then, $L$ can be expressed equivalently in block form as an operator on $W_0\oplus W_1$:
$$ L = \begin{bmatrix}A & B \\ C & D\end{bmatrix}, $$ 
where $A:W_0\to W_0$, $B:W_1\to W_0$, $C:W_0\to W_1$ and $D:W_1\to W_1$. 
If the operators  $D$ and $S:=A-BD^{-1}C: W_0\to W_0$ are invertible, then $L^{-1}$ exists  and 
is given by:
\begin{align} L^{-1} = 
	\begin{bmatrix}S^{-1} & -S^{-1}BD^{-1} \\
		-D^{-1}CS^{-1} & D^{-1}+D^{-1}CS^{-1}BD^{-1}
	\end{bmatrix}\,.  \label{eq:Schur's complement}\end{align}
We apply \cref{eq:Schur's complement} to study the invertibility $\tilde{H}-z\Id$, viewed as a block operator on 
$\orbitalSubSp\oplus \orbitalSubSp^\perp$:
$$ \tilde{H} = 
\begin{bmatrix}\Pi\tilde{H}\Pi & \orbitalProj \tilde{H}\orbitalProj^\perp \\
	\Pi^\perp \tilde{H}\Pi & \Pi^\perp \tilde{H}\Pi^\perp   \end{bmatrix}.
$$ 
It will be convenient at times to use the abbreviated conjugacy notation:
\[\orbitalProjBLO B \equiv \orbitalProj B \orbitalProj.\] 
By the above discussion $R_{\tilde{H}}(z)\equiv(\tilde{H}-z\Id)^{-1}$ exists if:
\begin{equation}
	D(z) :=  \Pi^\perp\left( \tilde{H} -z\Id\right)\Pi^\perp\ \equiv\ \orbitalProjBLO^\perp\left( \tilde{H} -z\Id\right)\ \quad\textrm{ is invertible on $\orbitalSubSp^\perp$} 
	\label{Dinv}\end{equation} 
and 
\begin{align} S(z) &:= \Pi\left(\tilde{H}-z\Id\right)\Pi - \Pi\tilde{H}\Pi^\perp\cdot \Pi^\perp (\tilde{H}-z\Id)^{-1}\Pi^\perp 
	\cdot \Pi^\perp\tilde{H}\Pi \nonumber\\
	&=\orbitalProjBLO\left(\ \tilde{H}-z\Id - \tilde{H} R_{\orbitalProjBLO^\perp \tilde{H}}(z)  \tilde{H} \right)
	\quad\textrm{ is invertible on $\orbitalSubSp$}
	\label{Sinv}\end{align}
For the application of \cref{eq:Schur's complement}, note that $B=\Pi \tilde{H} \Pi^\perp$ and $C=B^*$.
We first establish the invertibility of $D(z)=\Pi^\perp (\tilde{H}^\lambda-z\Id)\Pi^\perp $ on $\orbitalSubSp^\perp$ .        %
\begin{prop}\label{prop:the boundedness of the resolvent of the orthogonal Hamiltonian}
	For any $C>0$, there exists $\lambda_\star>0$ such that 
	for any $z\in\CC$ satisfying $\Re{z}<C$ and all $\lambda>\lambda_\star$,  the operator
	$\Pi^\perp (\tilde{H}^\lambda-z\Id)\Pi^\perp $ is invertible on $\orbitalSubSp^\perp$. Its inverse, 
	$R_{\orbitalProjBLO^\perp \tilde{H}}(z)=\Pi^\perp (\tilde{H}-z\Id)^{-1}\Pi^\perp $, satisfies the bound:
	\begin{align} 
		\norm{R_{\orbitalProjBLO^\perp \tilde{H} }(z)}
		\lesssim |\nnHopping^\lambda|, \label{eq:bound on resolvent of complement to orbitals}\end{align}
	where $\rho^\lambda$, displayed in \cref{rho-def}, is the magnetic hopping coefficient, which satisfies the exponentially small upper and lower bounds of \cref{thm:lb-hop}.
\end{prop}
\cref{prop:the boundedness of the resolvent of the orthogonal Hamiltonian} is a consequence of  the energy estimates of \cref{energy-est}; see \cref{H-Vperp}.
\medskip

We next turn to the invertibility of $S(z)$. Suppose that  $z\notin\sigma(\tbH)$, so by definition $\contToDiscBLO (\tbH - z \Id)$ is invertible on
$\orbitalSubSp$. Noting that $S(z)=\contToDiscBLO (\tbH - z \Id) +\ \left[ S(z)-\contToDiscBLO (\tbH - z \Id)\right]$,
we see that  $S(z)$ can be inverted on $\orbitalSubSp$ via a Neumann series, provided $S(z)-\contToDiscBLO (\tbH - z \Id)$ is sufficiently small in $\mathcal{B}(\orbitalSubSp)-$ norm. Furthermore, we have the basic bound which will be required below:
\begin{equation} \norm{S(z)^{-1}}\le\ \frac{\| R_{\contToDiscBLO\tbH}(z)\|}{1-\|R_{\contToDiscBLO\tbH}(z)\|\ \|S(z)-\contToDiscBLO(\tbH-z\Id)\|}.
	\label{Sinv-bd}
\end{equation}

We now show that $\| S(z)-\contToDiscBLO (\tbH - z \Id)\|_{\mathcal{B}(\orbitalSubSp)}$ tends to zero as $\lambda\to\infty$.  From \cref{Sinv} and $J^*J=\Pi$ we have that
\begin{align}
	\norm{S(z) - \contToDiscBLO (\tbH - z \Id)}_{\mathcal{B}(\orbitalSubSp)} &= \norm{\orbitalProjBLO \tilde{H} -\contToDiscBLO \tbH - \tilde{H} R_{\orbitalProjBLO^\perp \tilde{H}}(z) 
		\tilde{H}}_{\mathcal{B}(\orbitalSubSp)}\nonumber\\
	&\leq\ \norm{\orbitalProjBLO \tilde{H} -\contToDiscBLO \tbH}_{\mathcal{B}(\orbitalSubSp)}\  \nonumber\\
	&\qquad\ +\ \norm{\Pi\tilde{H}\Pi^\perp}_{\orbitalSubSp\leftarrow\orbitalSubSp^\perp}\ \norm{R_{\orbitalProjBLO^\perp \tilde{H} }(z)}_{\orbitalSubSp^\perp\leftarrow\orbitalSubSp^\perp}\ \norm{\Pi^\perp\tilde{H}\Pi}_{\orbitalSubSp^\perp\leftarrow\orbitalSubSp^\perp}  \nonumber\\ \nonumber\\
	&\leq \norm{\orbitalProjBLO \tilde{H} -\contToDiscBLO \tbH}_{\mathcal{B}(\orbitalSubSp)}\ +\ 
	\norm{\tilde{H}\orbitalProj}_{\mathcal{B}(L^2)}^2\ \norm{R_{\orbitalProjBLO^\perp \tilde{H} }(z)}_{\orbitalSubSp^\perp\leftarrow\orbitalSubSp^\perp}\nonumber\\
	&:= \textrm{Term}_1\ +\ \textrm{Term}_2. \label{eq:basic estimation of difference at the level of Hamiltonians}
\end{align}
By \cref{prop:the boundedness of the resolvent of the orthogonal Hamiltonian},
$\textrm{Term}_2\lesssim 
|\rho|\ \norm{\tilde{H} \orbitalProj}_{\mathcal{B}(L^2)}^2 = |\nnHopping|^{-1}\norm{H \orbitalProj}_{\mathcal{B}(L^2)}^2$. 
Hence, $\textrm{Term}_2\to0$
as $\lambda\to\infty$ by the following proposition, which is proved in \cref{subsec:proofs of propositions of the main theorem}. 
\begin{prop} Fix any $q<1$ and let $a$ be fixed and sufficiently large.
	Then,  
	\begin{equation}
		|\rho^\lambda|^{-q}\ \norm{H^\lambda \orbitalProj}_{\mathcal{B}(L^2)}\longrightarrow 0\ \textrm{as}\ \lambda\to\infty.\end{equation}
	\label{prop:comparing the almost eigenstates to nn hopping}
\end{prop}

That $\textrm{Term}_1\to0$ is a consequence of the following proposition, also proved in 
\cref{subsec:proofs of propositions of the main theorem}: 
\begin{prop}
	Let $a$ denote the minimal pairwise spacing of pointss in $\GG$. There exists $a_0>0$ such that for all $a>a_0$:
	$$ \norm{\Pi \tilde{H}^\lambda \Pi -  \contToDiscBLO \tbH }_{\mathcal{B}(\orbitalSubSp)} \to  0,\quad \textrm{as $\lambda\to\infty$}.  $$ 
	\label{prop:controlling ONB rotation and NN truncation}
\end{prop}
This establishes  the invertibility of $D(z)$ and $S(z)$, and therewith the invertibility of $\tilde{H}-z\Id$ 
for any fixed $z$ in the resolvent set of $\tbH $ and $\lambda$ sufficiently large. 

Finally, we turn to the proof of \cref{res-conv1}, on the convergence of $(\tilde{H}^\lambda-z\Id)^{-1}$ as $\lambda\to\infty$. Using \cref{eq:Schur's complement} we  express
$(\tilde{H}-z\Id)^{-1}$ as a bounded operator on $L^2(\RR^2)$, for $z\notin\sigma(\tbH )$ and $\lambda$ sufficiently large. For the left hand side of \cref{res-conv1}, only the top left block of
$R_{\contToDiscBLO\tbH}(z)$ is relevant. For that block, we have \begin{align*}
	\norm{S(z)^{-1}-R_{\contToDiscBLO\tbH}(z)}_{\mathcal{B}(\orbitalSubSp)} 
	&\leq \big\|S(z)^{-1}\big\|\ \norm{S(z)- \contToDiscBLO(\tbH - z \Id)}\ \norm{R_{\contToDiscBLO\tbH}(z)} \,.
\end{align*}

From \cref{Sinv-bd,eq:basic estimation of difference at the level of Hamiltonians} it follows that for $z$ outside the spectrum of $\tbH$, and $\lambda$ sufficiently large,
\begin{equation} \|S(z)^{-1}\|\le 2 \big\|R_{\contToDiscBLO\tbH}(z)\big\|.\label{Sinv-bd1}\end{equation}
Hence \cref{eq:basic estimation of difference at the level of Hamiltonians} implies this norm difference tends to zero by \cref{prop:controlling ONB rotation and NN truncation}.

The three blocks in \cref{res-conv1} other than the top left block are:
\begin{align}
	R_{\tilde{H}}(z)-\ \Pi S(z)^{-1} \Pi\ &=\ -\ \Pi S(z)^{-1} B D(z)^{-1} \Pi^\perp\nonumber\\
	&\qquad + \Pi^\perp D(z)^{-1} B^* S(z)^{-1} \Pi\ \nonumber\\
	&\qquad + \Pi^\perp\left(\ D(z)^{-1}+D(z)^{-1} B^* S(z)^{-1}BD(z)^{-1}\right)\Pi^\perp
	\label{HHTB-diff} \end{align}

We claim that the $\mathcal{B}(L^2(\RR^2))-$ norms of the three operators on the right hand side of \cref{HHTB-diff} tend to zero as $\lambda\to\infty$. Indeed, we bound the $\mathcal{B}(\orbitalSubSp)-$ norms of the first and third operators on the right hand side of \cref{HHTB-diff}. The second operator is bounded just as the first. By \cref{Sinv-bd1}, \cref{prop:comparing the almost eigenstates to nn hopping} and  \cref{prop:the boundedness of the resolvent of the orthogonal Hamiltonian}  the first operator is controlled as follows:
\begin{align*}
	\big\| \Pi S(z)^{-1} B D(z)^{-1} \Pi^\perp\big\|\ \le\ \big\| S(z)^{-1}\big\|\ \|\Pi \tilde{H}\big\|\ \big\|\ R_{\orbitalProjBLO^\perp \tilde{H} }(z)\big\| \le C |\rho|\to0
\end{align*} 
as $\lambda\to\infty$. For the third operator, we have
\begin{align*}
	&\|\Pi^\perp\left(\ D(z)^{-1}+D(z)^{-1} B^* S(z)^{-1}BD(z)^{-1}\right)\Pi^\perp\| \\
	&\quad \le \big\|D(z)^{-1}\big\|\ \| \Id + B^*S(z)^{-1}BD(z)^{-1} \big\| \\
	&\quad \le \big\|R_{\orbitalProjBLO^\perp \tilde{H} }(z)\big\|\ \left( 1  + \|\tilde{H}\Pi\|^2\ \big\|S(z)^{-1}\big\|\ \big\|R_{\orbitalProjBLO^\perp \tilde{H} }(z)\big\| \right)\\
	&\quad  \le\ C|\rho|\ \left(1+\mathcal{O}(|\rho|)\right)\ \to\ 0.
\end{align*}

This completes our outline of the proof of \cref{res-conv}. It remains to prove the various estimates used. We begin with the bounds on the resolvent which follow from energy estimates.

\section{Energy estimates and the proof of \cref{prop:the boundedness of the resolvent of the orthogonal Hamiltonian}}\label{energy-est}

Recall the decomposition $L^2(\RR^2)= \orbitalSubSp\oplus\orbitalSubSp^\perp$, where $\orbitalSubSp$ denotes
the subspace spanned by all ground state orbitals with centerings in the discrete set $\GG$. \cref{prop:the boundedness of the resolvent of the orthogonal Hamiltonian} is a consequence of the following energy (coercivity) estimate on $H^\lambda$ restricted to the range of $\Pi^\perp$.

\begin{prop}[Energy estimate on $\orbitalSubSp^\perp$]\label{H-Vperp}
	There exist positive constants $\lambda_*, C_1$ and  $C_2$ such that for all $\lambda>\lambda_\star$, the following holds. If $\psi\in H^2(\RR^2)\cap\orbitalSubSp^\perp$
	{\it i.e.} $\left\langle \varphi_n^\lambda, \psi\right\rangle=0$ for all $n\in\GG$, then 
	\begin{align}
		\left\langle \psi,H^\lambda\psi\right\rangle\ \ge\ C_1\ \|\psi\|^2\ +\ C_2\ \lambda^{-2}\ \|(P-b Ax)\psi\|^2,\quad \textrm{where}\ b=\lambda.
		\label{en-est}
	\end{align}
\end{prop}

The proof of \cref{H-Vperp}  follows the strategy in \cite{FLW17_doi:10.1002/cpa.21735},  starting with the energy estimate for the Hamiltonian for a single atom in a constant magnetic field, which is a consequence of \cref{ass:v0} (part (v3)):
For $\lambda>\lambda_\star$ sufficiently large, if $\psi\in H^2(\RR^2)$ and $\left\langle \varphi_0^\lambda,\psi\right\rangle=0$, then 
\begin{equation}
	\left\langle \psi, \left((P-b Ax)^2 +\lambda^2 v_0 -e_0^\lambda\right)\psi \right\rangle\ge c_{\rm gap}\|\psi\|^2,\qquad b=\lambda.
	\label{atom-en-est}\end{equation}
Here, $c_{\rm gap}$ is a positive constant which is independent of $\lambda$. We generalize the arguments of  \cite{FLW17_doi:10.1002/cpa.21735} in two ways: we do not assume translation invariance (and hence must deal with summability questions arising when not  working on the fundamental cell of a Bloch decomposition) and we allow for a magnetic field.

In this section, for convenience, we define
$$ \calP :=  P - b A X,\qquad b=\lambda $$
as the variable conjugate to the position in the presence of magnetic fields. In the discussion below we write $A(X)=bAX$. We also recall that $H=H^\lambda$ includes a centering about energy $e=e_0^\lambda$, and that $h_n=h^\lambda_n$ is the magnetically translated and  ground state energy-centered  Hamiltonian for a single atom at site $n\in\GG$. Hence $h_n \vf_n \equiv 0$. 

\begin{lem}\label{lem:FLW Lem 9.1} There exist positive constants $\lambda_\star$, $C'$ and $c$ such that the following holds: Choose $r$ such that
	\begin{align}0<r<\min\Set{a-\sqrt{2}r_0,\frac12 a}\label{eq:constraint on r to be able to apply FLW91}\,.\end{align}
	and assume $\lambda>\lambda_\star$. If $\psi\in\orbitalSubSp^\perp\cap H^2(\contSp)$ and  $\supp(\psi)\subseteq\bigcup_{n\in\discSp}B_r(n)$,  then
	\begin{equation}\label{en-est-loc} \ip{\psi}{\crysHamil[e]\psi} \geq c_{\rm gap}\left(1-C'\wellDepth^{1/4}\exp(-c\lambda)\right)\  \norm{\psi}^2 .\end{equation}\end{lem}
\begin{proof}
	We may write $\psi = \sum_{n\in\discSp}\psi_n$, where $\psi_n := \chi_{B_r(n)}(X)\psi$. Next, we claim that $\ip{\psi}{\crysHamil\psi} = \sum_{n\in\discSp}\ip{\psi_n}{\oneAtomHamil_n \psi_n}$. 
	Indeed, \begin{align*}\ip{\psi}{\crysHamil\psi} &= \sum_{n,m\in\discSp}\ip{\psi_n}{\crysHamil\psi_m} =\sum_{n,m\in\discSp}\ip{\psi_n}{\left(\oneAtomHamil_m + \sum_{l\in\setwo{\discSp}{m}}\oneAtomPotOp_l\right)\psi_m}\\ &=\sum_{n,m\in\discSp}\ip{\psi_n}{\oneAtomHamil_m \psi_m}\tag{$\sum_{l\in\setwo{\discSp}{m}}\oneAtomPotOp_l\psi_m=0$}\\
		&=\sum_{n\in\discSp}\ip{\psi_n}{\oneAtomHamil_n \psi_n}+\sum_{n,m\in\discSp:n\neq m}\ip{\psi_n}{\magMomOp^2 \psi_m}\, .\tag{$\oneAtomPotOp_m\psi_n=0$}
	\end{align*}
	Concerning the latter term,  $\ip{\psi_n}{\magMomOp^2 \psi_m}=0$ for $n\neq m$. This follows since $\psi_n$ and $\psi_m$ are here defined  to have disjoint support, and hence so do terms arising from the local operators, $P$, $PA(X)$, $A(X)P$ and $|A(X)|^2$  applied to them. 
	
	Let us define $\psi_n^\perp := \psi_n - \ip{\oneAtomGrndStt_n}{\psi_n}\oneAtomGrndStt_n$; while the total $\psi$ obeys $\psi\perp\oneAtomGrndStt_n$, each constituent $\psi_n$ may fail this. Using $\oneAtomHamil_n \oneAtomGrndStt_n = 0 $ we find $\oneAtomHamil_n\psi_n=\oneAtomHamil_n\psi_n^\perp$, whence $\ip{\psi}{\crysHamil\psi} = \sum_{n\in\discSp}\ip{\psi_n^\perp}{\oneAtomHamil_n \psi_n^\perp}$. But by \cref{atom-en-est} and $\psi_n^\perp \perp \oneAtomGrndStt_n$ we get $\ip{\psi_n^\perp}{\oneAtomHamil_n \psi_n^\perp}\geq c_{\rm gap}\norm{\psi_n^\perp}^2$ 
	and therefore
	\begin{equation}
		\ip{\psi}{\crysHamil\psi} \geq c_{\rm gap}\sum_{n\in\discSp}\norm{\psi_n^\perp}^2
		= c_{\rm gap}\left( \|\psi\|^2- \sum_{n\in\discSp}|\ip{\oneAtomGrndStt_n}{\psi_n}|^2\right)\,.
		\label{en-1}\end{equation}
	To conclude the proof of \cref{lem:FLW Lem 9.1} we bound the sum in 
	\cref{en-1} by $\ee^{-c\lambda}\|\psi\|^2$, for some $c>0$:
	\begin{align}\sum_{n\in\discSp}|\ip{\oneAtomGrndStt_n}{\psi_n}|^2 &= \sum_{n\in\discSp}|\ip{\oneAtomGrndStt_n}{\sum_{l\neq n}\psi_l}|^2\tag{$\psi_n = \psi-\sum_{l\neq n}\psi_l$ and $\psi\perp\oneAtomGrndStt_n$}\nonumber\\
		&= \sum_{n\in\discSp}|\ip{\oneAtomGrndStt_n}{Q_n\psi}|^2\tag{$Q_n := \sum_{l\neq n}\chi_{B_r(l)}(X)$}\nonumber\\
		&= \sum_{n\in\discSp}|\ip{Q_n\sqrt{|\oneAtomGrndStt_n|}}{Q_n\sqrt{|\oneAtomGrndStt_n|}\psi}|^2\nonumber \\
		&\leq \sum_{n\in\discSp}\norm{Q_n\sqrt{|\oneAtomGrndStt_n|}}^2\norm{Q_n\sqrt{|\oneAtomGrndStt_n|}\psi}^2\,.\label{QQ}
	\end{align}
	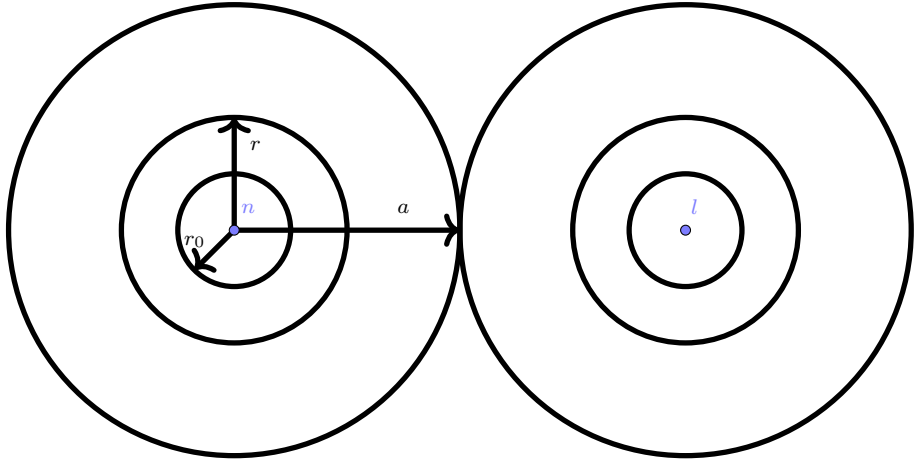
\begin{figure}
	    \centering
	    \begin{tikzpicture}[x=1cm,y=1cm,scale=0.75]
\draw [line width=2pt] (0,0) circle (1cm);
\draw [line width=2pt] (0,0) circle (2cm);
\draw [line width=2pt] (0,0) circle (4cm);
\draw [line width=2pt] (8,0) circle (1cm);
\draw [line width=2pt] (8,0) circle (2cm);
\draw [line width=2pt] (8,0) circle (4cm);
\draw [->,line width=2pt] (0,0) -- (4,0);
\draw [->,line width=2pt] (0,0) -- (0,2);
\draw [->,line width=2pt] (0,0) -- (-0.7,-0.7);

\begin{scriptsize}
\draw [fill=xdxdff] (0,0) circle (2.5pt);
\draw[color=xdxdff] (0.25,0.38) node {$n$};
\draw[color=black] (3,0.38) node {$a$};
\draw[color=black] (0.38,1.5) node {$r$};
\draw[color=black] (-0.7,-0.2) node {$r_0$};
\draw [fill=xdxdff] (8,0) circle (2.5pt);
\draw[color=xdxdff] (8.16,0.42) node {$l$};

\end{scriptsize}
\end{tikzpicture}
	    \caption{Proportions of $r_0,r$ and $a$.}
	    \label{fig:proportions_for_r0_r_and_a}
	\end{figure}
	Note that $\supp(Q_n)\subset B_{a-r}^c(n)$; see \cref{fig:proportions_for_r0_r_and_a}. Hence,  $Q_n\leq \chi_{B_{a-r}(n)^c}(X)$. 
	To estimate the first factor in \cref{QQ}, we have by \cref{gs-bound} (specifically, its consequence \cref{lem:integral of atomic ground state away from origin also decays}) that
	$$\norm{Q_n\sqrt{|\oneAtomGrndStt_n|}}^2\leq \int_{|y|\ge a-r}|\varphi_n(y)|\ dy\le C\wellDepth^{-1/2}\exp(-\frac{1}{4}\wellDepth((a-r)^2-r_0^2)),
	\quad n\in\GG.$$
	To bound the second factor in \cref{QQ}, we again again use that $Q_n\leq \chi_{B_{a-r}(n)^c}(X)$ and \cref{thm:Gaussian decay of ground state}:
	\begin{align*}\sum_n \norm{Q_n\sqrt{|\oneAtomGrndStt_n|}\psi}^2&\leq \sum_n \int_{x\in\contSp} \chi_{B_{a-r}(n)^c}(x)\vf_n(x)|\psi(x)|^2\dif{x}  \\
		&\leq  C \sqrt{\lambda} \exp(+\frac{1}{4}\wellDepth r_0^2)\left(\sup_{x\in\RR^2}\sum_{n\in\GG}\exp(-\frac14\wellDepth \normEuc{x-n}^2)\right)\int_{x\in\contSp}|\psi(x)|^2\dif{x}\\
		&\leq C' \sqrt{\lambda}\exp(+\frac{1}{4}\wellDepth r_0^2)\norm{\psi}^2\,.
	\end{align*}
	The last inequality follows from the summability hypothesis 
	\cref{eq:Gaussians are summable in discrete space} in \cref{ass:min-sp}.
	Combining the bounds on both factors, 
	we get $$ \sum_{n\in\discSp}|\ip{\oneAtomGrndStt_n}{\psi_n}|^2 \leq C' \exp(-\frac{1}{4}\wellDepth((a-r)^2-2 r_0^2))\norm{\psi}^2\lesssim \ee^{-c\lambda},\ c>0, $$
	provided  $(a-r)^2-2 r_0^2 > 0$, implying the constraint \cref{eq:constraint on r to be able to apply FLW91},  
	as asserted. This completes the proof of \cref{lem:FLW Lem 9.1}.
\end{proof}

In the following we let 
$\Theta\in C^\infty(\contSp)$, be such that $0\le \Theta\le 1$ and $\supp(\Theta)\subseteq\bigcup_{n\in\discSp}B_r(n)$, where $r>0$ satisfies the constraint \cref{eq:constraint on r to be able to apply FLW91}
of \cref{lem:FLW Lem 9.1}. We also assume $\supp(1- \Theta)\subseteq\bigcap_{n\in\GG}B_{2r_0}(n)^c$. For example we may take $2r_0\ll r \ll \frac12a$ and $\Theta\equiv1$  on $\bigcup_{n\in\discSp}B_{2r_0}(n)$. 

\begin{prop}[Localized energy estimate]
	There exists positive constants $\lambda_\star$, $C$ and $c$ such that the following holds: For any $\psi\in\orbitalSubSp^\perp\cap H^2(\contSp)$,  \begin{align} \ip{\Theta\psi}{H\Theta\psi} \geq C\  \norm{\Theta\psi}^2 - C\ \exp(-c\wellDepth)\ \norm{\psi}^2\ . \label{eq:localized energy estimate}\end{align}\label{lem:flw9.3}\end{prop}
\begin{proof}
	While $\orbitalProj\psi = 0$ by hypothesis,  $\orbitalProj\Theta\psi$ may be nonzero. We introduce a correction to $\psi$:  $$ \eta := \Theta\cdot(\psi-u)\,. $$
	where $u\in\orbitalSubSp$ is to be chosen to that $\orbitalProj \eta = 0$. It follows that $u$ must satisfy
	$\orbitalProjBLO \Theta u =\orbitalProj \Theta \orbitalProj^\perp \psi.$
	We claim that for $\lambda$ sufficiently large  $\left.\orbitalProjBLO \Theta\right|_\orbitalSubSp$ is invertible,  and hence we may write
	\begin{equation}
		u := (\left.\orbitalProjBLO \Theta\right|_\orbitalSubSp)^{-1}\orbitalProj \Theta \orbitalProj^\perp \psi.\label{u-def}
	\end{equation}
	To check the invertibility of 
	$\orbitalProjBLO \Theta\Big|_\orbitalSubSp$, note that:
	\begin{align}
		\orbitalProjBLO \Theta\Big|_\orbitalSubSp &= \Pi\Theta\Pi\Big|_\orbitalSubSp =  \Id\Big|_\orbitalSubSp +\Pi\left(\Theta-\Id\right)\Pi\Big|_\orbitalSubSp.
		\label{PiTh}
	\end{align}
	It suffices then to show that $\|\Pi\left(\Theta-\Id\right)\Pi\|$ is small for $\lambda$ large. This follows since 
	\begin{equation} \|\left(\Theta-\Id\right)\Pi\|\le C \ee^{-c\lambda},
		\label{Th-I}\end{equation}
	for some positive constants $C$ and $c$; indeed this bound is proved in \cref{lem:orbitals away from atomic centers} of the Appendix. Thus,
	$\norm{u} \leq \norm{(\left.\orbitalProjBLO \Theta\right|_\orbitalSubSp)^{-1}}\norm{\orbitalProj \Theta \orbitalProj^\perp }\norm{\psi}$. 
	Furthermore, since $\Pi\Theta\Pi^\perp=\Pi(\Theta-\Id)+\Pi(\Id-\Theta)\Pi$, we have from \cref{Th-I} that $\|\Pi\Theta\Pi^\perp\|\lesssim  \ee^{-c\lambda}$ and therefore
	\begin{equation}
		\|u\|\le C\ee^{-c\lambda}\|\psi\|.\label{u-bd}
	\end{equation}

	We now turn to bounding $\|\eta\|$ from below:
	\begin{align*}
		\|\eta\|^2 &= \|\Theta(\psi-u)\|^2 \ge \left(\|\Theta\psi\|-\|\Theta u\|\right)^2 \ge \frac12\|\Theta\psi\| - C\|\Theta u\|^2\\
		&\ge \frac12\|\Theta\psi\| - C \ee^{-c\lambda} \|\psi\|^2,
	\end{align*}
	where we have used \cref{u-bd}.

	Since $\supp(\eta)\subseteq\supp(\Theta)\subset\bigcup_{n\in\discSp}B_r(n)$, where $r>0$ satisfies the constraint \cref{eq:constraint on r to be able to apply FLW91}, we have that \cref{lem:FLW Lem 9.1} applies to  $\eta\in\orbitalSubSp^\perp$. This implies
	\begin{equation}
		\ip{\eta}{\crysHamil\eta} \geq C \norm{\eta}^2 \geq C \norm{\Theta\psi}^2 - C \exp(-c \wellDepth)\norm{\psi}^2\,. \label{Hetaeta}
	\end{equation}
	
	We next claim that  
	\begin{equation}
		\ip{\eta}{\crysHamil\eta}\leq \ip{\Theta\psi}{\crysHamil\Theta\psi} + C \exp(-c \wellDepth)\norm{\psi}^2,
		\label{Heta-up}
	\end{equation}
	which together with the lower bound \cref{Hetaeta} yields the localized energy estimate \cref{eq:localized energy estimate}. To prove \cref{Heta-up} note first that:
	\begin{align}
		\ip{\eta}{\crysHamil\eta} &= \ip{\Theta\psi}{\crysHamil\Theta\psi} + \ip{\Theta u}{\crysHamil\Theta u} - 2 \Re{\ip{\Theta\psi}{\crysHamil\Theta u}}\,.\label{eq:main equation to relate eta with theta psi}
	\end{align}
	From \cref{Heta-up} we see that it suffices to prove that the latter two terms in \cref{Heta-up} are $\lesssim \ee^{-c\lambda} \|\psi\|^2$.
	A simple calculation gives  
	\begin{align} H\Theta u = \Theta H u + [H,\Theta] u 
		= \Theta H u + [\calP^2,\Theta ] u = \Theta H u - (\Delta \Theta) u +2(-\ii \nabla \Theta) \cdot \calP u \label{eq:magnetic integration by parts}\end{align}
	Therefore, using that $u\in\textrm{Range}(\Pi)$, $\|H\Pi\|\lesssim 1$ (\cref{prop:comparing the almost eigenstates to nn hopping})  and the properties of $\Theta$, we have:
	\begin{align*}
		|\ip{\Theta u}{\crysHamil\Theta u}| &\leq |\ip{\Theta^2 u}{H u}| + |\ip{\Theta u}{(\Delta \Theta)u}|
		+ 2 \sum_{j=1,2}|\ip{(\partial_j\Theta)\Theta u}{\calP_j u}| \\ 
		&\leq \norm{\Theta}^2\norm{u}^2\norm{H\Pi}+\norm{\Theta}\norm{\Delta\Theta}\norm{u}^2+2\sum_{j=1,2}\norm{(\partial_j\Theta)\Theta}\norm{u}\norm{\calP_j u}	\\
		&\leq C \exp(-c\wellDepth)\ \|\psi\|^2,
	\end{align*}
	where the final inequality is a consequence of \cref{lem:magnetic sobolev bound}
	in the Appendix.
	
	It remains for us to bound $\ip{\Theta\psi}{\crysHamil\Theta u}$, the last term in \cref{eq:main equation to relate eta with theta psi}. We expand this expression using \cref{eq:magnetic integration by parts} and bound all terms as in the previous bound.
	For example, using $\norm{\calP u}\leq\norm{\calP \Pi}\norm{u}$ and \cref{lem:magnetic sobolev bound} again, we have
	\begin{align*}
		|\ip{(\partial_j\Theta)\Theta \psi}{\calP_j u}| &\leq \norm{\Theta}\norm{\psi}\norm{(\nabla\Theta)\cdot\calP u}\\
		&\le \|\Theta\| \|\nabla\Theta\| \|\psi\| \|\calP u\| \lesssim \|\psi\| \lambda^{power}\|u\|\lesssim \ee^{-c\lambda}\|\psi\|^2.
	\end{align*}
\end{proof}

\begin{prop}[Global energy estimate]\label{gl-en-est} There exists $\lambda_\star>0$ such that for all $\lambda>\lambda_\star$ the following holds: \\
	Let $\psi\in\orbitalSubSp^\perp\cap H^2(\contSp)$, {\it i.e.} $\left\langle \varphi^\lambda_n,\psi\right\rangle=0$ for all $n\in\GG$.  Then $$ \ip{\psi}{H^\lambda\psi} \geq C \norm{\psi}^2 + C\frac{1}{\wellDepth^2}\norm{\calP \psi}^2\,. $$
\end{prop}
\begin{proof}
	Let $\Theta$ be given as in \cref{lem:flw9.3} and define $\Sigma := \sqrt{1-\Theta^2}$ so that we have a smooth partition of unity $ \Theta^2 + \Sigma^2 = 1$; $\Theta$ and $\Sigma$ localize, respectively, near and far from the set of atomic centers.
	By the localized energy estimate (\cref{lem:flw9.3}):
	$$ \ip{\Theta\psi}{H\Theta\psi} \geq C \norm{\Theta\psi}^2 - C \exp(-c\wellDepth)\norm{\psi}^2\,. $$
	On the other hand, far from the atomic centers we have:
	\begin{align*}
		\ip{\Sigma\psi}{H\Sigma\psi} &= \ip{\Sigma\psi}{(\calP^2+V-e_0^\lambda\Id)\Sigma\psi} \\
		&=\ip{\Sigma\psi}{(\calP^2-e_0^\lambda\Id)\Sigma\psi}\tag{$\supp(\Sigma) \cap \supp(V) = \varnothing$} \\
		&=\norm{\calP\Sigma\psi}^2-e\norm{\Sigma\psi}^2 \geq -e_0^\lambda\norm{\Sigma\psi}^2\\
		&\geq C\wellDepth^2\norm{\Sigma\psi}^2\,.\tag{$e_0^\lambda \leq -C \wellDepth^2$; see \cref{C-gs}}
	\end{align*}
	Summing the two previous bounds
	yields
	\begin{align}\ip{\Theta\psi}{H\Theta\psi} + \ip{\Sigma\psi}{H\Sigma\psi} & \geq C \norm{\Theta\psi}^2  + C\wellDepth^2\norm{\Sigma\psi}^2 -  C \exp(-c\wellDepth)\norm{\psi}^2\nonumber\\
		&\geq C' \norm{\psi}^2  + C''\wellDepth^2\norm{\Sigma\psi}^2 -  C \exp(-c\wellDepth)\norm{\psi}^2,
		\label{sum-en}\end{align}
	where we have used that $\Theta^2=1-\Sigma^2$.

	We will next make use of the following
	\begin{prop}[magnetic integration by parts]\label{prop:the magnetic IMS localization formula} If $\Theta:\contSp\to\RR$ is smooth and $\psi\in H^2(\contSp)$ then we have 
		\begin{align}
			\ip{\Theta\psi}{\calP^2\Theta\psi} = \Re\left\{\ip{\Theta^2\psi}{\calP^2\psi}\right\}+\ip{\psi}{\normEuc{\nabla\Theta}^2\psi}\,.
			\label{eq:the magnetic IMS localization formula}
		\end{align}
	\end{prop}
	\begin{proof}
		The proof in \cite[Section 9.1]{FLW17_doi:10.1002/cpa.21735}, for the non-magnetic case, applies here with $k$ there replaced by $A(X)$ here. The only difference being that $A(X)=bAX$ is now a self-adjoint operator rather than scalar multiplication. $A(X)$ still commutes with $P$ in the present context, which is constant magnetic field in the symmetric gauge, so the same proof indeed goes through.
	\end{proof}
	We study $\ip{\Theta\psi}{H\Theta\psi} + \ip{\Sigma\psi}{H\Sigma\psi}$, the left hand side of \cref{sum-en}.
	Applying \cref{prop:the magnetic IMS localization formula} to each term we obtain,
	using that $\Theta^2+\Sigma^2=1$, that
	\begin{align*} \ip{\Theta\psi}{H\Theta\psi} + \ip{\Sigma\psi}{H\Sigma\psi} &= \Re\left\{\ip{\Theta^2\psi}{H\psi}\right\}+\ip{\psi}{\normEuc{\nabla\Theta}^2\psi} + \Re\left\{\ip{\Sigma^2\psi}{H\psi}\right\}+\ip{\psi}{\normEuc{\nabla\Sigma}^2\psi} \\
		&= \ip{\psi}{H\psi}+\ip{\psi}{(\normEuc{\nabla\Theta}^2+\normEuc{\nabla\Sigma}^2)\psi}\ .\end{align*}
	Here, $\Xi:=\normEuc{\nabla\Theta}^2+\normEuc{\nabla\Sigma}^2$ 
	satisfies as $\|\Xi\|_\infty\le C'''$ since
	$\Theta$ is smooth and $0\le\Theta\le1$. Therefore, 
	\begin{align}
		\ip{\psi}{H\psi} &= \ip{\Theta\psi}{H\Theta\psi} + \ip{\Sigma\psi}{H\Sigma\psi} - \ip{\psi}{\Xi\psi} \nonumber\\
		&\geq C\norm{\psi}^2+C'\wellDepth^2\norm{\Sigma\psi}^2-C \exp(-c\wellDepth)\norm{\psi}^2- \ip{\psi}{\Xi\psi}\label{eq:equation in FLW94 before dichotomy}\,.
	\end{align}
	
	Fix $\delta>0$, independent of $\wellDepth$, to be specified below. For such a choice, we may consider the following two cases. 
	\begin{enumerate}
		\item[Case 1] If $\psi$ is such that $\norm{\Sigma\psi}\geq \delta\norm{\psi}^2$, then using $\norm{\Xi}<C'''$ we have, for $\wellDepth$ sufficiently large,  that
		\begin{align*} \ip{\psi}{H\psi}&\geq C\norm{\psi}^2+C'\wellDepth^2\delta^2\norm{\psi}^2-C \exp(-c\wellDepth)\norm{\psi}^2- C''\norm{\psi}^2 \geq C'''\norm{\psi}^2\ .\end{align*}
		\item[Case 2] Conversely, suppose that $\norm{\Sigma\psi}< \delta\norm{\psi}^2$. Now introduce a second partition of unity, $\tilde{\Theta}^2+\tilde{\Sigma}^2=1$, which satisfies the constraints stated just above \cref{lem:flw9.3} and which we further constrain below.
		We first assert the analogue of \cref{eq:equation in FLW94 before dichotomy} 
		for the second partition of unity:
		\begin{align}
			\ip{\psi}{H\psi} &\geq C\norm{\psi}^2+C'\wellDepth^2\norm{\tilde{\Sigma}\psi}^2-C \exp(-c\wellDepth)\norm{\psi}^2- \ip{\psi}{\tilde{\Xi}\psi}
			\label{en-2}	\end{align}
		We next constrain $\Theta$ and $\tilde{\Theta}$ further such that for any $\xi\in\contHilSp$:
		\begin{align} \ip{\xi}{\tilde{\Xi}\xi} \leq C \norm{\Sigma\xi}^2\,. \label{eq:relationship between two partitions of unity}
		\end{align}
		A way to achieve this, is to take $\Theta(|x|),\tilde{\Theta}(|x|):\left[0,\infty\right)\to[0,1]$ to be radial functions such that $\supp(\Theta')$ and $\supp(\tilde{\Theta}')$ are compact sets, and so that $\supp(\tilde{\Theta}')$ is an interval disjoint from and \underline{above} $\supp(\Theta')$ (see \cref{fig:two partitions of unity}). Since $0\le\Sigma\le1$ and $\Sigma\equiv1$ on $\supp(\tilde\Xi)$, we have
		\[
		\ip{\xi}{\tilde{\Xi}\xi} \le C\int_{\supp(\tilde\Xi)} |\xi|^2 \le \int_{\RR^2} |\Sigma\xi|^2.\]
		
		Using \cref{eq:relationship between two partitions of unity} and $\norm{\Sigma\psi}< \delta\norm{\psi}^2$ in \cref{en-2}, we  obtain,
		by taking $\delta>0$ sufficiently small:
		\begin{align*}
			\ip{\psi}{H\psi} &\geq C\norm{\psi}^2+C'\wellDepth^2\norm{\tilde{\Sigma}\psi}^2-C \exp(-c\wellDepth)\norm{\psi}^2- C \norm{\Sigma\psi}^2 \\
			&\geq C\norm{\psi}^2+C'\wellDepth^2\norm{\tilde{\Sigma}\psi}^2-C \exp(-c\wellDepth)\norm{\psi}^2- C \delta\norm{\psi}^2\\
			&\geq C(1-\ee^{-c\wellDepth}-\delta)\norm{\psi}^2
			\geq C''''\norm{\psi}^2\,.
		\end{align*}
	\end{enumerate}
	\begin{figure}
	    \centering
	    \begin{tikzpicture}
                \draw[very thick,line width=0.05cm,->] (0,0) -- (8,0);
                \node[below] at (0.1,-0.1) {$0$};
                \node[left] at (-0.1,3) {$\Theta,\tilde{\Theta}$};
                \node[below] at (7.9,-0.1) {$\normEuc{x}$};
                \draw[very thick,line width=0.05cm,->] (0,0) -- (0,3);
                \node[left] at (-0.1,2) {$1$};
                \draw[black,very thick,line width=0.05cm] (0,2) -- (2.1,2);
                \draw[very thick, line width=0.05cm,black, domain=2.01:3, smooth, variable=\x] plot ({\x}, {2 * exp(-1/(1-(\x-2))) / (exp(-1/(1-(\x-2)))+exp(-1/((\x-2))))});
                \draw[black,very thick,line width=0.05cm] (3,0) -- (7.8,0);
                
                \draw[dashed,very thick,line width=0.05cm] (0,2.05) -- (5.1,2.05);
                \draw[dashed,very thick, line width=0.05cm, domain=5.01:6, smooth, variable=\x] plot ({\x}, {0.05 +2 * exp(-1/(1-(\x-5))) / (exp(-1/(1-(\x-5)))+exp(-1/((\x-5))))});
                \draw[dashed,very thick,line width=0.05cm] (6,0.05) -- (7.8,0.05);

                
        \end{tikzpicture}  
	    \caption{Our two partitions of unity, the solid one being $\Theta$ and the dashed one $\tilde{\Theta}$, which have disjoint support of their gradients. }
	    \label{fig:two partitions of unity}
	\end{figure}
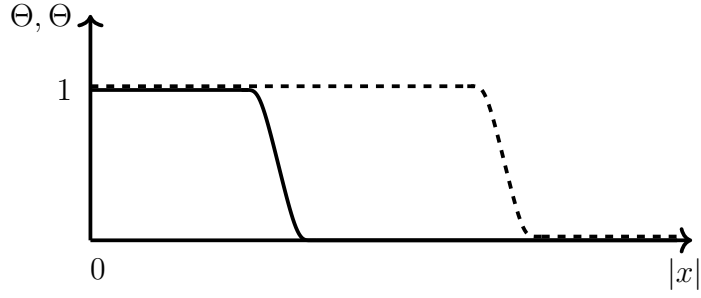
	Finally, using the bound $|\lambda^2V(x)-e_0^\lambda|<C\wellDepth^2$ we have
	\begin{align*}\ip{\psi}{\calP^2\psi} &\leq \ip{\psi}{H\psi}-\ip{\psi}{(\lambda^2 V-e_0^\lambda\Id)\psi}\\
		&\leq \ip{\psi}{H\psi} + C \wellDepth^2\norm{\psi}^2\lesssim\left(1+\lambda^2\right) \ip{\psi}{H\psi}.
	\end{align*}
	and hence  $\frac{1}{\wellDepth^2}\norm{\calP \psi}^2 \lesssim \ip{\psi}{H\psi}\,$.
	This completes the proof of \cref{gl-en-est}.
\end{proof}

\section{Proofs of topological equivalence}\label{sec:topological equivalence}

\subsection{Proof of equality of bulk indices; \cref{thm:topological equivalence for bulk geometries}}
For the statement and discussion of \cref{thm:topological equivalence for bulk geometries} see \cref{eq-bulk}.
\begin{proof}
	Let $\mu\in\RR$ such that $\mu\notin\sigma(\tbH)$. Then, the Fermi projection, defined as
	$$  P^{\rm TB}:=\chi_{(-\infty,\mu)}(\tbH),\quad
	P^{{\rm TB},\perp}=\Id- P^{\rm TB},$$ 
	has off-diagonal decaying matrix elements in the position basis (via the Combes-Thomas estimate \cite[Section 10.3]{AizenmanWarzel2016} and the Riesz formula) and hence its associated topological index, given either by the Kubo formula \cref{eq:kubo formula} or the index formula \cref{eq:Hall conductivity as Fredholm index}, is well-defined. We shall use the latter as our starting point.

	Define $$ u:\CC\setminus\Set{0}\to\mathbb{S}^1\subseteq\CC,\qquad u(z) := \exp(\ii\arg(z)) \,.$$ 
	Then the discrete bulk topological invariant is given by $$ \calN^\mathrm{TB} \equiv \findex_{\discHilSp}\left(P^{TB}u(X^\mathrm{TB})P^{TB}+
	P^{TB,\perp}\right) $$ 
	where we interpret $X^\mathrm{TB}\equiv X^\mathrm{TB}_1+\ii X^\mathrm{TB}_2\in\CC$ as a complex-valued multiplication operator. Next, using an approximation of the identity (e.g. $|\vf(y)|^2=|\vf^\lambda_0(y)|^2,\ \lambda>0$), we define $$\CC\setminus\Set{0}\to\mathbb{S}^1\subseteq\CC $$  
	$$ z\mapsto \xi(z) := (|\vf|^2 \star u)(z) \equiv \int_{y\in\RR^2}|\vf(y-z)|^2 u(y)\dif{y}\ : \CC\setminus\Set{0}\to\mathbb{S}^1\subseteq\CC. $$ 
	We work with $\xi$ in order to extrapolate from $u(X^{\rm TB})$ on $\discHilSp$ to a continuum operator on $\contHilSp$. However, the former is defined via the orthonormalized orbitals, $\tilde\vf_n$, which are not precisely localized about the $n\in\GG$. Dealing with this requires several steps.
	
{\it Step 1:}  $$ \calN^\mathrm{TB} = \findex_{\discHilSp}\left(P^{TB}\xi(X^\mathrm{TB})P^{TB}+P^{TB,\perp}\right)\,. $$
The Fredholm index is stable under compact perturbations, and the compact operators form an ideal within the space of bounded linear operators on a Hilbert space. Hence,  it is sufficient to show that $$ u(X^\mathrm{TB})-\xi(X^\mathrm{TB})\quad \textrm{is compact}.$$
Now a multiplication operator on $l^2(\GG)$, $\{x_n\}\mapsto \{m(n)x_n\}$,  is compact if and only $m(n)\to0$ as $n\to\infty$, so we need to show $ |u(n)-\xi(n)|\to0 $ as $|n|\to\infty$.
 To prove this note: $$ u(n)-\xi(n) = u(n)\Big(1-\int_{y\in\RR^2}|\vf(y)|^2\frac{u(n+y)}{u(n)}\dif{y}\Big)\,. $$ For each $y$,  $u(n+y) / u(n) \to 1$ as $|n|\to \infty$. Hence,  by the dominated convergence theorem:
$$ \lim_n \int_{y\in\RR^2}|\vf(y)|^2\frac{u(n+y)}{u(n)}\dif{y} = \int |\vf|^2 = 1 $$ and we get the result.
	
{\it Step 2:} Recall that we use $X$ to denote the position operator on $\contHilSp$. We claim 
	\begin{equation} \findex_{\discHilSp}\left(P^{TB}\xi(X^\mathrm{TB})P^\perp+ P^{TB,\perp} \right) = \findex_{\discHilSp}\left( P^{TB}  M^{-1}Ju(X)J^\ast M^{-1} P^{TB}+P^{TB,\perp}\right)\,. \label{ind-1}
	\end{equation}
	
	To demonstrate \cref{ind-1}, we use that the Fredholm index is stable under norm continuous perturbations. Since $P^{TB}$ is a projection, it suffices to show that 
	\begin{align} \norm{\xi(X^\mathrm{TB})- M^{-1}Ju(X)J^\ast M^{-1}}\to 0\,. \label{eq:xi_TB and MuM}\end{align}
	Let $n,m\in\discSp$. Writing $u(X)=\xi(n)+(u(X)-\xi(n))$
	 and applying \cref{rotation} yields:
  $$ ( M^{-1}Ju(X)J^\ast M^{-1})_{n,m} \equiv \ip{\tilde{\vf}_n}{ J M^{-1}Ju(X)J^\ast M^{-1}J \tilde{\vf}_m} = \delta_{n,m}\xi(n) + (1-\delta_{n,m})\ip{\vf_n}{u(X)\vf_m}\,. $$ Hence by the inequality of \cref{lem:hsy}, $$ \norm{\xi(X^\mathrm{TB})- M^{-1}Ju(X)J^\ast M^{-1}} \leq \max\left\{\sup_{n\in\discSp}\sum_{m\in\discSp\setminus\Set{n}}|\ip{\vf_n}{u(X)\vf_m}|,\quad n\leftrightarrow m\right\}. $$ The Gaussian decay of $\vf=\vf^\lambda$ now implies that this is summable \emph{and} decays to zero exponentially fast in $\lambda$.

{\it Step 3}: Since $M-\Id$ tends to zero in the operator norm as $\lambda\to\infty$ (\cref{Gram-prop}), 
$$ \findex_{\discHilSp}\left(P^{TB}  M^{-1}Ju(X)J^\ast M^{-1} P^{TB} +P^{TB,\perp}\right) = \findex_{\discHilSp}\left(P^{TB} J u(X)J^\ast P^{TB}+ P^{TB,\perp}\right) .$$

	{\it Step 4}:
	Now define 
	\[Q:=\contToDiscBLO P^{TB},\]
	an operator on the orbital subspace $\orbitalSubSp$. We  claim 
	$$ \findex_{\discHilSp}(P^{TB} Ju(X)J^\ast P^{TB}+P^{TB,\perp}) = \findex_{\orbitalSubSp}(Q u(X) Q+Q^\perp)=\findex_{\contHilSp}(Q u(X) Q+Q^\perp)\,. $$ Here we use that $J$ is a partial isometry which provides the necessary isomorphisms between the kernels of operators on $\discHilSp$ and on $\orbitalSubSp$.
	The second equality stems from $Q^\perp$ being the identity on $\orbitalSubSp^\perp$; it makes no difference to consider $Q$ also as an operator on $\contHilSp$. 
	We proceed to the final

\noindent{\it Step 5}: 
	\cref{res-conv} and the Riesz projection formula
	\cref{riesz} imply that  as $\lambda\to\infty$
	$$ P^\lambda:=\chi_{(-\infty, \mu)}\Big((\rho^{\lambda})^{-1} H^\lambda\Big) \to  Q $$ in norm. This 
	implies that as $\lambda\to\infty$, $$ P^\lambda u(X) P^\lambda + P^{\lambda,\perp} \to Q u(X) Q+Q^\perp $$ in norm. In other words, for all $\lambda$ sufficiently large
	$$ \calN^{\mathrm{TB}}  = \findex_{\contHilSp }\left( P^\lambda u(X) P^\lambda + P^{\lambda,\perp} \right)\,.$$
	The proof of \cref{thm:topological equivalence for bulk geometries} is now complete.
\end{proof}

\subsection{Proof of equality of edge indices; \cref{thm:edge continuum discrete correspondence}}
For the statement and discussion of \cref{thm:edge continuum discrete correspondence}, see \cref{sec:edge}.
\begin{proof} 
Let $\tbH$ and the smooth function $g:\RR\to[0,1]$ be as in \cref{def:edge Hamiltonian}. We let  $\Lambda_1=\Lambda(X_1)$ and $\Lambda_1^{\mathrm{TB}}=\Lambda(X^{\rm TB}_1)$ as defined in \cref{sec:edge}. We recall that, 
 for the present purposes, without loss of generality (as such a change does not affect the index), $\Lambda(s)$ can be taken to be the Heaviside step function and so 
 $\Lambda_1$ and $\Lambda_1^{\mathrm{TB}}$ are \emph{projection} operators to the subspace of $L^2(\RR^2)$ consisting of functions supported in right half-plane of the continuum ($x_1\ge0$), respectively discrete ($n_1\ge0$), spaces. Furthermore, we may assume without loss of generality that $g$ obeys the normalization condition \begin{align}g(0)\in\Set{0,1}\label{eq:normalization condition for g}\,.\end{align} Indeed, should that not be the case (if e.g. $0\in\supp(g^2-g)$ and $g(0)\in(0,1)$), we may use another smooth function $f$ such that $\supp(f^2-f)\subseteq\supp(g^2-g)$ and $f$ obeys the normalization condition. The topological invariants defined with $f$ and with $g$ coincide (by homotopy). At the end of the proof we revert back from $f$ to $g$, again by homotopy invariance.

Now let 
\begin{align}
 W^{\rm TB} &\equiv \exp(-2\pi\ii g(\tbH)),\quad {\rm and}\label{UTB}\\
  \hat{\calN}^{\mathrm{TB}} &\equiv \findex_{\ell^2}\left(\Lambda^{\mathrm{TB}}_1 W^{\rm TB} \Lambda^{\mathrm{TB}}_1+\Lambda^{\mathrm{TB},\perp}_1\right).\label{Nhat}
  \end{align}
	Recall the scaled continuum Hamiltonian $\tilde{H}^\lambda\equiv H /\nnHopping^\lambda$ and introduce
	\begin{equation}
	    W^\lambda= \exp\left(-2\pi\ii g\left(\tilde{H}^\lambda\right)\right).
	    \label{Ulam}
	\end{equation}
	We shall prove the following two assertions:
	\begin{enumerate}
\item 	$\Lambda_{1} W^\lambda \Lambda_{1}+ \Lambda^\perp_{1} $
	is a Fredholm operator on $L^2(\RR^2)$ and hence 
	\[\hat{\mathcal{N}}^\lambda\equiv \findex_{L^{2}}\Lambda_{1} W^\lambda \Lambda_{1}+ \Lambda^\perp_{1}\quad \textrm{is well-defined}.\]
	\item For all $\lambda$ sufficiently large, $\hat{\mathcal{N}}^\lambda$ is independent of $\lambda$ with
	\[ \hat{\mathcal{N}}^\lambda =\hat{\calN}^{\mathrm{TB}}.\]
	\end{enumerate}

	Our starting point is the expression for $\hat{\calN}^{\mathrm{TB}}$. We shall deform, through  several operator-norm continuous perturbations as well as compact perturbations, the tight-binding operator in the expression for  $\hat{\calN}^{\mathrm{TB}}$ to the continuum operator in the expression for  $\hat{\calN}^\lambda$, all while keeping the Fredholm index fixed.
	We outline the three main steps:
	\begin{enumerate}
	    \item[Step 1:] Extrapolate from  $\Lambda^{\mathrm{TB}}_1$ to $\Lambda_{1}$;  the atomic orbitals, $\{\vf_n^\lambda\}$, are broadened about points in $\GG$. The arguments here follow a similar plan to those used above in the bulk equivalence proof, with the replacement of the flux insertion function of position operators by switch functions of position operators.
	    \item[Step 2:] Control of a commutator arising from algebraic manipulations which relate the indices of operators on $L^2(\RR^2)$  and indices  of operators on $l^2(\GG)$.  This step crucially uses the fact that $\Pi$ (orthogonal projection to the orbital subspace) has off-diagonal decaying integral kernel.
	    \item[Step 3:] We apply the main theorem, \cref{res-conv}, to relate the resolvent convergence: $\tilde{H}^\lambda\equiv H^\lambda / \nnHopping^\lambda\to \tbH$ to the resolvent convergence of analytic functions of these operators. By the stability of the Fredholm index, this is sufficient to conclude the indices $\hat{\mathcal{N}}$ and  $\hat{\calN}^{\mathrm{TB}}$, which are built on \emph{smooth} functions of these operators, agree.
	\end{enumerate}

	\paragraph{Step 1}
	We follow similar ideas to those used in the bulk index case. Recall that $\Lambda(s)$ was taken to be the Heaviside function and hence  $\Lambda^{\mathrm{TB}}_1$ is a projection. Thus we may  just as well rewrite $ \hat{\calN}^\mathrm{TB}$ as 
	$$ \hat{\calN}^\mathrm{TB} = \findex_{\ell^2}\left(\Id+\Lambda^{\mathrm{TB}}_1\left( W^\mathrm{TB}-\Id\right)\Lambda^{\mathrm{TB}}_1\right)\,. $$ Because we assume \cref{eq:g(H) almost a projection} holds, \cite[Prop A.4]{Fonseca2020} implies that $ W^\mathrm{TB}-\Id$ is a local operator that exhibits diagonal decay in the $2$ direction (in the terminology of \cite{Shapiro2019} it is weakly-local and confined in the 2-direction). Hence, by \cite[Lemma 3.14]{Shapiro2019}, if we multiply $W^\mathrm{TB}-\Id$ with a local operator that has decay in the $1$ direction, the result is a trace-class, 
	and hence compact operator. So we may add or remove such products from the index without changing its value. 
	
	Let us define the diagonal operator $L_1$ on $\discHilSp$ via its matrix elements $$ L_{1,nn} := \left\langle \vf_n^\lambda,\Lambda(X_1)\vf_n^\lambda\right\rangle = \int_{x\in\contSp} |\vf_0^\lambda(x-n)|^2 \Lambda(x_1) \dif{x},\quad n=(n_1,n_2)\in\discSp\,. $$
	The operator $ \Lambda^{\mathrm{TB}}_1 - L_1$ is trivially local (as it is diagonal). We claim it has decay in the $1-$direction. Indeed, $$ (\Lambda^{\mathrm{TB}}_1 - L_1)_{nn} = \int_{x\in\contSp}|\vf_0^\lambda(x)|^2 (\Lambda(n_1)-\Lambda(x_1+n_1))\dif{x}$$ but $\Lambda(n_1)-\Lambda(x_1+n_1)$ is supported outside a column of width about $|n_1|$.
	 Due to the Gaussian decay of $\vf_0^\lambda$, the integral tends  to zero as $|n_1|\to\infty$.
	
	By stability under compact perturbations (and since the space of compact operators is an ideal) we find thus
	$$ \hat{\calN}^{\mathrm{TB}} = \findex_{\ell^2}\ L_1 \left(W^\mathrm{TB}-\Id\right)L_1+\Id\,.$$
	
	We now define $A_1 := M^{-1}J \Lambda(X_1)J^\ast M^{-1} $. In a manner similar to the proof for the bulk case, the difference $L_1-A_1$  decays to zero;  it has no diagonal element (by construction) and the off-diagonal matrix elements are summable \emph{and} exponentially decaying (just as in \cref{eq:xi_TB and MuM}). We thus find 
	$$ \hat{\calN}^{\mathrm{TB}} = \findex_{\ell^2}A_1\left( W^\mathrm{TB}-\Id\right)A_1+\Id\, .$$
	
	Let us  now  define $ B_1 := J \Lambda(X_1)  J^\ast \equiv J \Lambda_1 J^\ast$ which is arbitrary close to $A_1$ in norm, since  $\norm{M-\Id}\to 0$  as $\wellDepth\to\infty$.
	Hence,
	\begin{align} \hat{\calN}^{\mathrm{TB}} = \findex_{\ell^2}B_1 \left(W^\mathrm{TB}-\Id\right)B_1+\Id\,.
	\label{eq:finished rewriting edge discrete index using continuum position operator}\end{align}
	
	We are now prepared to initiate re-writing our index in the continuum Hilbert space. We have, starting with \cref{eq:finished rewriting edge discrete index using continuum position operator}, by algebraic manipulations based on the definitions and properties of $J$ and $J^\ast$, and properties of the Fredholm $\findex$ under direct sums and isomorphisms:

	\begin{align}
	    \hat{\calN}^{\mathrm{TB}} 
	    &=\findex_{\ell^2}J\ \Lambda_1 J^\ast\left(W^\mathrm{TB}-\Id_{\ell^2}\right)J\Lambda_1J^\ast+\Id_{\ell^2}\nonumber\\
	    &=\findex_{\ell^2}J\ \left(\Lambda_1 J^\ast \left(W^\mathrm{TB}-\Id_{\ell^2}\right)J\Lambda_1+\Id_{\mathcal{V}}\right)J^\ast\nonumber\\  
	    &=\findex_{\mathcal{V}}\ \Pi\Lambda_1 J^\ast \left(W^\mathrm{TB}-\Id_{\ell^2}\right)J\Lambda_1\Pi+\Id_{\mathcal{V}}\nonumber\\
	    &=\findex_{L^2}\ \Pi^\perp+\Pi\Lambda_1 J^\ast \left(W^\mathrm{TB}-\Id_{\ell^2}\right)J\Lambda_1\Pi+\Pi\nonumber\\
	    &=\findex_{L^2}\ \Id_{L^2}+\Pi\Lambda_1  \left(J^\ast W^\mathrm{TB}J-J^\ast J\right)\Lambda_1\Pi \nonumber\\
	   &= \findex_{L^2}\ \Id_{L^2}+\Pi\Lambda_1
	   \left(\mathbb{J}W^\mathrm{TB}-\Id_{\mathcal{V}}\right)\Lambda_1\Pi\ . \label{eq:the end of step 1 for edge correspondence}
	\end{align}
	\paragraph{Step 2}
	In \cref{eq:the end of step 1 for edge correspondence}, we next seek to commute the outer factors of $\Pi$ with $\Lambda_{1}$: 
	\[ \Pi\Lambda_1\left(\mathbb{J}W^\mathrm{TB}-\Id_{\mathcal{V}}\right)\Lambda_{1}\Pi\quad \xmapsto\quad\quad  \Lambda_{1}\Pi\left(\mathbb{J}W^\mathrm{TB}-\Id_{\mathcal{V}}\right)\Pi\Lambda_{1}.\] 
	 To do so without changing the value of the index, it is sufficient to show that
	  \begin{align}  [\Pi,\Lambda_{1}]\Pi J^\ast \left(W^\mathrm{TB}-\Id\right)J\Pi\quad\textrm{ is compact.} \label{eq:operator which should be compact}\end{align}
	  We claim in fact that this operator derives from a Hilbert-Schmidt kernel; it is hence
	   compact, and even of trace class.
To see this, we first note $$ \Pi(x,y) = \sum_{n\in\discSp} \ti{\vf_n}(x)\overline{\ti{\vf}_n(y)} $$ which we may be expressed (via \cref{eq:rotation from orbital basis to ONB})   in terms of the orbitals: $$ \Pi(x,y) = \sum_{n,n',n''\in\discSp} \overline{M_{n, n'}}M_{n,n''} \vf_{n'}(x)\overline{\vf_{n''}(y)} ,$$
  where we have the basic bound
    \begin{align*} |\vf_n(x)| &\leq \chi_{B_{r_0}(n)}(x)|\vf_n(x)|+\chi_{B_{r_0}(n)^c}(x)C\sqrt{\lambda}\exp(-\frac14\lambda(\normEuc{n-x}^2-r_0^2))  \\
    &=\chi_{B_{r_0}(x)}(n)|\vf_n(x)|+\chi_{B_{r_0}(x)^c}(n)C\sqrt{\lambda}\exp(-\frac14\lambda(\normEuc{n-x}^2-r_0^2))\tag{Replace $\chi_{B_{r_0}(n)^c}(x)$ with $\chi_{B_{r_0}(x)^c}(n)$.}\,.
    \end{align*}
    Now let $m_x\in\GG$ be the atomic center closest to $x$. There are two possibilities: either $x\in\RR^2$ is within a ball of radius $r_0$ of $m_x$, or it is outside of \emph{all} balls. We conclude $$ |\vf_n(x)| \leq \chi_{B_{r_0}(m_x)}(x)  |\vf_{m_x}(x)|\delta_{n,m_x}+\chi_{B_{r_0}(m_x)^c}(x)C\sqrt{\lambda}\exp(-\frac14 \lambda(\normEuc{x-n}^2-r_0^2)) $$ i.e., either way this expression is decaying in $\normEuc{x-n}$, for fixed $x$.
    
    Now using \cref{prop:M has off-diagonal decay too} and repeated application of the triangle inequality yields $$ |\Pi(x,y)|\leq C\exp(-\mu\normEuc{x-y})\,. $$
    
	The remainder of the proof arguments are along the lines of those presented in \cite{Fonseca2020}, so we only present a sketch. Recall that we have assumed  $\left[g(\tbH)^2-g(\tbH)\right]_{n,m}$ is decaying in the $e_2=(0,1)$ direction (as $|n_2|\to\infty$) see  \cref{eq:g(H) almost a projection}.
	\begin{enumerate}
	    \item Since $g(\tbH)$ is essentially a projection away from the edge truncation ($|n_2|\to\infty$) and ${\rm exp}(-2\pi\ii\mathrm{projection})=\Id$,  it follows that the kernel of  $W^\mathrm{TB}-\Id$ should decay away from the truncation, which we have used already. That is, we expect an estimate of the form $$ |(W^\mathrm{TB}-\Id)_{n,m}| \leq C\exp(-\mu(\normEuc{n-m})-\nu(|n_2|+|m_2|)) \qquad (n,m\in\GG)\,.$$ 
	    This is proven in \cite[Proposition A.4]{Fonseca2020}. Extrapolation of an operator $A$ with matrix elements $A_{nm}$ on $\discHilSp$  to $\contHilSp$, gives the integral kernel:
	    $$ A(x,y) = \sum_{n,n',m,m'\in\discSp} A_{nm}\overline{M_{n,n'}} M_{m,m'} \vf_{n'}(x)\overline{\vf_{m'}(y)}. $$ 
	    Therefore, by a very similar analysis to that  done for  $\Pi$ we have
	    \begin{equation}
	    |\left(\mathbb{J}W^\mathrm{TB}-\Id_{\mathcal{V}}\right)(x,y)| \leq C \exp(-\mu \normEuc{x-y}-\nu|x_2|)\qquad(x,y\in\RR^2)\,. 
	    \label{JUm1}\end{equation}
	    Here, we repeatedly use  that the product of off-diagonal decaying operators is again off-diagonal decaying, and that if one of such operators in the product also has diagonal decay, then the full product has the same diagonal decay; see \cite[Section 3]{Shapiro2019} for details.
	    \item The operator $[\Lambda_1,\Pi]$ has the integral kernel $$ ([\Lambda_1,\Pi])(x,y) = (\Lambda(x_1)-\Lambda(y_1))\Pi(x,y)\,. $$ It is proven in \cite[Corollary 3.16]{Shapiro2019} that if $\Pi$ has off-diagonal decay and $\Lambda$ is a switch function, 
	    \begin{equation}
	  |([\Lambda_1,\Pi])(x,y)|\leq C\exp(-\mu\normEuc{x-y}-\nu|x_1|)\qquad(x,y\in\RR^2)\,. \label{JUm2}\end{equation}
	    \item Combining \cref{JUm1} and \cref{JUm2} implies 
	    $$ |([\Lambda_1, \Pi] \left(\mathbb{J}W^\mathrm{TB}-\Id_{\mathcal{V}}\right))(x,y)| \leq C \exp(-\mu\normEuc{x-y}-\nu\normEuc{x}).$$ 
	    Hence, $[\Lambda_1, \Pi] \left(\mathbb{J}W^\mathrm{TB}-\Id_{\mathcal{V}}\right)$ has a Hilbert-Schmidt kernel and is hence compact. 
	\end{enumerate}

	Let $w(z):=\exp(-2\pi\ii g(z))$ for all $z\in\RR$ so that $W^\mathrm{TB} = w(\tbH)$.
	Since the operator displayed in \cref{eq:operator which should be compact} is compact, we may commute factors of $\Pi$ and $\Lambda_1$ to obtain  from \cref{eq:the end of step 1 for edge correspondence}, after further algebraic steps:
	\begin{align}
	    \hat{\calN}^{\mathrm{TB}} & = \findex_{L^{2}}\Id_{L^{2}}+\Lambda_{1}\Pi\left(\mathbb{J}w(\tbH)-\Id_{\mathcal{V}}\right)\Pi\Lambda_{1}\nonumber \\ 
	    &=\findex_{L^{2}}\Id_{L^{2}}+\Lambda_{1}\left(\Pi\left(\mathbb{J}w(\tbH)\right)\Pi-\Pi\right)\Lambda_{1}\tag{Using \ensuremath{J\Pi=J}}\nonumber \\ 
	    &=\findex_{L^{2}}\Id_{L^{2}}+\Lambda_{1}\left(\mathbb{J}w(\tbH)+\Pi^{\perp}-\Id_{L^{2}}\right)\Lambda_{1}\nonumber\\
	    &=\findex_{L^{2}}\Id_{L^{2}}-\Lambda_1+\Lambda_{1}\left(\mathbb{J}w(\tbH)+\Pi^{\perp}\right)\Lambda_{1}
	     \nonumber \\
	    &=\findex_{L^{2}}\Id_{L^{2}}-\Lambda_1+\Lambda_{1}\left(\mathbb{J}w(\tbH)+w(0)\Pi^{\perp}\right)\Lambda_{1},\nonumber\end{align}
	    where we have used the normalization $w(0)=1$;  (see \cref{eq:normalization condition for g}). Let us define for convenience the unitary operator $Y := \mathbb{J}w(\tbH)+w(0)\Pi^{\perp}$.
	    
	    
	    It follows that 
	       \begin{align}
	 \hat{\calN}^{\mathrm{TB}} & =\findex_{L^{2}} \ \Lambda_{1}Y\Lambda_{1}\ +\ 
	  \Lambda_1^\perp,\quad w(z)=\exp(-2\pi\ii g(z)). \label{eq:end of step 2 in edge correspondence proof}
	\end{align}
	There is a sense in which $Y$ is equivalent to $w(\tilde{H})$, and once we have established that, we would be finished. Before we do so, let us note that the space of Fredholm operators is open in the operator norm topology on $L^2$, and moreover, if $G$ is a parametrix of a Fredholm operator $F$, then $$ B_{\norm{G}^{-1}}(F) \equiv \Set{ \tilde{F} | \norm{F-\tilde{F}}<\norm{G}^{-1}} $$ lies within the space of Fredholm operators. Furthermore, a calculation establishes that \begin{align} \Lambda_{1}Y^\ast \Lambda_{1}+\Lambda_1^\perp \label{eq:parametrix}\end{align} is the parametrix of the operator within the index in \cref{eq:end of step 2 in edge correspondence proof}: $$ \Id-(\Lambda_{1}Y\Lambda_{1}+\Lambda_1^\perp)(\Lambda_{1}Y^\ast\Lambda_{1}+\Lambda_1^\perp) = [\Lambda_1,Y-\Id ]\Lambda_1^\perp Y^\ast\Lambda_1\,. $$ In that calculation, one uses that $Y-\Id$ decays in the $2$-axis and that taking the commutator $[\Lambda_1, Y-\Id]$ thus yields a compact operator. Now to get a lower bound on \begin{align}\label{eq:norm of parametrix}\norm{\Lambda_{1}Y^\ast\Lambda_{1}\ +\ 
	\Lambda_1^\perp}^{-1}\,,\end{align} since $Y$ is unitary and $\Lambda_1$ is a projection, we learn that \cref{eq:norm of parametrix} is bounded below by $1/2$, i.e., the radius of the ball about the operator in \cref{eq:end of step 2 in edge correspondence proof} within which the Fredholm property holds is independent of $\lambda$.
	
	\paragraph{Step 3} 
Recall that $\tilde{H}^\lambda \equiv H^\lambda / \nnHopping^\lambda$. To complete the proof of \cref{thm:edge continuum discrete correspondence}, we will use \cref{res-conv} on resolvent convergence:  $R_{\tilde{H}^\lambda}(z)\to\mathbb{J}R_{\tbH}(z)$ as $\lambda\to\infty$ to replace $Y$ in \cref{eq:end of step 2 in edge correspondence proof} by $w(\tilde{H}^\lambda)$ for all $\lambda$ sufficiently large.
Since $g$ is only smooth (not analytic) to control (for large $\lambda$), to deal with it we require the smooth functional calculus i.e., the 
Helffer-Sjostrand formula; see \cite[Lemma B1]{Hunziker_Sigal_2000_doi:10.1142/S0129055X0000040X} for unbounded operators or \cite[(A.10)]{Elbau_Graf_2002} for $f'$ (but not $f$) with compact support.
 For any smooth $f:\RR\to\CC$, 
\begin{align} f(\ti{H}) = \frac{1}{2\pi}\int_{z\in\CC} R_{\ti H}(z) (\partial_{\bar{z}} \ti f)(z) \dif{z} \label{eq:the smooth functional calculus}\end{align} where $\ti f$ is a quasi-analytic extension of $f$ onto a strip of finite width about the real axis, which can be chosen to obey 
$$ \left|\left(\partial_{\overline{z}}\tilde{f}\right)\left(x+\ii y\right)\right|\leq C_{N}\left|y\right|^{N} \qquad(x,y\in\RR) $$ for any $N\in\NN$ large. 

To apply the Helffer-Sj\"ostrand formula \cref{eq:the smooth functional calculus} to a proof that  $w(\tilde{H}^\lambda) \to \mathbb{J}w(\tbH)+w(0)\Pi^{\perp}\equiv Y$ in operator norm would require control of resolvent convergence up to the real axis. This is not implied by \cref{res-conv}, so we proceed by an approximation argument.

Let now $w$ be smooth and for any given $\ve>0$, we define an approximation of $w$:
\begin{align} w_\ve(\ti{H}^\lambda) = \frac{1}{2\pi}\int_{z\in\Set{z'\in\CC | |\Im\{z'\}|\geq\ve}} R_{\ti H^\lambda}(z) (\partial_{\bar{z}} \ti w)(z) \dif{z}\,. \label{eq:analytic approximation of smooth functional calculus}\end{align}
and define:
    \begin{equation}
     \mathbb{\Lambda}_1 Y \equiv \Lambda_1 Y \Lambda_1 + \Lambda_1^\perp 
     \label{bbLam1}\end{equation}
     and similarly for $w(\tilde{H}^\lambda)$. Using the triangle inequality and that $\norm{\Lambda_1}\leq 1$, we have:
     \begin{align} &\norm{\mathbb{\Lambda}_1 Y-\mathbb{\Lambda}_1 w(\ti H)}\nonumber\\
     &\quad \leq \norm{w( \tbH)-w_\ve( \tbH)} + \norm{\contToDiscBLO w_\ve( \tbH)+w_\ve(0)\Pi^\perp-w_\ve(\ti H)} + \norm{w_\ve(\ti H) - w(\ti H)}\,. \label{eq:estimate the difference in norm between continuum and tb operators for edge invariant}\end{align}
    We begin by bounding the first and third terms in \cref{eq:estimate the difference in norm between continuum and tb operators for edge invariant}. For the first term we have:
    \begin{align*}
         \norm{ w( \tbH)- w_\ve(\tbH)}
        &\leq \frac{1}{2\pi}\norm{\int_{z\in\Set{w\in\CC | |\Im\{w\}|<\ve}} R_{\tbH}(z) (\partial_{\bar{z}} \ti w)(z) \dif{z}} \\
        &\leq \frac{C C_N}{2\pi}\int_{\eta \in (-\ve,\ve)} \frac{1}{|\eta|} |\eta|^N \dif{\eta} = \frac{C C_N}{2\pi} \frac{2 \ve ^{N}}{N} ,
    \end{align*}
    where we have used that $w$ is bounded by $1$, the trivial bound on the resolvent, and the known bounds on $\ti{w}$. The same bound holds for the third term in \cref{eq:estimate the difference in norm between continuum and tb operators for edge invariant}.
    
    For the middle term, we have using the approximation \cref{eq:analytic approximation of smooth functional calculus}, 
    \begin{align}\norm{\contToDiscBLO w_\ve( \tbH)+w_\ve(0)\Pi^\perp-w_\ve(\ti H)} & \leq \frac{1}{2\pi}\norm{\int_{z\in\Set{z'\in\CC | |\Im\{z'\}|\geq\ve}} \left(\contToDiscBLO R_{ \tbH}(z)-R_{\ti H^\lambda}(z)\right) (\partial_{\bar{z}} \ti w)(z) \dif{z}}\nonumber\\
    &\leq \frac{C C_N}{2\pi}\frac{2 \ve ^{N+1}}{N+1}\sup_{|\Im\{z\}|\geq \ve}\norm{\contToDiscBLO R_{ \tbH}(z)-R_{\ti H^\lambda}(z)}\,.
    \label{diff1}\end{align}
    
   In this bound, the spectral parameter, $z$, is uniformly bounded away (at least $\ve$ distant) from $\sigma(\tbH)$. We may therefore use  \cref{res-conv} on resolvent convergence (in particular, \cref{res-conv1}) to obtain a $\lambda_\star(\ve)>0$ such that for all $\lambda\geq\lambda_\star(\ve)$, 
   \begin{equation}
  \sup_{|\Im\{z\}|\geq \ve}\norm{\contToDiscBLO R_{ \tbH}(z)-R_{\ti H^\lambda}(z)} \leq C \exp(-c \lambda) \frac{\frac{1}{\ve^2}}{1-\frac{1}{\ve}\exp(-c\lambda)} .
   \label{diff2}\end{equation}
   Combining \cref{diff1} and \cref{diff2} we find
   $$ \norm{\contToDiscBLO w_\ve( \tbH)+w_\ve(0)\Pi^\perp-w_\ve(\ti H)} \leq C \frac{\ee^{-c\lambda}\ve^{N-1}}{1-\frac{1}{\ve}\ee^{-c\lambda}}\,.$$
     It follows that if  $ \lambda \geq \max \left(\Set{\lambda_\star(\ve),-\frac{1}{c}\log(\frac12\ve)}\right) $ then we are guaranteed that $$ \norm{\contToDiscBLO w_\ve( \tbH)+w_\ve(0)\Pi^\perp-w_\ve(\ti H)} \leq C \ee^{-c\lambda}\ve^{N-1}\,.$$
    Hence, for such sufficiently large $\lambda$, we find then $$ \norm{\mathbb{\Lambda}_1 Y-\mathbb{\Lambda}_1 w(\ti H^\lambda)} \leq 2\times \frac{C C_N}{2\pi} \frac{2 \ve ^{N}}{N} + C \ee^{-c\lambda}\ve^{N-1} \leq C \ve^{N-1}\,.  $$ 
    
 Since the space of Fredholm operators is open in the norm topology and contains $\mathbb{\Lambda}_1 Y$, by \cref{lem:discrete edge index is well-defined}, there exists $\ve$ (depending only on $\tbH$ and not on $\lambda$) such that $\mathbb{\Lambda}_1 w(\ti H^\lambda)$ is contained 
  in any prescribed open subset containing $\mathbb{\Lambda}_1 Y$. Hence, $\mathbb{\Lambda}_1 w(\ti H^\lambda)$ is Fredholm.
    
   Proceeding now from \cref{eq:end of step 2 in edge correspondence proof},  we find that for all $\lambda$ sufficiently large, using the norm continuity of the Fredholm index that:
   \begin{align*}
        \hat{\calN}^{\mathrm{TB}} &= \findex_{L^{2}} \mathbb{\Lambda_1} Y  =\findex_{L^{2}} \mathbb{\Lambda_1} w(\ti H^\lambda)\\
        &=\findex_{L^{2}}\Lambda_{1}\exp\left(-2\pi\ii g\left(\tilde{H}^\lambda\right)\right)\Lambda_{1}+ \Lambda_{1}^\perp\\
        &\equiv \hat{\mathcal{N}}^\lambda\ .
    \end{align*}
 The proof of \cref{thm:edge continuum discrete correspondence} is now complete.
\end{proof}
\appendix
\section{The Gramian; proof of \cref{Gram-prop} }\label{sec:gramian}
We prove \cref{Gram-prop}. 
\begin{proof}  Part 1. Proof of \cref{eq:locality of Gramian}:
	\begin{align}
		|\Gram[e]_{nm}| &\le \int_{\RR^2} |\varphi^\lambda(y)|\ |\varphi^\lambda(y-(m-n))|\ dy \nonumber\\
		&= \Big(\ \int_{B_{r_0}(0)}+\int_{B_{r_0}(m-n)|}+\int_{\left(B_{r_0}(0)\cup B_{r_0}(m-n)\right)^c}\ \Big) |\varphi^\lambda(y)|\ |\varphi^\lambda(y-(m-n))|\ dy\label{gmn1}
	\end{align}
	Call the three integrals in \cref{gmn1}: $I_1$, $I_2$ and $I_3$.  Starting with $I_1$, we have by the Cauchy-Schwarz inequality and the normalization of $\varphi^\lambda_0$ that
	\[ | I_1|\le \Big| \int_{B_{r_0}(0)} |\varphi^\lambda(y)|\ |\varphi^\lambda(y-(m-n))|\ dy\Big|\ 
	\le\ 
	\left( \int_{B_{r_0}(0)} |\varphi^\lambda(y-(m-n))|^2\ dy\right)^{1\over2}.\]
	Since $y$ ranges over $|y|<r_0$, we have 
	\begin{equation}  |y-(m-n)|\ge \Big| \frac{m-n}{|m-n|}r_0- (m-n)\Big|=|m-n| \left(1-\frac{r_0}{|m-n|}\right) .
		\label{y-m-n}\end{equation}
	Since $|m-n|\ge a>2r_0$, it follows that $|y-(m-n)|\ge 2r_0 (1-\frac{r_0}{2r_0})>r_0$
	and so using the pointwise Gaussian bound of \cref{gs-bound} we have
	$ | I_1| \le C\ (\pi r_0^2)^{\frac12}\ \lambda^{1\over 2}\ \max_{|y|\le r_0} \ee^{-\frac14\lambda(|y-(m-n)|^2-r_0^2)}.
	$
	Again applying \cref{y-m-n} we have 
	\[ | I_1| \le C\ \lambda^{1\over 2}\  \ee^{-\frac{1}{16}\lambda(|m-n|^2-(2r_0)^2)}.\] 
	Similarly, $|I_2|$ satisfies the identical upper bound. Finally, to estimate $|I_3|$ we may apply \cref{gs-bound} to both factors in the integrand. Together with the identity: 
	\[ |y|^2 + |y-(m-n)|^2=2\Big|y-\frac{m-n}2\Big|^2+\frac12 |m-n|^2\]
	we obtain
	\[ |I_3|\le C\lambda \ee^{-\frac{\lambda}{8}|m-n|^2}\ \int \exp\left(-\frac{\lambda}{2} \Big|y-\frac{m-n}2\Big|^2\right)\ dy
	\le C^\prime \ee^{-\frac{\lambda}{8}|m-n|^2}.\]
	Hence, $|\Gram[e]_{nm}|\le |I_1|+|I_2|+|I_3|$ satisfies the upper bound \cref{eq:locality of Gramian}.
	\medskip
	
	\noindent Part 2. Proof of \cref{eq:G_offdiag}:  
	Using \cref{lem:hsy}
	and that $G_{nm}=\overline{G_{mn}}$, we have 
	$$              \norm{\Gram-\Id}_{_{\mathcal{B}(\discHilSp}} \leq \sup_{n\in\discSp} \sum_{m\in\discSp} |\Gram_{nm}-\delta_{nm}| =
	\sup_{n\in\discSp} \sum_{m\in\setwo{\discSp}{n}} |\Gram_{nm}| .
	$$ 
	The second equality holds since $\Gram_{nn}-1 = \norm{\oneAtomGrndStt_n}^2 - 1 = 0$, 
	by normalization of $\varphi_0^\lambda$. Thus, that $G$ is self-adjoint and bounded on $\discHilSp$ will follow from 
	\cref{eq:G_offdiag}.
	
	Recall that $m\ne n$ implies $|m-n|\ge a>2r_0$  . So using \cref{eq:locality of Gramian} we have 
	\begin{align*}
		\norm{\Gram-\Id}_{_{\mathcal{B}(\discHilSp}}\ &=\  \sup_{n\in\discSp}\sum_{|m-n|\geq a}|G_{nm}|  \leq
		C\ \lambda^{1\over 2}\  \sup_{n\in\discSp}\sum_{|m-n|\geq a}  \ee^{-\frac{\lambda}{16}(|m-n|^2-(2r_0)^2)}\\
		&\leq C\ \lambda^{1\over 2}\ee^{-\alpha\frac{\lambda}{16}a^2+\frac14\wellDepth r_0^2}\sup_{n\in\discSp}\sum_{|m-n|\geq a}  \ee^{-(1-\alpha)\frac{\lambda}{16}|m-n|^2}\tag{For any $\alpha\in(0,1)$} \\
		&\leq C\ \lambda^{1\over 2}\ee^{-\alpha\frac{\lambda}{16}a^2+\frac14\wellDepth r_0^2}\sup_{n\in\discSp}\sum_{m\in\discSp}  \ee^{-(1-\alpha)\frac{\lambda}{16}|m-n|^2}\\
		&\leq C\ \lambda^{1\over 2}\ee^{-\alpha\frac{\lambda}{16}a^2+\frac14\wellDepth r_0^2}\,.\tag{Using \cref{eq:Gaussians are summable in discrete space}}
	\end{align*}
	\medskip
	
	\noindent Part 3. Proof of \cref{eq:strictly positive gap for Gramian}: Since $G=\Id + (G-\Id)$, $G^{-1}$ exists is bounded
	if $\|G-\Id\|<1$. This holds for $\lambda$ sufficiently large and $G^{-1}$ can be expressed as a geometric (Neumann) series. The bound  \cref{eq:strictly positive gap for Gramian} follows from summing the norms of terms in this series.
	Finally, for large $\lambda$ we have that $G$ is invertible and hence $\{\varphi_n^\lambda\}_{n\in\GG}$ is a linearly independent set. Thus, for any non-zero $w\in \discHilSp$, we have
	$\left\langle w,Gw\right\rangle=\sum_{m,n}\left\langle \varphi_m,\varphi_n\right\rangle\overline{w_m}w_n=
	\|\sum_m w_m\varphi_m\|^2>0 $.

	\medskip
	
	\noindent Part 4.  Since $\Gram^{-1}$ is a strictly positive and bounded linear operator on $\discHilSp$, there is a bounded invertible operator on $\discHilSp$, 
	\begin{equation}
		\ONBer:=\sqrt{\Gram^{-1}}=\Gram^{-\frac12},\label{M-def}\end{equation} such that:
	$ \Gram^{-1} = \ONBer^*\ONBer \equiv  |\ONBer|^2.$
	Furthermore, by invertibility of $\ONBer$ we have  $G=\ONBer^{-1}(\ONBer^*)^{-1}$. 
	This completes the proof of \cref{Gram-prop}.
\end{proof}

\section{Proofs of propositions in \cref{sec:proof of main theorem}}\label{subsec:proofs of propositions of the main theorem}
The following is a preliminary lemma that will be used several times in the sequel.
\begin{lem}\label{lem:basic lemma to bound inner product of Hamiltonian}
	Let $\alpha\in(0,1)$. Then \begin{align} \sum_{l\in\setwo{\discSp}{n,m}}|\ip{|\oneAtomGrndStt_n|}{\chi_{B_{r_0}(l)}|\oneAtomGrndStt_m|}| \lesssim \frac{1}{1-\alpha}
		\ee^{-\alpha \frac12\wellDepth(a-r_0)^2} 
		\ee^{-\frac{\lambda(1-\alpha)}{8}\normEuc{n-m}^2}\  \ee^{\frac{\lambda}{2}r_0^2}\,.\label{eq:basic lemma to bound inner product of Hamiltonian} \end{align}
\end{lem}
\begin{proof}
	Since $l\ne n,m$, we may use the Gaussian bound of \cref{gs-bound}: $$\vf(z) \leq C \sqrt{\wellDepth} \exp (-\frac14 \wellDepth(\normEuc{z}^2-\oneAtomPotSuppRad^2))$$ for $|z|>r_0$ to bound the integrand on the LHS of \cref{eq:basic lemma to bound inner product of Hamiltonian}:
	\begin{align} |\oneAtomGrndStt_n(z)||\oneAtomGrndStt_m(z)|&\leq C\wellDepth\ 
		\ee^{-\frac{\lambda}{4}(\normEuc{z-n}^2+\normEuc{z-m}^2)}\ \ee^{\frac{\lambda}{2}r_0^2} 
		\nonumber\\
		&\lesssim \wellDepth \ee^{-\frac{\lambda}{8}\normEuc{n-m}^2-\frac{\lambda}{2}\normEuc{z-\frac12(n+m)}^2}\ 
		\ee^{\frac{\lambda}{2}r_0^2}\label{vphi-nm1}
	\end{align}
	On the other hand, since for $|z-l|<r_0$ with $l\ne m,n$,  $|z-m|>r_0$ and $|z-n|>r_0$, 
	the Gaussian bound on $\varphi_0$ implies
	\begin{equation}
		|\oneAtomGrndStt_n(z)||\oneAtomGrndStt_{m}(z)|\leq C\wellDepth 
		\ee^{-\frac12\wellDepth(a-r_0)^2} \ee^{\frac{\lambda}{2} r_0^2} .
		\label{vphi-nm2}\end{equation}
	By  \cref{vphi-nm1} and \cref{vphi-nm2}, for any  $\alpha\in(0,1)$, we have
	\begin{align*}
		|\oneAtomGrndStt_n(z)||\oneAtomGrndStt_{m}(z)| 
		&= (|\oneAtomGrndStt_n(z)||\oneAtomGrndStt_{m}(z)|)^{\alpha}(|\oneAtomGrndStt_n(z)||\oneAtomGrndStt_{m}(z)|)^{1-\alpha}\\
		&\lesssim\ \lambda  \ee^{-\alpha \frac12\wellDepth(a-r_0)^2} 
		\ee^{-(1-\alpha)\left[\frac{\lambda}{8}\normEuc{n-m}^2+\frac{\lambda}{2}\normEuc{z-\frac12(n+m)}^2\right]}\ \ \ee^{\frac{\lambda}{2}r_0^2}\\
		&=\ \lambda\ \  \ee^{-\alpha \frac12\wellDepth(a-r_0)^2} 
		\ee^{-\frac{\lambda(1-\alpha)}{8}\normEuc{n-m}^2}\  \ee^{\frac{\lambda}{2}r_0^2}\ \ 
		\ee^{-\frac{\lambda(1-\alpha)}{2}\normEuc{z-\frac12(n+m)}^2}\ 
	\end{align*}
	For $l\ne n,m$, we integrate this bound over the set $|z-l|<r_0$ and then sum over $l$. Bounding the sum of  integrals by a Gaussian integral over all $\RR^2$ we obtain the RHS of \cref{eq:basic lemma to bound inner product of Hamiltonian}.
\end{proof}
\begin{proof}[Proof of \cref{prop:comparing the almost eigenstates to nn hopping}]

	We will prove that for all $q<1$, 
	\begin{equation}
		|\rho|^{-q}\|H^\lambda\Pi\|_{\mathcal{B}(L^2(\RR^2)}\ \to\ 0\quad {\rm as}\ \lambda\to\infty .
		\label{HPi-lim}\end{equation}
	Let $\psi\in\orbitalSubSp={\rm Range}(\Pi)={\rm span}\{\oneAtomGrndStt_n\}_{n\in\GG}$.
	We write $\psi = \sum_{n\in\discSp}\ \beta_n\oneAtomGrndStt_n$ for some $\beta\in l^2(\GG)$ and have
	\begin{align*}
		\norm{\crysHamil\psi}^2 &= \sum_{n,m\in\discSp}\ A_{nm}\
		\overline{\beta_n} \beta_m
	\end{align*}
	where 
	\begin{equation}
		A^\lambda_{nm} := \sum_{l\in\setwo{\discSp}{n},k\in\setwo{\discSp}{m}}\ip{\oneAtomPotOp_l\oneAtomGrndStt_n}{\oneAtomPotOp_k\oneAtomGrndStt_m}
		=\  \sum_{l\in\setwo{\discSp}{n,m}}\ip{\oneAtomPotOp_l\oneAtomGrndStt_n}{\oneAtomPotOp_l\oneAtomGrndStt_m},\label{Amn}\end{equation}
	since $v_l$ and $v_k$ are disjointly supported if $l\ne k$. 
	We next bound $\norm{\crysHamil\psi}^2$ using \cref{lem:Relating the little l^2 norm with the L^2 norm}:
	\begin{align*} 
		\norm{\crysHamil\psi}^2 &\leq \sum_{n,m\in\discSp} \sqrt{|A^\lambda_{nm}|}\ |\beta_n|\ 
		\sqrt{|A^\lambda_{nm}|}\ |\beta_m| \\
		&\leq \sqrt{\sum_{n,m\in\discSp} |A^\lambda_{nm}|\ |\beta_n|^2 }\sqrt{\sum_{n,m\in\discSp} |A^\lambda_{nm}|\ |\beta_m|^2} \\ 
		&\leq \sqrt{\sup_{n\in\discSp}\sum_{m\in\discSp} |A^\lambda_{nm}|}\ \sqrt{\sup_{m\in\discSp}\sum_{n\in\discSp} |A^\lambda_{nm}|}\  \norm{\beta}^2 \\
		&\leq \big\|G^{1\over2}\big\|^2\  \max\Big(\norm{A^\lambda}_{1,\infty}, \norm{A^\lambda}_{\infty,1}\Big)
		\ \norm{\psi}^2\ .
	\end{align*}
	Hence, to prove \cref{HPi-lim}, it suffices to show that  
	\begin{equation}
		|\nnHopping|^{-2q}\max\left(\Set{\norm{A}_{1,\infty},\norm{A}_{\infty,1}}\right)\to 0,\quad  \wellDepth\to\infty.
		\label{rat-lim}\end{equation}
	We first bound 
	$\norm{A}_{1,\infty}=\sup_m\sum_n |A_{nm}|$. The bound on 
	$\norm{A}_{\infty,1}$ is obtained similarly. 
	
	Using $|v_0(x)|\leq |v_{min}|\chi_{_{B_{\oneAtomPotSuppRad}}}(x)$ and $a>2r_0$, we have 
	\begin{equation}
		|A_{nm}| \le \lambda^4 |v_{min}|^2 \sum_{l\in\setwo{\discSp}{n,m}}|\ip{|\oneAtomGrndStt_n|}{\chi_{B_{r_0}(l)}|\oneAtomGrndStt_m|}|\,. \label{Anm-bd}
	\end{equation}
	Hence \cref{lem:basic lemma to bound inner product of Hamiltonian} implies that for all $\alpha\in(0,1)$,
	\begin{align}
		|A_{nm}| 
		&\lesssim \frac{\wellDepth^4 |v_{min}|^2}{1-\alpha}
		\ee^{-\alpha \frac12\wellDepth(a-r_0)^2} 
		\ee^{-\frac{\lambda(1-\alpha)}{8}\normEuc{n-m}^2}\  \ee^{\frac{\lambda}{2}r_0^2}\,.
		\label{eq:bound on Anm}	\end{align} 
	Summing over $n\in\GG$ and then taking supremum over $m\in\GG$, we have
	\begin{align*} \|A\|_{1,\infty}\ =\ \sup_{m\in\GG}\ \sum_{n\in\GG} |A_{nm}| \lesssim  \frac{\wellDepth^4 |v_{min}|^2}{1-\alpha}
		\ee^{-\alpha \frac{\lambda}{2}(a-r_0)^2} \   \ee^{\frac{\lambda}{2}r_0^2},
	\end{align*}
	with an implicit constant, which depends on \cref{eq:Gaussians are summable in discrete space} and which is bounded for all $\lambda$ bounded away from zero. The norm $ \|A\|_{\infty,1}$ satisfies the same type bound. 
	
	With a view toward verifying \cref{rat-lim}, we next bound $|\rho^\lambda|^{-1}$. 	By \cref{eq:bound on rho}, we have 
	$$|\nnHopping|^{-1} \leq \ee^{\frac{\lambda}{4}(\minLatticeSpacing^2+4\sqrt{|v_{min}|}\minLatticeSpacing+\gamma_0)	}$$
	where $\gamma_0$ is a constant which depends on $v_0$ but not on $a$. Therefore,
	$$|\nnHopping|^{-2q} \leq \ee^{q\frac{\lambda}{2}(\minLatticeSpacing^2+4\sqrt{|v_{min}|}\minLatticeSpacing+\gamma_0)	}.$$
	Hence, \cref{rat-lim} holds provided
	\begin{align} 
		q (\minLatticeSpacing^2+4\sqrt{|v_{min}|}\minLatticeSpacing+\gamma_0)	) < 
		\alpha(\minLatticeSpacing-\oneAtomPotSuppRad)^2-r_0^2 .
		\label{eq:a-constr}
	\end{align} 
	Fix $\alpha$ such that $q<\alpha<1$. 
	For this choice of $q$, we have that \cref{eq:a-constr} holds for all fixed  $a$ sufficiently large,
	and for such $q$ and $a$ \cref{rat-lim} holds as $\lambda\to\infty$. 
\end{proof}
\medskip

\begin{proof}[Proof of \cref{prop:controlling ONB rotation and NN truncation}]
	Our goal is to show that as $\lambda\to\infty$
	\[
	\norm{\orbitalProjBLO \tilde{H} -\contToDiscBLO\tbH }_{\mathcal{B}(L^2(\RR^2)} \to 0\quad \textrm{or  equivalently}\quad 
	\nnHopping^{-1}\norm{\orbitalProjBLO H -\nnHopping\contToDiscBLO\tbH }_{\mathcal{B}(L^2(\RR^2)} \to 0\,.\]
	We first introduce a bounded linear operator on $l^2(\GG)$, $H^{\mathrm{NN},\lambda}$, defined in terms of the matrix elements:
	\begin{equation}
		H^{\mathrm{NN},\lambda}_{n,m} := \frac{1}{\nnHopping^\lambda}
		\ip{\vf^\lambda_n}{\crysHamil^\lambda\vf^\lambda_m}\ \chi_{_{\rm NN}}(n,m),
		\label{HNNnm}
	\end{equation}
	where 
	\begin{equation}
		\chi_{_{\rm NN}}(n,m)\ :=\ \begin{cases} 1& \textrm{if  $n=m$ or $n\sim m$ in $\GG$ }\\ 0 & \textrm{otherwise}\end{cases}
		\label{chi-nm}
	\end{equation}
	We write 
	\[ \orbitalProjBLO \tilde{H}^\lambda -\contToDiscBLO\tbH = \rho^{-1}\left(\orbitalProjBLO {H}^\lambda - \rho^\lambda\ \contToDiscBLO H^{\mathrm{NN},\lambda}\right) +  \contToDiscBLO\left( H^{\mathrm{NN},\lambda}-\tbH\right)\]
	\cref{prop:controlling ONB rotation and NN truncation} is a consequence of the following two assertions:
	\begin{equation}
		|\nnHopping^{\lambda}|^{-1}\norm{\orbitalProjBLO H^\lambda -\nnHopping^\lambda\contToDiscBLO H^{\mathrm{NN},\lambda} } \to 0\,. \label{claim1}\end{equation}
	and 
	\begin{equation}
		\norm{H^{\mathrm{NN},\lambda}-\contToDiscBLO\tbH} \to 0\,.\label{claim2}\end{equation}
	\medskip
	
	\noindent{\it Proof of \cref{claim1}:} Since $H^{\mathrm{NN}}$ is defined in terms of the $\{\vf_n\}$ basis of
	$\orbitalSubSp$ we express its matrix elements in terms of the $\{\tilde\vf_n\}$ basis using the relations
	$\vf_n=\contToDiscBLO M^{-1} \tilde\vf_n= \contToDiscBLO G^{\frac12} \tilde\vf_n$; see    \cref{rotation}.
	For every $\psi=\sum_n\beta_n\tilde\vf_n\in {\rm Range}(\Pi)$. Then, 
	\begin{equation}
		\rho \contToDiscBLO  H^{\mathrm{NN}} \psi = \rho J^* H^{\mathrm{NN}} J \psi = 
		\sum_{m\in\GG}\ \Big[\sum_{n\in\GG} 
		\left\langle\tilde\vf_m, (\contToDiscBLO M^{-1} H \contToDiscBLO M^{-1})  \tilde\vf_n\right\rangle\ \chi_{_{\rm NN}}(n,m)\beta_n\Big]\tilde\vf_m.
		\label{mels-HNN}
	\end{equation}
	Hence, with respect to the basis $\{\tilde\vf_n\}$,
	\[ \rho \contToDiscBLO  H^{\mathrm{NN}} \ =\ \mathcal{N}\ \contToDiscBLO M^{-1} H \contToDiscBLO M^{-1},\]
	where $\mathcal{N}$ is a super-operator acting on $\mathcal{B}(\orbitalSubSp)$, which truncates all matrix elements,  $(m,n)$, other than those pairs for which $m=n$ or $m\sim n$ in $\GG$. We also write 
	$\mathcal{N}^\perp=\Id_{\mathcal{B}(\orbitalSubSp)}-\mathcal{N}$. 
	Finally, we  introduce the abbreviated notation $\frak{M}^{-1}$ which acts on $\mathcal{B}(\orbitalSubSp)$ has follows:
	$$\frak{M}^{-1} B := \contToDiscBLO M^{-1} B \contToDiscBLO M^{-1} , $$ 
	for any bounded operator $B\in \mathcal{B}(\orbitalSubSp)$. Thus,
	$$ \orbitalProjBLO H -\nnHopping\contToDiscBLO H^{\mathrm{NN}} = \orbitalProjBLO H -\calN\mathfrak{M}^{-1} \orbitalProjBLO H = (\Id - \calN\mathfrak{M}^{-1})\orbitalProjBLO H = (\Id-\mathfrak{M}^{-1})\orbitalProjBLO H+\nnProj^\perp\mathfrak{M}^{-1}\orbitalProjBLO H, $$
	and therefore
	\begin{align}
		|\rho|^{-1} \|\orbitalProjBLO H -\nnHopping\contToDiscBLO H^{\mathrm{NN}}\|_{\mathcal{B}(\orbitalSubSp)}  &\le
		|\rho|^{-1} \|(\Id-\mathfrak{M}^{-1})\orbitalProjBLO H\|_{\mathcal{B}(\orbitalSubSp)} + |\rho|^{-1}\|\nnProj^\perp\mathfrak{M}^{-1}\orbitalProjBLO H\|_{\mathcal{B}(\orbitalSubSp)} \nonumber\\
		&\le     |\rho|^{-1} \|(\Id-\mathfrak{M}^{-1})\|_{\mathcal{B}(\orbitalSubSp)\to \mathcal{B}(\orbitalSubSp)}\
		\|\orbitalProjBLO H\|_{\mathcal{B}(\orbitalSubSp)} + |\rho|^{-1}\|\nnProj^\perp\mathfrak{M}^{-1}\orbitalProjBLO H\|_{\mathcal{B}(\orbitalSubSp)}\nonumber\\
		&=:\ {\rm Term}_1\ +\ {\rm Term}_2. \label{term12}
	\end{align}
	Since $|\nnHopping|^{-q}\norm{\orbitalProjBLO\crysHamil}\to0$ for any $q<1$ (\cref{prop:comparing the almost eigenstates to nn hopping}), we write
	\begin{equation}
		{\rm Term}_1 = 	
		|\rho|^{q-1} \|(\Id-\mathfrak{M}^{-1})\|_{\mathcal{B}(\orbitalSubSp)\to \mathcal{B}(\orbitalSubSp)}\
		\times  |\rho|^{-q}   \|\orbitalProjBLO H\|_{\mathcal{B}(\orbitalSubSp)}.\label{term1a}
	\end{equation}
	To prove \cref{claim1}, we now show that for some fixed $q<1$:  
	\[  |\rho|^{q-1} \|(\Id-\mathfrak{M}^{-1})\|_{\mathcal{B}(\orbitalSubSp)\to \mathcal{B}(\orbitalSubSp)}\to0.\]
	Toward proving this, we first claim that: 
	\begin{equation} \|(\Id-\mathfrak{M}^{-1})\|_{\mathcal{B}(\orbitalSubSp)\to \mathcal{B}(\orbitalSubSp)} 
		\le \norm{\Id -  M^{-1}}_{\boundedLO{\discHilSp}}\ \left(1 + \|M^{-1}\|_{\boundedLO{\discHilSp}}\right).
		\label{claim3}\end{equation}
	To prove \cref{claim3} use that
	\begin{align*}
		\|(\Id-\mathfrak{M}^{-1})\|_{\mathcal{B}(\orbitalSubSp)\to \mathcal{B}(\orbitalSubSp)} 
		=  \sup_{\|B\|_{\mathcal{B}(\orbitalSubSp)}=1}\ \|(\Id-\mathfrak{M}^{-1})B\|_{\mathcal{B}(\orbitalSubSp)}.
	\end{align*}
	For any $B\in \mathcal{B}(\orbitalSubSp)$, using  that 
	$\contToDiscHilSp:\orbitalSubSp\to\discHilSp$ is an isometry and $\contToDiscBLO^\ast B := \contToDiscHilSp B\contToDiscHilSp^\ast$, we have
	\begin{align*}\norm{(\Id-\mathfrak{M}^{-1})B}_{\boundedLO{\orbitalSubSp}} &= \norm{B - \contToDiscBLO(M^{-1})B\contToDiscBLO(M^{-1})}_{\boundedLO{\orbitalSubSp}} \\
		&= \norm{B - (\contToDiscHilSp^\ast M^{-1}\contToDiscHilSp)\ B\ (\contToDiscHilSp^\ast M^{-1}\contToDiscHilSp ) }_{\boundedLO{\orbitalSubSp}} \\
		&= \norm{\JJ^*\left(B - (\contToDiscHilSp^\ast M^{-1}\contToDiscHilSp)\ B\ (\contToDiscHilSp^\ast M^{-1}\contToDiscHilSp ) \right)}_{\boundedLO{\discHilSp}}\\
		&= \norm{\contToDiscBLO^\ast B -  M^{-1}  (\JJ^\ast B) M^{-1}}_{\boundedLO{\discHilSp}}  \\
		&=\norm{(\contToDiscBLO^\ast B -  M^{-1}\JJ^* B) + (M^{-1}\JJ^* B - M^{-1}  (\JJ^\ast B) M^{-1})}_{\boundedLO{\discHilSp}}\\
		&\leq \norm{(\Id -  M^{-1})\JJ^* B\|_{\boundedLO{\discHilSp}} + \|(M^{-1}\JJ^* B (\Id- M^{-1})}_{\boundedLO{\discHilSp}}\\
		&\leq \|\Id -  M^{-1}\|_{\boundedLO{\discHilSp}}\ \|\JJ^* B\|_{\boundedLO{\discHilSp}}\\
		&\qquad\qquad + \|M^{-1}\|_{\boundedLO{\discHilSp}}\ \|\JJ^* B\|_{\boundedLO{\discHilSp}}\  \|\Id- M^{-1}\|_{\boundedLO{\discHilSp}}\\
		&\leq \norm{\Id -  M^{-1}}_{\boundedLO{\discHilSp}}\ \left(1 + \|M^{-1}\|_{\boundedLO{\discHilSp}}\right)\ \|B\|_{\mathcal{B}(L^2(\RR^2)} .
	\end{align*}	
	This implies \cref{claim3}. 
	Recall  that $M^{-1} = \sqrt{G} > 0$ and so  
	$$ 
	\norm{M^{-1}-\Id} \leq \norm{G-\Id}_{\boundedLO{\discHilSp}}
	\norm{( \sqrt{G}+ \Id)^{-1}}_{\boundedLO{\discHilSp}}\le \norm{G-\Id}_{\boundedLO{\discHilSp}}. 
	$$
	By \cref{Gram-prop}, for any $\ve>0$,
	$$ \norm{\Id-\mathfrak{M}^{-1}} \lesssim  \norm{G-\Id}\lesssim \exp(-\frac{1}{16}\lambda((1-\ve)a^2-4r_0^2))\,.$$ 
	Therefore, 
	\[ |\rho|^{q-1} \norm{\Id-\mathfrak{M}^{-1}}\ \lesssim\ \ee^{\frac{\lambda}{4}(1-q)(a^2 + 4\sqrt{|v_{min}|}+\gamma_0)} 
	\times \ee^{-\frac{\lambda}{16}((1-\ve)a^2 -4r_0^2)}\, . \]
	Chose $q$ so that  $\frac{1}{4}(1-q) < \frac{1}{16}(1-\ve)$, that is $1-\frac{1-\ve}{4}<q<1$. Then for any fixed $a$ sufficiently large we have  $ |\rho|^{q-1} \norm{\Id-\mathfrak{M}^{-1}}\ \to0$. Since $|\rho|^{-q}\|\Pi H\Pi\|\to0$ as $\lambda\to\infty$ (\cref{prop:comparing the almost eigenstates to nn hopping}), we have by \cref{term1a} that ${\rm Term}_1\to0$ as $\lambda\to\infty$.
	\medskip
	
	We next bound ${\rm Term}_2$ in \cref{term12}.  By a calculation similar to \cref{mels-HNN}, the matrix elements of the operator $\nnProj^\perp\mathfrak{M}^{-1}\orbitalProjBLO\crysHamil$ in the $\orbitalSetONB$ basis are: 
	$$ (\nnProj^\perp\mathfrak{M}^{-1}\orbitalProjBLO\crysHamil)_{nm} = \ip{\oneAtomGrndStt_n}{\crysHamil\oneAtomGrndStt_m}\left(1-\chi_{_{\rm NN}}(n,m)\right) , $$ 
	where $\chi_{_{\rm NN}}(n,m)$ is given in \cref{chi-nm}.
	\cref{lem:hsy} and self-adjointness implies
	$$ |\rho|^{-1} \norm{\nnProj^\perp\mathfrak{M}^{-1}\orbitalProjBLO\crysHamil}\leq\sup_{n}\sum_{\substack{m\neq n\\ m\nsim n}}\  |\rho|^{-1}\ | \ip{\oneAtomGrndStt_n}{\crysHamil\oneAtomGrndStt_m}|\,. $$
	We consider, for $n$ fixed,
	\begin{equation} \ip{\oneAtomGrndStt_n}{\crysHamil\oneAtomGrndStt_m} = \ip{\oneAtomGrndStt_n}{\oneAtomPotOp_n\oneAtomGrndStt_m} +\sum_{l\neq n,m} \ip{\oneAtomGrndStt_n}{\oneAtomPotOp_l\oneAtomGrndStt_m} \quad m\ne n,\ m\nsim n.
		\label{ip-nm}\end{equation}
	The first term in \cref{ip-nm} is $\rho^\lambda(m-n)$, the  hopping coefficient from site $m$ to site $n$, a distance strictly greater distance than the minimal lattice spacing, $\minLatticeSpacing$. Hence, by  monotonicity property of $\xi\mapsto\nnHopping^\lambda(\xi)$, \cref{thm:monotonicity of rho in displacement} as well as the uniform constraint \cref{min-NNN}, when divided by the nearest neighbor hopping, $\rho^\lambda$, the ratio tends to zero exponentially as $\lambda\to\infty$ .
		For the second term in \cref{ip-nm}, 
		we note that $$ \sum_{l\neq n,m} \ip{\oneAtomGrndStt_n}{\oneAtomPotOp_l\oneAtomGrndStt_m} \leq \wellDepth^2 |v_{min}| \sum_{l\in\setwo{\discSp}{n,m}}|\ip{|\oneAtomGrndStt_n|}{\chi_{B_{r_0}(l)}|\oneAtomGrndStt_m|}|,\ \ m\ne n,\ m\nsim n .$$ 
		Using the bound of \cref{lem:basic lemma to bound inner product of Hamiltonian} we obtain, for any $\alpha\in(0,1)$,  that 
		\begin{equation}
			\sum_{l\neq n,m} \ip{\oneAtomGrndStt_n}{\oneAtomPotOp_l\oneAtomGrndStt_m} \lesssim \frac{\wellDepth^2 |v_{min}|}{1-\alpha}
			\ee^{-\alpha \frac12\wellDepth(a-r_0)^2} 
			\ee^{-\frac{\lambda(1-\alpha)}{8}\normEuc{n-m}^2}\  \ee^{\frac{\lambda}{2}r_0^2} \label{sm-nlm}\end{equation} 
		where $\normEuc{n-m}> a$. 
		We divide the expression in \cref{sm-nlm} by $\rho^\lambda$ and conclude, using  \cref{eq:bound on rho}, that this ratio tends to zero as $\lambda$ tends to infinity if  
		$$ \ee^{-\alpha \frac12\wellDepth(a-r_0)^2} 
		\ee^{-\frac{\lambda(1-\alpha)}{8}a^2}\  \ee^{\frac{\lambda}{2}r_0^2} \ee^{\frac{\lambda}{4}(\minLatticeSpacing^2+4\sqrt{|v_{min}|}\minLatticeSpacing+\gamma_0)	} \to 0\,. $$ 
		For example, this can be arranged by choosing
		$\alpha > 1/3$ and $a>a_0$ sufficiently large so that
		$$ \alpha (a-r_0)^2 
		+\frac{1-\alpha}{4}a^2- r_0^2-\frac{1}{2}(\minLatticeSpacing^2+4\sqrt{|v_{min}|}\minLatticeSpacing+\gamma_0)	 > 0 .$$
		We note in passing that to bound this term we didn't really need $|n-m| > a$ but only $|n-m| \geq a$. That is, the projection outside of nearest-neighbors was only important for the first term in \cref{sm-nlm}.
	\qed
	\medskip
	
	\noindent{\it Proof of \cref{claim2}:} We study  $H^{\mathrm{NN,\lambda}}-\tbH$. By \cref{lem:hsy} we have 
	\[
	\norm{H^{\mathrm{NN,\lambda}}-\tbH} \leq 
	\sup_{m\in\discSp}\sum_{n\in\discSp}\abs*{H^{\mathrm{NN,\lambda}}_{n,m}-\tbH_{n,m}}\,. 
	\]
	We expand $H^{\mathrm{NN},\lambda}_{nm}$, given in \cref{HNNnm}, for fixed $n, m\in\GG$:
	\begin{align*}
		H^{\mathrm{NN},\lambda}_{n,m} &= \rho^{-1} \left\langle \vf_n, H\vf_m\right\rangle \chi_{_{\rm NN}}(n,m)\\
		&=  \rho^{-1} \sum_{l\ne m} \left\langle \vf_n, \lambda^2 v_l\vf_m\right\rangle \chi_{_{\rm NN}}(n,m)\\
		&= \rho^{-1} \sum_{l\ne m} \left\langle \vf_n,  \lambda^2 v_l\vf_m\right\rangle 
		\Big( \Id(n=m) + \Id(n\sim m, |n-m|=a) + \Id(n\sim m, |n-m|>a \Big)\\
		&=  \rho^{-1} \sum_{l\ne m} \left\langle \vf_n,  \lambda^2 v_l\vf_m\right\rangle\ \delta_{n,m}
		\ +\ \rho^{-1} \left\langle \vf_n, v_n\vf_m\right\rangle\ \delta_{|n-m|,a}\\
		&\qquad\qquad +\ \rho^{-1} \sum_{l\ne m} \left\langle \vf_n,  \lambda^2 v_l\vf_n\right\rangle\ \Id\left(n\sim m, |n-m|>a\right)
	\end{align*}
	Thus, 
	\begin{align*}
		H^{\mathrm{NN},\lambda}_{n,m} - \tbH_{n,m} &=  \rho^{-1} \sum_{l\ne m} \left\langle \vf_n,  \lambda^2 v_l\vf_m\right\rangle\ \delta_{n,m}\ +\ \Big(\rho^{-1} \left\langle \vf_n, v_n\vf_m\right\rangle\ - \tbH_{n,m}\Big)
		\delta_{|n-m|,a} \\
		&\quad + \rho^{-1} \sum_{l\ne m} 
		\left\langle \vf_n,  \lambda^2 v_l\vf_n\right\rangle\ \Id\left(n\sim m, |n-m|>a\right)
	\end{align*}

	For fixed $m\in\GG$, we sum over $n\in\GG$ and estimate:
	\begin{align}
		\sum_{n\in\GG} \Big| H^{\mathrm{NN},\lambda}_{n,m} - \tbH_{n,m}\Big| &\le  |\rho|^{-1} \Big| \sum_{l\ne m} \left\langle \vf_m,  \lambda^2 v_l\vf_m\right\rangle\Big|\ +\ 
		\sum_{\substack{n\in\GG\\ |n-m|=a}} \Big|\ \Big(\rho^{-1} \left\langle \vf_n, v_n\vf_m\right\rangle\ - \tbH_{n,m}\Big)\ \Big|\nonumber\\
		&\quad + |\rho|^{-1}\ \sum_{\substack{n\sim m\\ |n-m|>a}}\ \Big|\ \sum_{l\ne m} 
		\left\langle \vf_n,  \lambda^2 v_l\vf_m\right\rangle\ \Big|
		\label{HNN-HTB}\end{align}
	
	By definition of the sequence $\Set{\lambda_j}_j$ in \cref{eq:appropriate beta-sequences}, the second term on the right hand side 
		of \cref{HNN-HTB} vanishes. The supremum over $m\in\GG$ of the first term on the right hand side of \cref{HNN-HTB} tend to zero by \cref{lem:basic lemma to bound inner product of Hamiltonian} (pick any $\alpha>1/2$ and $a$ large enough). Finally, concerning the last term in  \cref{HNN-HTB}, we again use \cref{lem:basic lemma to bound inner product of Hamiltonian} to get \begin{align*} |\rho|^{-1}\ \sum_{\substack{n\sim m\\ |n-m|>a}}\ \Big|\ \sum_{l\ne m} 
			\left\langle \vf_n,  \lambda^2 v_l\vf_m\right\rangle\ \Big| &\lesssim \wellDepth^2|v_{\text{min}}||\rho|^{-1}\ \sum_{\substack{n\sim m\\ |n-m|>a}}\frac{1}{1-\alpha}
			\ee^{-\alpha \frac12\wellDepth(a-r_0)^2} 
			\ee^{-\frac{\lambda(1-\alpha)}{8}\normEuc{n-m}^2}\  \ee^{\frac{\lambda}{2}r_0^2} \\
			&\lesssim \wellDepth^2|\rho|^{-1}
			\ee^{-\alpha \frac12\wellDepth(a-r_0)^2+\frac{\lambda}{2}r_0^2} \sup_{m\in\discSp}\sum_{\substack{n\sim m\\ |n-m|>a}}
			\ee^{-\frac{\lambda(1-\alpha)}{8}\normEuc{n-m}^2}\,. \end{align*} So that the term before the supremum decays to zero we need to choose $\alpha>1/2$ and $a$ sufficiently large. The latter term, which involves a supremum, is bounded above by \begin{align*}\sup_{m\in\discSp}\sum_{\substack{n\sim m\\ |n-m|>a}}
			\ee^{-\frac{\lambda(1-\alpha)}{8}\normEuc{n-m}^2} &\leq \sup_{m\in\discSp}\sum_{n\in\discSp}
			\ee^{-\frac{\lambda(1-\alpha)}{8}\normEuc{n-m}^2}\ \leq C\,.\end{align*}
	by the summability asssumption \cref{eq:Gaussians are summable in discrete space}. This completes the proof of \cref{claim2} and therewith \cref{prop:comparing the almost eigenstates to nn hopping}.
\end{proof}

\section{Decay properties and the overlap integrals}
\label{sec:decay-props}
We derive various consequences of \cref{thm:Gaussian decay of ground state}.
\begin{lem}
	If $\alpha>0$ and $r>r_0$ then $\norm{\chi_{B_r(0)^c}\oneAtomGrndStt^\alpha}\leq C\lambda^{\frac{\alpha-1}{2}} \ee^{-\frac{\alpha}{4} \wellDepth (r^2-r_0^2)}$.\label{lem:integral of atomic ground state away from origin also decays}
\end{lem}
\begin{proof}
	We have, using that $r>r_0$ and 
	\cref{thm:Gaussian decay of ground state}:
	\begin{align*}
		\norm{\chi_{B_r(0)^c}\oneAtomGrndStt^\alpha}^2 &= \int_{|x|\ge r} \oneAtomGrndStt(x)^{2\alpha}\dif{x} \leq C^{2\alpha}\lambda^\alpha\exp(+\frac{\alpha}{2}\lambda r_0^2)\int_{|x|\ge r} \exp(-\frac{\alpha}{2}\lambda\normEuc{x}^2)\dif{x}\\
		&= \frac{1}{\alpha}C^{2\alpha}\lambda^{\alpha-1}\exp(-\frac{\alpha}{2}\lambda (r^2-r_0^2))\,.
	\end{align*}
	
\end{proof}

\begin{lem}\label{lem:orbitals away from atomic centers}
	Let $\Theta: \contSp\to[0,1]$ be such that $\supp(1-\Theta)\subseteq\bigcap_{n\in\discSp}B_{2r_0}(n)^c$.
	
	Then $$ \|(\Id-\Theta(X))\orbitalProj\| \leq C \exp(-c\lambda) $$ for some $C,c\in(0,\infty)$ independent of $\lambda$. 
	\begin{proof}
		For any $n\in\GG$,  the support of $\oneAtomGrndStt_n$ is concentrated within $B_{\oneAtomPotSuppRad}(n)$ for all $n\in\discSp$ by \cref{thm:Gaussian decay of ground state}, so we work with the set $\{\oneAtomGrndStt_n\}$ rather than $\{\oneAtomGrndSttONB_n\}$.
		
		We have (with the notation $\Theta^c:=1-\Theta$):
		\begin{align*}
			\norm{\Theta^c\orbitalProj} & \equiv  \sup\left(\Set{\norm{\Theta^cu}|u\in\orbitalSubSp,\ \ \norm{u}=1}\right)\,.
		\end{align*}
		For  $u\in\orbitalSubSp$,  $u=\sum_{n\in\GG}\alpha_{n}\oneAtomGrndStt_{n}$ where
		$\Set{\alpha_{n}}_{n}\subseteq\CC$. 
		We shall prove $\|\Theta^cu\|^2\lesssim \ee^{-\frac12 r_0^2}\|u\|^2$. We calculate
		\begin{align*}
			\norm{\Theta^cu}^{2} & =  \norm{\Theta^c\sum_{n\in\GG}\alpha_{n}\oneAtomGrndStt_n}^{2}\\
			& =  \left\langle \Theta^c\sum_{n\in\GG}\alpha_{n}\oneAtomGrndStt_n,\Theta^c\sum_{m\in\GG}\alpha_{m}\oneAtomGrndStt_m\right\rangle \\
			& =  \sum_{n\in\GG}\sum_{m\in\GG}\alpha_{m}\overline{\alpha_{n}}\left\langle \Theta^c\oneAtomGrndStt_n,\Theta^c\oneAtomGrndStt_m\right\rangle \,.
		\end{align*}
		Next we note that
		\begin{align}
			\left|\left\langle \Theta^c\oneAtomGrndStt_{n},\Theta^c\oneAtomGrndStt_{m}\right\rangle \right| & \leq  \int_{x\in\mathbb{R}^{2}}\left|\Theta^c\left(x\right)\right|^{2}\left|\oneAtomGrndStt_{n}\left(x\right)\right|\left|\oneAtomGrndStt_{m}\left(x\right)\right|\dif{x}\,.
			\label{Tcnm} \end{align}
		By assumption on $\Theta^c$, it has support \emph{outside} of $\bigcup_{n\in\GG}B_{r_0}(n)$ and so we may bound the two eigenfunctions via \cref{thm:Gaussian decay of ground state} since the integration happens away from their centers (indeed, away from \emph{all} centers). In fact the closest $\Theta^c$ gets to any of the centers is $2r_0$ by assumption, so that in the integration, $\normEuc{x-n}^2 \geq 4r_0^2$ and the same for $m$. So 
		\begin{align}
			\left|\oneAtomGrndStt_{n}\left(x\right)\right|\left|\oneAtomGrndStt_{m}\left(x\right)\right| &\leq C^2 \wellDepth \exp(-\frac14\wellDepth(\normEuc{x-n}^2+\normEuc{x-m}^2-2r_0^2)) \nonumber\\ 
			&\leq C^2 \wellDepth \exp(+\frac12\wellDepth r_0^2)\ \exp(-\wellDepth r_0^2)\  \exp\left(-\frac18\wellDepth(\normEuc{x-n}^2+\normEuc{x-m}^2)\right) \nonumber\\
			&= C^2 \wellDepth \exp(-\frac12\wellDepth r_0^2)\ \exp\left(-\frac14\wellDepth\normEuc{x-\frac{n+m}{2}}^2\right)\ \exp\left(-\frac{1}{16}\wellDepth\normEuc{n-m}^2\right)\,.
			\label{3exp}\end{align}
		Substitution of \cref{3exp} into the integral \cref{Tcnm} and performing the integral over $x$
		we obtain:
		$$ \left|\left\langle \Theta^c\oneAtomGrndStt_{n},\Theta^c\oneAtomGrndStt_{m}\right\rangle \right| \leq 4 \pi C^2 \exp(-\frac12\wellDepth r_0^2)\exp(-\frac{1}{16}\wellDepth\normEuc{n-m}^2)\,. $$
		Hence, by the Cauchy-Schwarz inequality applied to the summation over $n,m\in\GG$:
		\begin{align*}
			\norm{\Theta^cu}^{2} & \leq  4\pi C^2 \ee^{-\frac12\wellDepth r_0^2}\  \sum_{n,m\in\GG}|\alpha_{m}||\alpha_{n}|
			\ee^{-\frac{1}{16}\wellDepth\normEuc{n-m}^2}\\
			& \leq 4\pi C^2 \ee^{-\frac12\wellDepth r_0^2}\ \left(\sum_{n,m\in\GG}|\alpha_{m}|^2
			\ee^{-\frac{1}{16}\wellDepth\normEuc{n-m}^2}\right)^{1/2}\left(\sum_{n,m\in\GG}|\alpha_{n}|^2 \ee^{-\frac{1}{16}\wellDepth\normEuc{n-m}^2}\right)^{1/2}\\
			&\leq 4\pi C^2 \ee^{-\frac12\wellDepth r_0^2}\ \sum_{n,m\in\GG}|\alpha_{m}|^2 \ee^{-\frac{1}{16}\wellDepth\normEuc{n-m}^2}.
		\end{align*}
		We bound the sum as follows:
		\begin{align*}
			\sum_{n,m\in\GG}|\alpha_{m}|^2 \ee^{-\frac{1}{16}\wellDepth\normEuc{n-m}^2}
			&= \sum_{m\in\GG}\Big( |\alpha_{m}|^2\cdot \left[\sum_{n\in\GG}\ee^{-\frac{1}{16}\wellDepth\normEuc{n-m}^2}
			\right]\Big)
			\\ & \leq\ \sum_{m\in\GG}|\alpha_{m}|^2\cdot \sup_{m\in\GG}\Big[\sum_{n\in\GG} \ee^{-\frac{1}{16}\wellDepth\normEuc{n-m}^2}\Big]\leq C \norm{\alpha}^2 \, , 
		\end{align*}
		where we have used  \cref{eq:Gaussians are summable in discrete space}, the Gaussian summability lemma on the discrete set $\GG$.
	Combining the previous two bounds, and applying \cref{lem:Relating the little l^2 norm with the L^2 norm} yields
	\begin{align*} \norm{\Theta^cu}^{2} &\leq 4\pi C \exp(-\frac12\wellDepth r_0^2)\norm{\alpha}^2 \leq 4\pi C \exp(-\frac12\wellDepth r_0^2)\norm{M}^2\norm{u}^2 \ .
	\end{align*}
	This completes the proof of \cref{lem:orbitals away from atomic centers}.
\end{proof}
\end{lem}

\begin{lem}
If $a > \sqrt{8} r_0$ then $$ \norm{\calP_j \Pi} < C \wellDepth^{3/2} $$ where $\calP \equiv P-b A X$ and $j=1,2$.
\label{lem:magnetic sobolev bound}
\end{lem}
\begin{proof}
As above, we write for any $u\in\orbitalSubSp$, $u=\sum_{n}\alpha_{n}\oneAtomGrndStt_{n}$ for
some $\Set{\alpha_{n}}_{n}\subseteq\CC$. Then,  
\begin{align*}
	\norm{\calP_j u}^{2} & =  \norm{\calP_j\sum_{n}\alpha_{n}\oneAtomGrndStt_n}^{2} =  \left\langle \calP_j\sum_{n}\alpha_{n}\oneAtomGrndStt_n,\calP_j\sum_{m}\alpha_{m}\oneAtomGrndStt_m\right\rangle \\
	& =  
	\sum_{n}\sum_{m}\alpha_{m}\overline{\alpha_{n}}
	\left\langle \calP_j \oneAtomGrndStt_n,\calP_j \oneAtomGrndStt_m\right\rangle \\
	&\leq \sum_{n}\sum_{m}|\alpha_{m}|\ |\alpha_{n}|
	\left|\left\langle  \oneAtomGrndStt_n,\calP^2 \oneAtomGrndStt_m\right\rangle\right|,\qquad \left(\calP_j^2\le\calP^2\right) \\
	&= \sum_{n}\sum_{m}|\alpha_{m}| |\alpha_{n}| 
	|\left\langle  \oneAtomGrndStt_n,(e_0^\lambda\Id-\lambda^2 v_m) \oneAtomGrndStt_m\right\rangle|\,.
	\tag{$h_m\vf_m = 0$}
\end{align*}
By orthogonality and the Gramian bound  \cref{eq:locality of Gramian}
we have
\begin{align*} 
	|\left\langle  \oneAtomGrndStt_n,(e\Id-v_m) \oneAtomGrndStt_m\right\rangle| &\leq \Big|
	\int_{x\in\RR^2} \overline{\vf_n(x)}\  \left(e_0^\lambda-\wellDepth^2v(x-m) \vf_m(x)\right)\dif{x}\Big| \\
	&\leq C\wellDepth^2 \delta_{n,m}+(1-\delta_{n,m})C\ \lambda^{1\over 2}\  \ee^{-\frac{\lambda}{16}(|m-n|^2-(2r_0)^2)})\,.
\end{align*}
Since $\normEuc{n-m}\geq a$ for all $n\neq m \in \GG$, we may write
$$ \ee^{-\frac{\lambda}{16}(|m-n|^2-(2r_0)^2)} \leq \ee^{-\frac{1}{4}\wellDepth (\frac{1}{8}a^2-r_0^2)}\  \ee^{-\frac{1}{32}\wellDepth \normEuc{n-m}^2} .$$ 
After inserting the previous two bounds into the above bound on $\|\calP_j u\|^2$, the proof is completed as in the conclusion of the  proof of \cref{lem:orbitals away from atomic centers} and we obtain 
\begin{align*}\norm{\calP u}^{2} & \leq C\lambda^2 \norm{\alpha}^2 + C \lambda ^{3/2} \exp(-\frac{1}{4}\wellDepth (\frac{1}{8}a^2-r_0^2))\norm{\alpha}^2 \leq C \lambda^{2} \norm{M}\norm{u}^2\,.
\end{align*}

This concludes the proof of \cref{lem:magnetic sobolev bound}.
\end{proof}

\section{Topological indices for the integer quantum Hall effect}\label{topology}
In this paper we have focused on the archetypal example of the integer quantum Hall effect (IQHE). As stated, we believe that, for any given cell of the Kitaev table, our approach provides a general program for identifying  topological indices of the corresponding discrete and continuum models. 
Indeed, central to our proof of topological equivalence between continuum and discrete systems is the convergence of the continuum resolvent to the discrete resolvent in the operator norm. Via the analytic functional calculus, this implies that the difference of analytic functions of operators of the discrete and (scaled) continuum systems converge. Since in principle all topological invariants may be expressed as Fredholm indices built on the \emph{smooth} functional calculus of the Hamiltonian (in the presence of a spectral gap), see \cite{Grossmann2016,Katsura_Koma_2018_doi:10.1063/1.5026964}, and using the stability of the Fredholm index (so as to approximate smooth functions by analytic ones) the mechanism for identifying topological invariants that we exploit here is rather general.

Let us now focus on the IQHE. We consider a 2D system of non-interacting electrons (however accounting for the Pauli exclusion principle) immersed in a constant perpendicular strong magnetic field. Here, the macroscopic quantity which exhibits quantization is the transversal Hall conductance.
We next give brief presentation of mathematical expressions  for this quantity in both the bulk and the edge settings; see \cite{Graf_2007,PSB_2016} for an overview of mathematical aspects of the IQHE. We start with a bulk structure which is an insulator
 and later obtain an edge structure from a bulk structure via truncation to half-space.
\subsection{The Hall conductance for bulk systems} \label{Hall-bulk}
A bulk Hamiltonian $H$ models an insulator if it has a gap in its spectrum about the Fermi energy $\mu\in\RR$ (see \cref{rem:the mobility gap} below for a generalization of this definition). Central to the definition of the bulk topological index is the  {\it Fermi projection}, $$ P_{\mu} := \chi_{(-\infty,\mu)}(H),$$
where $\chi_A$ denotes the indicator function of the set $A$. In the language of second quantization, $P_{\mu}$ is the many-body ground state of a non-interacting system of Fermions.
The quantized physical observable is the transverse Hall conductance, $\sigma_{\mathrm{Hall}}$, which measures the electrical conductance in the $2$ direction as a result of an electric field applied in the $1$ direction. 
The seminal experimental discovery \cite{vonKlitzing_1980} is that the transverse Hall conductance is quantized.

This conductance may be computed as a linear response via the Kubo formula, treating the electric field as a small perturbation:
\begin{align} \sigma_{\mathrm{Hall}} := -\ii \tr P_\mu [\partial_1 P_\mu, \partial_2 P_\mu] \label{eq:kubo formula}\end{align}  
and it was shown first in \cite{TKNN_1982}
that in the periodic setting,
$$2\pi \sigma_{\mathrm{Hall}} \in \ZZ\,. $$

There are various ways to interpret \cref{eq:kubo formula}. For example, if our Hamiltonian $H$ is a local discrete operator, say acting on $\ell^2(\ZZ^2)\otimes\CC^N$, and if we further assume translation invariance, we may Bloch decompose it to obtain fibered multiplication operators, $$k\in\TT^2\mapsto H(k)\in\mathrm{Herm}_{N\times N}(\CC),$$ 
the space of $N\times N$ Hermitian matrices. Actually the spectral gap assumption translates to the fact that $\mu$ never touches any of the eigenvalues of $H(k)$ for all $k\in\TT^2$, i.e., $H$ takes values in a restriction of $\mathrm{Herm}_{N\times N}(\CC)$ to matrices for which the $j$th and $(j+1)$th eigenvalues are non-degenerate, where $j\in\Set{1,\dots,N-1}$ labels the band below $\mu$. This space is homeomorphic to $\mathrm{Gr}_{j}(\CC^N)$, the Grassmannian manifold, i.e., the manifold of $j$-dimensional vector subspaces of $\CC^N$. The trace in \cref{eq:kubo formula} then includes an integral on $\TT^2$ and $\partial_j$ is the derivative with respect to the two possible momentum directions. We note that $\TT^2\ni k\mapsto H(k)\in\mathrm{Mat}_{N\times N}(\CC)$ is a smooth function due to the  locality of $H$ (i.e. it has sufficiently rapid off-diagonal decay of its matrix elements in the position basis). In this setting this formula reduces to the well-known Berry-curvature integral familiar to physicists \cite{TKNN_1982}. One can further identify this number as the Chern number associated with the $\TT^2$-vector bundle induced by the family of projections $k\mapsto P_\mu(k)$, $k\in\TT^2$. In this way, $2\pi\sigma_{\mathrm{Hall}}$ emerges as an integer-valued topologically stable number, partially explaining the experimental observation that $2\pi\sigma_{\rm Hall}\in\ZZ$. 

More generally, \cref{eq:kubo formula} provides a route to quantization even if one drops the assumption of translation invariance. Indeed, let $X_j$ denote the position operator in direction $j=1,2$ and introduce a {\it switch function} $\Lambda:\RR\to[0,1]$ \cite{Elbau_Graf_2002}, which transitions between $0$ at $-\infty$ and $1$ at $+\infty$. One may take $\Lambda$ as linear near the origin and asymptotically constant far from it, so $\Lambda(X_j)$ should really be thought of as a regularization of $X_j$. The trace in \cref{eq:kubo formula} is now performed in real-space, and  $\partial_j A$ is the so-called non-commutative derivative \cite{Bellissard_1994} of the operator $A$, given by
$$\partial_j A := -\ii [\Lambda(X_j),A]\,.$$
Below, we shall use the abbreviated notation: $\Lambda_j=\Lambda(X_j)$. In this generality, the topological nature of \cref{eq:kubo formula}  was first established in the ground-breaking work of Bellissard {\it et. al.} \cite{Bellissard_1994} who  proved the  index formula: 
\begin{align} 2\pi \sigma_{\mathrm{Hall}} = \findex (P_\mu U P_\mu + P_\mu^\perp) \,,\label{eq:Hall conductivity as Fredholm index}\end{align} 
where $\findex$ is the Fredholm index of an operator. Here, $U \equiv \exp(\ii\arg(X_1+\ii X_2))$ models a {\it flux-insertion} at the origin; see also \cite{ASS1994_charge_def}.

We recall that if $F$ is a bounded linear operator on a Hilbert space, then it is a Fredholm operator iff its kernel and cokernel are finite dimensional, in which case its index is given by the formula $$ \findex F \equiv \dim \ker F - \dim \ker F^\ast \in \ZZ\,. $$ The index of a Fredholm operator is stable under compact perturbations and is also a continuous mapping $$\mathrm{index}:\mathcal{B}(\mathcal{H})\to\ZZ\,,$$ with respect to the operator norm topology. Hence, it is locally constant; see, for example, \cite{Booss_Topology_and_Analysis}. 

Since $2\pi \sigma_{\mathrm{Hall}}$ has been identified as the index of a Fredholm operator, it is a topological invariant; it is integer-valued, and stable under operator norm continuous and compact perturbations of $P_\mu U P_\mu + P_\mu^\perp$.
We remark that the question of the precise class of   perturbations of $H$, which induces Fredholm-preserving perturbations of $P_\mu U P_\mu + P_\mu^\perp$ is not completely resolved in the greatest level of generality. See \cite{Shapiro2020} for results in this direction.
\begin{rem}\label{rem:the mobility gap}
More generally (and crucially to explain the plateaus of the IQHE), an insulator may be modelled without a spectral gap. Indeed, in the presence of \emph{strong disorder}, the spectral gap begins to fill with energies corresponding to localized states and hence the Fermi energy is no longer isolated from the spectrum. Nonetheless, it is immersed in what is known as a \emph{mobility gap} \cite{EGS_2005} since these bulk states around it cannot contribute to mobility, i.e., to electric conductance. While
the bulk formulas (but not the edge) we use remain valid  when there is a mobility gap  \cite{Aizenman_Graf_1998}, the analysis of this paper only applies when there is spectral gap; indeed, the present method only controls the resolvent of the tight-binding Hamiltonian when the spectral parameter is away from the spectrum.
\end{rem}
\subsection{The Hall conductance for edge systems}
\label{edge-hall}
We discuss now the expression for the Hall edge-conductivity, 
 $\hat{\sigma}_{\mathrm{Hall}}$.

Starting with a bulk Hamiltonian, $H$, which is insulating (has a gap) in a certain energy range, we may obtain an edge Hamiltonian $\hat{H}$ by truncation to a half-space.

Suppose that the Fermi energy, $\mu$, is in a spectral gap of $H$. In many  systems, the operation of truncation has the consequence that $\hat{H}$, a non-compact perturbation of $H$, has states which are localized transverse
 to the boundary but extended in the direction parallel to  the boundary--edge states. Sometimes, the energies associated to these states fill the bulk spectral gap (which includes the energy $\mu$). The {\it edge Hall conductance} is the Hall conductance associated with these edge states.
 
 Let $g:\RR\to[0,1]$ be a smooth version of $\chi_{(-\infty,\mu)}$ such that $\supp(g')$ is contained entirely within the spectral gap of $H$ which contains $\mu$. Hence, $g(\hat{H})$ is smooth version of the Fermi projection $\chi_{(-\infty,\mu)}(\hat{H})$ which differs from it only within the bulk gap (this smooth version is needed for technical reasons).
 
 Now the edge Hall conductance is computed, informally speaking, as the quantum mechanical expectation value of the velocity operator in the $1$ direction (say, the direction parallel to the boundary if the truncation restricts space to the upper half plane) on the edge states: \begin{align} \hat{\sigma}_{\mathrm{Hall}} := \tr g'(\hat{H})\left(-\ii [\hat{H},\Lambda_1]\right) \,.\label{eq:edge hall conductance trace formula} \end{align} Indeed, quantum-mechanically, for any operator $A$, the commutator with the Hamiltonian yields the time derivative as in Heisenberg's picture, i.e., $(-\ii [\hat{H},A])\sim \dot{A}$ and $\Lambda_1$ is a form of regularization of the position operator along the $1$-axis, so that, $(-\ii [\hat{H},\Lambda_1]) \sim \dot{X}_1$, i.e., velocity in the $1$ direction. Furthermore, since $g$ changes only within the bulk gap, $g'$ is a smoothening of a projection onto edge states corresponding to that gap.
 
 As in the bulk case, the most immediate topological interpretation of this formula comes in the translation invariant  setting, in which case it is seen as the spectral flow of edge states in the bulk gap \cite{Halperin_1982_PhysRevB.25.2185}, and is hence a stable integer. More generally, assuming only a spectral gap, a powerful index theorem has been proven both in the discrete \cite[Theorem 3.1]{SBKR_2000} and in the continuum \cite[Theorem 3]{KELLENDONK2004388} setting exhibiting $2\pi \hat{\sigma}_{\mathrm{Hall}}$ as the index of some Fredholm operator, i.e., \begin{align}2\pi \hat{\sigma}_{\mathrm{Hall}} = \findex (\Lambda_1 \exp(-2\pi\ii g(\hat{H}))\Lambda_1+\Lambda_1^\perp)\label{eq:edge Hall conductivity as a Fredholm index}\end{align} and so we see it is also a stable integer. We emphasize this formula holds only in the spectral gap regime, see \cite{EGS_2005} for a generalized formula which is however not manifestly integer-valued a-priori. For edge geometries, we take \cref{eq:edge Hall conductivity as a Fredholm index} as our starting point.
 \begin{rem}
    We note in passing that \cref{eq:edge Hall conductivity as a Fredholm index} has been proven for continuum systems \cite{KELLENDONK2004388} in a slightly different setting than the one we have here: there, the Hilbert space was $L^2(\RR\times [0,\infty))$ with Dirichlet boundary conditions. Here, on the other hand the edge Hilbert space is $\contHilSp$ and it is merely the atomic crystal $\GG$ that gets truncated. 
    We have no doubt however that \cref{eq:edge Hall conductivity as a Fredholm index} holds also in our setting by generalizing \cite[Theorem 3]{KELLENDONK2004388}.
 \end{rem}

\begingroup
\let\itshape\upshape
\printbibliography
\endgroup
\end{document}